\def\confversion{0}
\def\ifconf{\ifnum\confversion=1}
\def\ifnotconf{\ifnum\confversion=0}

\documentclass[10 pt, reqno]{article}

\def\showauthornotes{0}
\def\showkeys{0}
\def\showdraftbox{1}

\usepackage{xspace,xcolor,enumerate}
\usepackage{amsmath,amssymb}
\usepackage{amsthm}
\usepackage[toc,page]{appendix}
\usepackage{comment}
\usepackage{thmtools}
\usepackage{thm-restate}
\usepackage{color,graphicx}
\usepackage{boxedminipage}
\usepackage{makecell}
\usepackage{tabularx}
\usepackage[shortlabels]{enumitem}
\usepackage[linesnumbered,ruled,vlined]{algorithm2e}
\ifnum\showkeys=1
\usepackage[color]{showkeys}
\fi

\definecolor{darkred}{rgb}{0.5,0,0}
\definecolor{darkgreen}{rgb}{0,0.35,0}
\definecolor{darkblue}{rgb}{0,0,0.55}

\usepackage[pdfstartview=FitH,pdfpagemode=UseNone,colorlinks,linkcolor=darkblue,filecolor=darkred,citecolor=darkgreen,urlcolor=darkred,pagebackref]{hyperref}

\usepackage[capitalise,nameinlink]{cleveref}
\usepackage[T1]{fontenc}
\usepackage{mathtools,dsfont,bbm}

\ifnum\showauthornotes=1
\newcommand{\nnote}[1]{{\sf\small\color{orange}{[Nikhil: #1]}}}
\else
\newcommand{\nnote}[1]{}
\newcommand{\Authornote}[2]{}
\newcommand{\Authorcomment}[2]{}
\newcommand{\Authorfnote}[2]{}
\fi

\ifnum\showdraftbox=1

\else

\fi

\newtheorem{theorem}{Theorem}[section]

\newtheorem{observation}[theorem]{Observation}
\newtheorem{definition}[theorem]{Definition}

\newtheorem{lemma}[theorem]{Lemma}
\newtheorem{remark}[theorem]{Remark}
\newtheorem{proposition}[theorem]{Proposition}
\newtheorem{corollary}[theorem]{Corollary}
\newtheorem{claim}[theorem]{Claim}
\newtheorem{fact}[theorem]{Fact}

\theoremstyle{remark}
\newtheorem{algo}[theorem]{Algorithm}

\def\FullBox{\hbox{\vrule width 6pt height 6pt depth 0pt}}

\def\qedsketch{\ifmmode\Box\else{\unskip\nobreak\hfil
\penalty50\hskip1em\null\nobreak\hfil$\Box$
\parfillskip=0pt\finalhyphendemerits=0\endgraf}\fi}

\def\eps{\varepsilon}
\def\epsilon{\varepsilon}
\def\e{\epsilon}
\def\eps{\epsilon}

\def\cal{\mathcal}

\def\implies{\Rightarrow}

\newcommand{\cout}{\mathcal{C}_{\textnormal{out}}}
\newcommand{\cin}{\mathcal{C}_{\textnormal{in}}}
\newcommand{\cael}{\mathcal{C}_{\textnormal{AEL}}}
\newcommand{\dout}{\delta_{\textnormal{out}}}
\newcommand{\din}{\delta_{\textnormal{in}}}
\newcommand{\rout}{R_{\textnormal{out}}}
\newcommand{\rin}{R_{\textnormal{in}}}
\newcommand{\rael}{R_{\textnormal{AEL}}}
\newcommand{\sigout}{{\Sigma_{\textnormal{out}}}}
\newcommand{\sigin}{\Sigma_{\textnormal{in}}}

\newcommand{\Fq}{\ensuremath{\mathbb{F}_q}}
\newcommand{\Fqz}{\ensuremath{\mathbb{F}_{q_0}}}

\newcommand{\inabs}[1]{\left|#1\right|}
\newcommand{\inset}[1]{\left\{#1\right\}}
\newcommand{\inabset}[1]{\left|\{#1\}\right|}
\newcommand{\inbrak}[1]{\left(#1\right)}

\newcommand{\D}{\Delta}
\newcommand{\Dp}{\Delta}
\newcommand{\G}{\Gamma}
\newcommand{\type}{\mathsf{type}}
\newcommand{\ty}{\emph{type} }

\newcommand{\RIM}{\mathsf{RIM}}
\newcommand{\wt}{\mathsf{wt}}

\newcommand{\mM}{\mathbf{M}}
\newcommand{\mG}{\mathbf{G}}

\newcommand{\pdeg}{\ensuremath{\mathsf{\Phi}}}
\newcommand{\dist}{\ensuremath{\mathsf{Dist}}}

\newcommand{\imp}{\ensuremath{\mathsf{imp}}}
\newcommand{\rspn}{\ensuremath{\mathsf{row}\text{-}\mathsf{span}}}
\newcommand{\proj}{\ensuremath{\mathsf{proj}}}

\newcommand{\fl}{\ensuremath{\mathsf{fl}}}

\newcommand{\0}{\ensuremath{\mathbf{0}}}

\newcommand{\dd}{}
\newcommand{\db}{\ensuremath{D_\beta}}

\DeclareMathOperator{\Rob}{Rob}

\newcommand{\N}{{\mathbb{N}}}

\newcommand{\F}{{\mathbb F}}

\newcommand{\cC}{\mathcal{C}}
\newcommand{\cH}{\mathcal{H}}

\newcommand{\cE}{\mathcal{E}}
\newcommand{\cL}{\mathcal{L}}
\newcommand{\cM}{\mathcal{M}}

\newcommand{\cG}{\mathcal{G}}
\newcommand{\cP}{\mathcal{P}}

\newcommand{\cV}{\mathcal{V}}

\newcommand{\delzero}{\delta_{\textnormal{in}}}
\newcommand{\delone}{\delta_{\textnormal{out}}}

\usepackage{nicefrac}
\makeatletter

\def\ProbabilityRender#1#2{
  \@ifnextchar\bgroup%
  {\renderwithdist{#1}{#2}}
   {\singlervrender{#1}{#2}}
}
\def\singlervrender#1#2{%
   \ensuremath{\mathchoice
       {{#1}\left[ #2 \right]}
       {{#1}[ #2 ]}
       {{#1}[ #2 ]}
       {{#1}[ #2 ]}
   }
}
\def\renderwithdist#1#2#3{%
   \@ifnextchar\bgroup
   {\superfancyrender{#1}{#2}{#3}}
   {\ensuremath{\mathchoice
      {\underset{#2}{#1}\left[ #3 \right]}
      {{#1}_{#2}[ #3 ]}
      {{#1}_{#2}[ #3 ]}
      {{#1}_{#2}[ #3 ]}
     }
   }
}
\def\superfancyrender#1#2#3#4#5{
   \ensuremath{\mathchoice
      {\underset{#1}{{#1}}\left#4 #3 \right#5}
      {{#1}_{#2}#4 #3 #5}
      {{#1}_{#2}#4 #3 #5}
      {{#1}_{#2}#4 #3 #5}
   }
}
\makeatother

\newfont{\inhead}{eufm10 scaled\magstep1}

\newcommand{\poly}{{\mathrm{poly}}}

\DeclareMathOperator{\rank}{\operatorname {rank}}

\DeclareMathOperator*{\argmax}{\arg\!\max}

\DeclareSymbolFont{extraup}{U}{zavm}{m}{n}
\DeclareMathSymbol{\varheart}{\mathalpha}{extraup}{86}
\DeclareMathSymbol{\vardiamond}{\mathalpha}{extraup}{87}

\usepackage{tikz}

\usepackage{tikz}
\usetikzlibrary{shapes.symbols}
\usepackage{float}
\usepackage{graphicx}
\usepackage{algpseudocode}
\algrenewcommand{\algorithmicrequire}{\textbf{Input:}}
\algrenewcommand{\algorithmicensure}{\textbf{Output:}}

\usepackage[top=1in, bottom=1in, left=1.10in, right=1.10in]{geometry}

\begin{document}
\sloppy

\title{Probabilistic Guarantees to Explicit Constructions:\\ Local Properties of Linear Codes}

\author{
Fernando Granha Jeronimo\thanks{{\tt University of Illinois, Urbana-Champaign}. {\tt granha@illinois.edu}. }
\and
Nikhil Shagrithaya\thanks{{\tt University of Michigan, Ann Arbor}. {\tt nshagri@umich.edu}. }
}



\maketitle
\thispagestyle{empty}

\begin{abstract}
  We present a general framework for derandomizing random linear codes with respect to a broad class of properties, known as \emph{local properties}, which encompass several standard notions such as distance, list-decoding, list-recovery, and perfect hashing.  
  Our approach extends the classical Alon--Edmonds--Luby (AEL) construction through a modified formalism of local coordinate-wise linear (LCL) properties, introduced by Levi, Mosheiff, and Shagrithaya (2025). 
  The main theorem demonstrates that if random linear codes satisfy the complement of an LCL property $\cP$ with high probability, then one can construct explicit codes satisfying the complement of $\cP$ as well, with an enlarged yet constant alphabet size.
  This gives the first explicit constructions for list recovery, as well as special cases (e.g., list recovery with erasures, zero-error list recovery, perfect hash matrices), with parameters matching those of random linear codes.
  More broadly, our constructions realize the full range of parameters associated with these properties at the same level of optimality as in the random setting, thereby offering a systematic pathway from probabilistic guarantees to explicit codes that attain them. Furthermore, our derandomization of random linear codes also admits efficient (list) decoding via recently developed expander-based decoders.
\end{abstract}

\newpage

\ifnotconf
\pagenumbering{roman}
\tableofcontents
\clearpage
\fi

\pagenumbering{arabic}
\setcounter{page}{1}

\section{Introduction}
Error-correcting codes play an important role in numerous areas~\cite{GRS23}. In addition to their more immediate applications to protect data against errors in transmission and storage, they have found various uses in diverse fields such as complexity theory~\cite{NW94}, pseudorandomness~\cite{Vadhan12}, and cryptography~\cite{GoldreichL89,YZ24}. The quest for codes approaching optimal parameter trade-offs for a given property has been a central theme of coding theory. This is typically a two-fold quest. First, one needs to understand what are the optimal parameter trade-offs. This step is often established via an existential proof, commonly using randomness and the probabilistic method. This can be far from trivial in some cases and is established via innovative randomized techniques. Secondly, one proceeds to search for an explicit construction approaching the ideal parameter trade-offs. There are many reasons why an explicit construction is desirable or needed in an application over, say, a randomized one. For instance, certifying the minimum distance and decoding can be computationally hard, making it unfit for use. One then faces the following natural question.

\begin{center}
  How to explicitly construct codes having proved an existential result?
\end{center}

A common challenge is that an existential proof may shed little to no light on how to explicitly construct such codes, and it may take decades, new ideas, and great ingenuity for the discovery of a corresponding explicit construction. The gulf separating existential and explicit code construction results has been prevalent throughout the history of coding theory. 
Shannon’s seminal work~\cite{S48} established the existence of capacity-achieving codes through random constructions, but it was only many decades later that Arıkan introduced his breakthrough explicit construction of polar codes~\cite{Arikan08}. In the case of list decoding, the capacity theorem of Zyablov and Pinsker~\cite{ZP82} was followed, decades later, by the explicit construction of Guruswami and Rudra~\cite{GuruswamiR06} (which was inspired by~\cite{PV05}). Recently, the seemingly stronger notion of list-decoding capacity was shown to indeed imply\footnote{Under mild assumptions.} Shannon's capacity on symmetric channels~\cite{PSW25}. Likewise, from the classical existential Gilbert–Varshamov bound~\cite{G52,V57}, it took many decades until Ta-Shma~\cite{Ta-Shma17} obtained an explicit construction of binary (balanced) codes with near-optimal parameters. Despite the gulf between several existential and explicit results, we can try to remain hopeful and ask the ambitious question,

\begin{center}
   Is there a general procedure to convert existential code constructions into explicit ones?
\end{center}

Here, we show that this is indeed possible for a vast range of properties of random linear codes known as local properties, which include list decoding, list recovery, perfect hashing, and average pair-wise distance, among many others.
Random linear codes are widely used as a powerful yardstick for understanding parameter trade-offs of codes, often achieving the best possible trade-offs for many tasks.
Our results provide a framework to convert existential guarantees on local properties into explicit ones, while preserving all the parameters attained by random linear codes, at the cost of increased alphabet size.

\paragraph{Local-to-Global Phenomena.} A popular paradigm seen in several pseudorandom constructions is that of a local-to-global transfer of properties.
Generally, this involves the coupling of two entities.
The first is a constant-sized object possessing the property we desire, whose existence is guaranteed by probabilistic arguments and obtained via a brute-force search.
The second is an infinite family of objects (typically expander graphs) whose construction is known through previous results.
The novelty of such constructions often lies in identifying the appropriate manner of integrating the two objects so that the resulting construction inherits the desired property from the first object.
One of the earliest constructions employing such techniques is the work of Tanner~\cite{Tan81}, which details a construction of a family of error-correcting codes that involves integrating several copies of a single constant-sized code (having good rate-distance tradeoffs) with a bipartite graph having large girth.
Sipser and Spielman~\cite{SS96} modified this construction by substituting the large girth graph for an expander graph, and using its expansion properties to prove a lower bound on the distance of the final code.

In recent years, there has been a resurgence of constructions utilizing the local-to-global phenomenon.
Examples include quantum LDPC codes~\cite{PK22}, locally testable codes~\cite{DELL22, PK22}, unique neighbor expanders~\cite{AD24, Che25, HMMP24}, lossless vertex expanders~\cite{Gol24, HLMO25, HLMR25}, and erasure code ensembles~\cite{CCS25}.
Our framework builds upon the Alon–Edmonds–Luby (AEL) construction, a classical instance of the local-to-global paradigm. We provide a few details about this construction.

\paragraph{Alon-Edmonds-Luby (AEL) Construction.} The Alon-Edmonds-Luby (AEL) construction was first introduced by Alon, Edmonds, and Luby in \cite{AEL95} to construct codes with constant alphabet size that approached the Singleton bound.
The construction has three components: a constant-sized inner code having good minimum distance, found by a brute-force search, an explicit outer code having a sub-optimal rate-distance tradeoff, and a bipartite spectral expander graph.
The construction can be described in two steps: the outer code is first concatenated with the inner one, upon which the symbols of codewords from the concatenated code are permuted, in a manner prescribed by the expander, to produce codewords in the final code.
The expansion properties of the underlying graph are used to ``lift'' the minimum distance property of the inner code onto the final code.

Over the years, the AEL construction has been adapted and applied in numerous subsequent works. 
A recurring paradigm in these constructions is to employ a constant-sized inner code with strong parameters, 
combined with an outer code that may have weaker parameters, but is fully explicit.
We use the term \emph{AEL procedure} to refer to such constructions henceforth.
In \cite{GI02}, Guruswami and Indyk gave explicit, linear time encodable and decodable codes for unique decoding that approached the Singleton bound, by utilizing the AEL procedure.
In \cite{KMRS17}, Kopparty, Meir, Ron-Zewi, and Saraf, utilized it for constructions of locally testable codes and locally correctable codes.
It was also leveraged in the work of Kopparty, Ron-Zewi, Saraf, and Wootters~\cite{KRZSW23} to provide list-recoverable and list-decodable codes that matched the parameters achieved by Folded Reed-Solomon codes, while having constant alphabet size.
Very recently, it was utilized by Jeronimo, Mittal, Srivastava, and Tulsiani in \cite{JMST25} to give constructions of list-decodable codes that approached the Generalized Singleton bound over constant size alphabets.

\subsection{Our Results}

\paragraph{Local Properties.}
A local property, in the context of codes, is a property for which the existence of a constant number of codewords suffices as a ``witness'' to the code satisfying that property.
For example, the complement of $(\rho, L)$-list-decodability is a local property, as a set of $L+1$ codewords within a Hamming ball of relative radius $\rho$ serves as a witness for any code possessing the property.
For a locality parameter $L$ independent of the block length, a local property $\cP$ can be informally defined by a collection of (pairwise distinct) vector sets of size $L$.
A code is said to satisfy $\cP$ if it contains all vectors in a vector set from the collection corresponding to $\cP$.
Typically, our objective is to understand codes that satisfy the complement of local properties, that is, codes that avoid containing any vector set from the collection defined by $\cP$.
For instance, a $(\rho, L)$-list-decodable code must avoid containing all pairwise distinct vector sets of size $L+1$ that lie entirely within a Hamming ball of relative radius $\rho$.

The concept of local properties for codes first originated in the work of Mosheiff, Resch, Ron-Zewi, Silas, and Wootters  \cite{MRRSW20}, where they introduced the framework with the purpose of proving the existence of LDPC codes achieving list-decoding capacity.
The existence follows from a more general result; the first step consists of proving a threshold result for local properties achieved by random linear codes, followed by the establishment of a transfer type result, which states that random LDPC codes achieve the same parameters for all local properties as random linear codes.
The threshold result states that every local property has a threshold rate, above which random linear codes satisfy the property with exponentially high probability, and below which they do not.
The framework was employed by Guruswami, Li, Mosheiff, Resch, Silas, and Wootters in \cite{GLMRSW22} to provide lower bounds for list sizes for list-decoding and list-recovery, and also by Guruswami and Mosheiff in \cite{GM22} to prove that punctured low-bias codes achieve the same parameters as random linear codes, with respect to local properties.
Later, Guruswami, Mosheiff, Resch, Silas, and Wootters~\cite{GMRSW22} gave a local properties framework for random codes as well.

While providing a powerful framework to investigate list-decoding and list-recovery in the low alphabet regime, the precise formulation of local properties in \cite{MRRSW20} did not allow one to study local properties in the large alphabet regime.
There are two important random (linear) code families in this regime: random linear codes whose alphabet size is a large constant that is independent of the block length (but may depend on other parameters, such as gap to capacity), and random Reed-Solomon codes (whose alphabet size is at least the block length).
As a consequence, these code families could not be analyzed in the context of local properties.
Addressing this limitation required a new formulation tailored to the large alphabet regime, which was developed by Levi, Mosheiff, and Shagrithaya in \cite{LMS25}.
In this work, the authors establish a threshold result for local coordinate-wise linear (LCL) properties of large alphabet random linear codes, and leverage it to establish an equivalence between random Reed–Solomon codes and random linear codes, with respect to LCL properties.

As noted previously, we seek to understand codes that satisfy the complement of local properties.  
One approach to do so is through explicit constructions.  
Our main result demonstrates that it is possible to explicitly construct codes satisfying the complement of LCL properties, as long as the properties meet certain requirements.

\begin{theorem}[Informal, see \cref{cor:the-big-one}]\label{thm:inf-the-big-one}
    For any  LCL property $\cP$, there exists a suitable (linear) inner code, a bipartite expander, and an outer code such that the AEL procedure, when instantiated with these components, yields an explicit linear code $\cael$ that does not satisfy $\cP$, and whose rate is arbitrarily close to the threshold rate.
\end{theorem}

A few remarks about the result are in order.
\begin{remark}[Optimality, LDPC property]
    We observe that our explicit codes achieve parity with random linear codes in terms of all parameters associated with the LCL property.
    Moreover, our explicit constructions are LDPC codes: meaning that each row of the parity check matrix only has constant number of non-zero entries.
    This property has previously been useful for designing very fast encoding and decoding algorithms.
\end{remark}

\begin{remark}[Combining LCL properties]\label{rem:comine-lcl-prop}
    The nature of LCL properties enables one to ``combine'' a list of LCL properties $\cP_1, \cP_2,\ldots$ to create a single LCL property $\cP$.
    Upon applying \cref{thm:inf-the-big-one} to $\cP$, we obtain an explicit construction that does not satisfy $\cP_1, \cP_2,\ldots$ simultaneously.
    We refer the reader to \cref{cor:combine-lcl-prop} for more details.
\end{remark}

\begin{remark}[Tradeoffs]
    Two important tradeoffs arise in our constructions: alphabet size and the underlying field of linearity.
    \begin{itemize}
        \item Increased alphabet size: Our constructions require an increased alphabet size, with the increase being exponential in ${q_0}^{\poly(L/\eps)}$, where $q_0$ is the size of the field whose elements are used to specify the constraints of the property, $L$ is the locality, and $\eps$ is the gap between the rate of the explicit construction and the threshold rate.
        This phenomenon is characteristic of several constructions based on the AEL procedure in the literature, and our setting is no exception.
        \item Linearity over a subfield: While our codes are defined over a larger field, they remain linear only with respect to a subfield---namely, the field over which the inner code is defined.
    This limitation, however, does not pose difficulties for most applications.
    \end{itemize}

\end{remark}

\begin{remark}[Local Properties of Random Reed-Solomon Codes]
    \cite{LMS25} also proved that random Reed-Solomon codes and random linear codes are equivalent with respect to LCL properties: that is, they have the same threshold rates for all LCL properties.
    This implies that our constructions also match the parameters attained by random Reed-Solomon codes for all LCL properties, with the additional property of having constant alphabet size. 
\end{remark}

\paragraph{List Decoding, List Recovery.}
We turn to discuss two important local properties studied in the literature: list-decoding and list-recovery.
A code is $(\rho, L)$-list-decodable if for every vector $y$, the number of codewords that have (relative) Hamming distance less than $\rho$ from $y$ is at most $L$.
A code $\cC \subseteq \Sigma^n$ is $(\rho, \ell, L)$-list-recoverable if for input lists $S_1,\ldots,S_n$ satisfying $\inabs{S_i} \le \ell$ for all $i \in [n]$, we have that the output list size $L$ is at most
\[
    \inabset{c \in \cC \mid \inabset{i \in [n] \mid c[i] \in S_i} \ge (1-\rho)n} \le L.
\]
Clearly, $(\rho, 1, L)$-list recoverability is equivalent to $(\rho, L)$-list-decodability. One can think of these notions as generalizations of the notion of minimum distance, which requires every pair of distinct codewords to be far from one another.

List-decodable and list-recoverable codes have found uses in numerous areas of theoretical computer science, including pseudorandomness \cite{Tre99, GUV09, LP20}, compressed sensing \cite{NPR12}, and algorithms \cite{LNNT19, DW22}.
Their broad applicability has motivated the development of several explicit constructions spanning a wide range of parameter regimes.
For example, list-recoverable codes have been used in constructions of list-decodable codes \cite{GI03, GR06, KRZSW23}, and locally decodable codes \cite{HRW20}.
For a comprehensive overview of applications of list-recoverable codes, we refer the reader to the recent survey by Resch and Venkitesh~\cite{RV25}.

On the other side of the coin, there has been a significant line of work investigating existential properties of linear codes through the probabilistic method.
The linear structure of such codes means that the codewords of a random linear code are not mutually independent.
This dependence introduced substantial obstacles in analyzing their list sizes for list decoding.
For instance, the probabilistic argument of Zyablov and Pinsker~\cite{ZP82} provided list size upper bounds of $2^{O(1/\eps)}$ and $O(1/\eps)$ for random linear codes and random codes respectively, where $\eps$ is the gap to capacity.
The exponential gap in list sizes was closed  by Guruswami, Håstad, Sudan, and Zuckerman~\cite{GHSZ02}, where they showed that random linear codes indeed achieve a list size of $O(1/\eps)$, by means of a clever potential method argument.
Their result only held in expectation, however, and not with very high probability; this was subsequently resolved by Guruswami, Håstad, and Kopparty~\cite{GHK11}.

In the large alphabet regime, the works \cite{Sud97, GS98} proved that full length Reed-Solomon codes are decodable upto the Johnson bound.
However, it is known that the Johnson bound is not optimal, and an exciting line of work \cite{ST20, GLS24, BGM23, GZ23, AGL24} showed that randomly punctured Reed-Solomon codes approached the Generalized Singleton Bound \cite{ST20}, which is a tight bound on the radius $\rho$ of list-decodable codes having rate $R$ and list size $L$.
The bound proves that
\[
    \rho \le \frac{L}{L+1}(1-R).
\]
In \cite{AGL24}, Alrabiah, Guruswami, and Li also proved that large alphabet random linear codes (with alphabet size $2^{O(1/\eps^2)}$, where $\eps$ is the gap to capacity) approached the Generalized Singleton bound.

In the case of list-recovery, the list-recovery capacity theorem (see \cite{Res20}, Proposition 2.4.14 for a proof) establishes the existence of codes that are $(\rho, \ell, L)$-list-recoverable, with
\[
    \rho \ge 1-R- \eps,
\]
and $L \le O(\ell/\eps)$, as long as $\inabs{\Sigma} \ge \exp(\Omega(\log \ell / \eps))$.
However, this result is for random codes, and it was unclear whether random linear codes could achieve the same  output list size.
\cite{LMS25} proved that this is not the case: the output list size is lower bounded by $L \ge \ell^{\Omega(R/\eps)}$, and this bound was subsequently shown to hold for all linear codes by Li and Shagrithaya in \cite{LS25}.

A long line of works (e.g.,~\cite{RW18}, \cite{LP20}, \cite{GLS24}, \cite{LS25}) have studied list-recovery in the large alphabet regime.
\cite{GLS24} showed that random Reed-Solomon codes are $(1-R-\eps, \ell, O(\ell/\eps))$-list recoverable codes with rate $\Omega(\eps/(\sqrt{\ell}\log (1/\eps)))$.
In \cite{LS25}, the authors showed for any rate $R$, random linear codes are $(1-R-\eps, \ell, L)$-list recoverable, where $L \le \inbrak{\frac{\ell}{\eps}}^{O\inbrak{\frac{\ell}{\eps}}}$.
Very recently, the work of Brakensiek, Chen, Dhar, and Zhang~\cite{BCDZ25b} improved the upper bound on the output list size to $\inbrak{\frac{\ell}{R+\eps}}^{O(R/\eps)}$.
In their concurrent work~\cite{BCDZ25a}, they were able to transfer this result to explicit constructions of Folded Reed-Solomon codes and univariate Multiplicity codes as well.

\paragraph{Explicit constructions of list-recoverable, list-decodable codes.} We discuss results exhibiting explicit constructions of list-decoding and list-recoverable codes in the large alphabet regime.
Parvaresh and Vardy \cite{PV05} introduced the first family of error-correcting codes that was provably list-decodable beyond the Johnson bound.
This was improved upon by Guruswami and Rudra in \cite{GR06}, where they showed that Folded Reed-Solomon codes achieved list-decoding capacity, with polynomial list size.
Further improvements in the analysis of the list size in a fruitful line of works \cite{KRZSW23, Tam24, Sri25, CZ25} proved that the list size matches the one implied by the Generalized Singleton bound.
For list-recovery, the works of \cite{KRZSW23} and \cite{Tam24} showed that Folded Reed-Solomon codes are $(1-R-\eps, \ell, L)$-list-recoverable with output list sizes upper bounded by $(\ell/\eps)^{O(\ell/\eps)}$ and $(\ell/\eps)^{(\log \ell / \eps)}$, respectively.

In the constant-alphabet regime, existing works on capacity-achieving list-decodable and list-recoverable codes fall into two main categories.
Both approaches employ the AEL procedure, where the inner code is a constant-alphabet list-decodable or list-recoverable code obtained via brute force.
The distinction lies in the choice of the outer code.
The first category, exemplified by \cite{KRZSW23}, employs Folded Reed–Solomon codes, which are known to possess strong list-decoding and list-recovery guarantees.
In contrast, the second category relies on outer codes with significantly weaker parameters, requiring only rate $1-\varepsilon$ and distance $\varepsilon^3$.
As a result, the analysis in the latter case is more involved, but it yields improved bounds on the list sizes.
Examples of works belonging to the second category are \cite{JMST25} and \cite{ST25}.

\begin{table}[h!]
\centering
\renewcommand{\arraystretch}{1.5}
\begin{tabular}{|>{\centering\arraybackslash}m{4.3cm}|
                    >{\centering\arraybackslash}m{1.9cm}| m{3.6cm} | m{3.6cm} |}
 \hline
 Work & Radius & Output List Size & Alphabet Size \\ 
 \hline\hline
 \cite{KRZSW23}~(Theorem 6.7), \cite{Tam24}~(Theorem 4.5) & $(1-R-\eps)$ & $\left(\ell/\eps\right)^{\left(\frac{\ell}{\eps}\right)^2 \cdot \log \left(\frac{\ell}{\eps}\right)}$ & $\ell^{O(1/\eps^4)}$ \\ 
 \hline
  \cite{Tam24} & $(1-R-\eps)$ & $(\ell/\eps)^{O(\log \ell/\eps)}$ & $(n\ell/\eps^2)^{O(\ell/\eps^2)}$ \\ 
 \hline
 \cite{ST25} & $(1-R-\eps)$ & $\exp(\exp((\ell/\eps) \log(\ell/\eps)))$ & $\exp(\exp((\ell/\eps) \log(\ell/\eps)))$ \\ 
 \hline
 \cite{BCDZ25a} & $(1-R-\eps)$ & $(\ell/(R+\eps))^{O(R/\eps)}$ & $(n\ell/\eps^2)^{O(\ell/\eps^2)}$ \\
 \hline
 Our Work & $(1-R-\eps)$ & $(\ell/(R+\eps))^{O(R/\eps)}$ & $\exp\bigl((\ell/\eps)^{(\ell/\eps)^{(\ell/\eps)}}\bigr)$ \\ 
 \hline
\end{tabular}
\caption{Parameters for recent explicit constructions of list-recoverable codes at capacity. The second and fourth entries hold for both Folded Reed-Solomon codes and Univariate Multiplicity codes.}
\label{table:1}
\end{table}

\cref{table:1} above highlights the parameters achieved by recent explicit constructions for list-recoverable codes at capacity.
We also refer to Table 1 in \cite{BCDZ25b} for parameters achieved by recent randomized constructions.
Comparing our result to two works: first, to the construction of \cite{KRZSW23} (Theorem 6.7) combined with Tamo's analysis of the output list size of Folded Reed-Solomon codes~\cite{Tam24} (Theorem 4.5), and second, to \cite{ST25}, we note that our result attains smaller output list sizes at the cost of increased alphabet size.
Compared with~\cite{BCDZ25a}, both results achieve the same output list size bounds, but their alphabet size grows polynomially with the block length, whereas ours is constant.

All results mentioned below follow from \cref{thm:inf-the-big-one} by specializing to the appropriate list recovery variant.
\begin{theorem}[Informal, see \cref{cor:cap-ach-list-rec}, \cref{cor:list-rec-params}]\label{thm:inf-cap-ach-list-rec}
There exist explicit constructions of linear codes of rate $R-2\eps$ that are $(1-R-\eps, \ell, L=L_{R, \eps, \ell})$-list recoverable, where $L_{R, \eps, \ell}$ denotes the smallest output list size attained by random linear codes of rate $R-\eps$ that are list recoverable with radius $(1-R-\eps)$ and input list size $\ell$.
The codes have an alphabet size that is at most $\exp((L/\eps)^{O (L)})$.
\end{theorem}

The nature of our result ensures that the output list sizes of our construction exactly match those attained by random linear codes. Consequently, any improvement establishing a tighter upper bound on the list sizes of random linear codes immediately carries over to our construction. In contrast to \cite{KRZSW23}, which relies on an outer code with strong list-size guarantees---a potential bottleneck for future constructions which use this method---our approach requires no such assumption. Indeed, we only assume that the outer code has rate $1-\eps$ and distance $\eps^3$.

Instantiating the codes in \cref{thm:inf-cap-ach-list-rec} with the upper bound on output list sizes for random linear codes, from the recent result of \cite{BCDZ25b}, we get the following corollary.
\begin{corollary}[Informal, follows from \cref{thm:inf-cap-ach-list-rec}]
    There exist explicit constructions of $(1-R-\eps, \ell, (\ell/(R+\eps))^{O(R/\eps)})$-list recoverable codes with rate $R-2\eps$, and alphabet size at most $\exp\bigl((\ell/\eps)^{(\ell/\eps)^{(\ell/\eps)}}\bigr)$.
\end{corollary}

The generality of our main result also yields explicit constructions for related notions such as zero-error list recovery and erasure list recovery. Zero-error list recovery is a special case of list recovery in which the decoding radius is zero. Such codes have found applications in the design of data structures for the heavy hitters problem~\cite{DW22}. In the case of erasure list recovery, some of the input lists may contain only the blank symbol, and the objective is to minimize the number of codewords that remain consistent with the non-blank input lists.

\begin{theorem}[Informal, see \cref{cor:exp-zero-err-list-rec}]\label{thm:inf-exp-zero-err-list-rec}
    There exist explicit constructions of linear codes of rate $R-2\eps$ that are $(\ell, L=L_{R, \eps, \ell})$-zero error list-recoverable, where $L_{R, \eps, \ell}$ denotes the smallest output list size attained by random linear codes of rate $R-\eps$ that are zero error list-recoverable with input list size $\ell$.
\end{theorem}

\begin{theorem}[Informal, see \cref{cor:exp-eras-list-rec}]\label{thm:inf-exp-eras-list-rec}
    There exist explicit constructions of linear codes of rate $R-2\eps$ that are $(\sigma, \ell, L=L_{R, \sigma, \eps, \ell})$-erasure list-recoverable, where $L_{R, \sigma, \eps, \ell}$ denotes the smallest output list size attained by random linear codes of rate $R-\eps$ that are erasure list-recoverable with erasure fraction $\sigma$, and input list size $\ell$.
\end{theorem}

\paragraph{Perfect Hash Matrices.}
An $(n,m,t)$-perfect hash matrix is defined as an $n \times m$ matrix with the following property: for every set of $t$ columns, there exists at least one row in which the entries of those $t$ columns are all distinct.
Perfect hash matrices were first introduced by \cite{Mel84} in the context of database management, and have since found applications in circuit complexity \cite{NW95} and networking \cite{LPB06}.
Consequently, there has been significant work on the explicit construction of such matrices, including \cite{FKS82, AN96, BW98, Bla00}.
In particular, Blackburn and Wild \cite{BW98} presented constructions of optimal linear perfect hash matrices, which are perfect hash matrices where the columns belong to a vector space.
It is straightforward to observe that the set of codewords of a $(0,t-1,t-1)$-list-recoverable code of block length $n$ and size $m$ is equivalent to the set of columns of an $(n,m,t)$-perfect hash matrix.
Thus, constructing linear perfect hash matrices is equivalent to constructing linear $(0,t-1,t-1)$-list-recoverable codes.

We recover the result of \cite{BW98} by providing an alternate construction, which are close to optimal.
\begin{theorem}[Informal, see \cref{cor:exp-per-hash-mat}]\label{thm:inf-exp-per-hash-mat}
  For an integer $t \ge 2$, there exist explicit constructions of linear $(n, Q^{\inbrak{\frac{1}{t-1}-\eps}n}, t)$-perfect hash matrices where the entries belong to $\F_Q$, where $Q=\exp(t^t)$.
\end{theorem}

We also get a modest improvement in the alphabet size of explicit list-decodable codes approaching the Generalized Singleton Bound.
Previously, \cite{JMST25} constructed explicit codes with the same parameters, except with alphabet size equal to $2^{\poly (L^L/\eps)}$.

\begin{theorem}[Informal, see \cref{cor:list-dec-ael}]\label{thm:inf-list-dec-ael}
    There exist explicit constructions of $\bigl(\frac{L}{L+1}(1-R-\eps), L\bigr)$-list-decodable codes, having rate $R-\eps$ and alphabet size at most $2^{\poly (2^L/\eps)}$.
\end{theorem}
The proof in \cite{JMST25} is based on induction, and leverages the expander property of the graph used in the construction.
In contrast, ours is a proof by contradiction, and relies on the sampling property of the underlying graph, thereby providing an alternative proof of essentially the same result.
Although implied by \cref{thm:inf-the-big-one}, we include a proof of \cref{thm:inf-list-dec-ael} as it more transparently illustrates the ideas behind the proof of the former theorem.

\paragraph{Efficient Decoding.}  Our derandomization of random linear codes from~\cref{thm:inf-the-big-one} has the added benefit of admitting efficient decoding algorithms. More precisely, the recently developed efficient (list) decoding algorithms for AEL using the Sum-of-Squares hierarchy~\cite{JMST25} and regularity lemmas ~\cite{ST25,JS25} can also be used to decode our AEL based  constructions. Those efficient decoding algorithms are possible thanks to the use of expander graphs in AEL rendering the decoding task tractable.
In contrast, a considerable amount of evidence~\cite{DMS03, FM04, BLVW19} seems to point towards the problem of decoding for random linear codes being computationally inefficient.

\paragraph{Concurrent Work.} Around the time of release of an earlier version of this paper, Brakensiek, Chen, Dhar, and Zhang released two papers~\cite{BCDZ25b, BCDZ25a}, independently of our work.
We provide a brief discussion about the connections between our work and theirs.
\begin{itemize}
    \item \cite{BCDZ25b} shows an improved upper bound on the output list sizes achieved by random linear codes, in the context of list-recovery.
    Specifically, they show that for integers $\ell, L$ and $R, \eps \in [0, 1]$, a random linear code with alphabet size $\exp((L\log \ell) /\eps)$ is $(1-R-\eps, \ell, L)$-list-recoverable with high probability, where $L$ satisfies
    \[
        L \le \inbrak{\frac{\ell}{R+\eps}}^{O(R/\eps)}.
    \]
    In the current version of our paper, we have utilized this upper bound to show that our explicit constructions attain the same parameters, with the exception of an increased (yet constant) alphabet size.
    \item Their result was established by using a novel two-way reduction between subspace designable codes and random linear codes, shown in their second paper~\cite{BCDZ25a}.
    The reduction proves that subspace designable codes (which include explicit Folded Reed-Solomon codes and Univariate Multiplicity codes) achieve all LCL properties simultaneously, at the same level of optimality as random linear codes.
    Our constructions also achieve the same result (see \cref{rem:comine-lcl-prop}).
    
    We now compare these two results.
    The instantiation of Folded Reed-Solomon codes and Univariate Multiplicity codes in \cite{BCDZ25a} requires an alphabet size that is at least polynomial in the block length of the code, while our constructions require constant alphabet size.
    Additionally, our codes are LDPC codes, a property useful in designing encoding and decoding algorithms.
    However, our constructions have weaker parameters than~\cite{BCDZ25a} in the regime of very large alphabet size, which is a regime of interest in their paper.
    The constructions of~\cite{BCDZ25a} require alphabet sizes that can be significantly larger than the block length, in order to obtain results on tensor codes, as detailed in the second part of their paper.
\end{itemize}

\subsection{Organization}
In \cref{sec:tech-view}, we provide an overview of the proofs for the list-decoding case (\cref{thm:inf-list-dec-ael}) and the general result (\cref{thm:inf-the-big-one}).
The detailed proofs appear in \cref{sec:warmup} and \cref{sec:ael-loc-prop}, respectively.
\cref{sec:ael-loc-prop} is independent of \cref{sec:warmup}; therefore, readers interested solely in the general theorem's proof may proceed directly to it.
Finally, the results pertaining to list-recovery variants and perfect hash matrices (\cref{thm:inf-cap-ach-list-rec}, \cref{thm:inf-exp-zero-err-list-rec}, \cref{thm:inf-exp-eras-list-rec}, \cref{thm:inf-exp-per-hash-mat}) are presented in \cref{sec:conseq}.

\section{Technical Overview}\label{sec:tech-view}
We begin with an overview of the proof of \cref{thm:inf-list-dec-ael}, which addresses the special case of list decoding. 
We instantiate the AEL procedure with the following three components: an inner code of constant block length, found through brute force, an explicit outer code having high rate and distance $\dout$, where $\dout$ is a constant, and an explicit bipartite expander graph $G$.
We note that $G$ can equivalently be interpreted as a sampler, a viewpoint that has been leveraged in prior works (cf.~\cite{KMRS17, KRZSW23}) to analyze AEL-based constructions.  
This sampling perspective naturally motivates a strategy for the explicit construction of codes with strong list-decoding parameters: perform a brute-force search to identify a constant block-length code with good list-decoding parameters, and then use this code as the inner code in the AEL procedure to obtain $\cael$.
If $\cael$ has pairwise distinct codewords $c_1,\ldots,c_{L+1}$ close to some received word $y$, then the codewords agree with $y$ at a large number of coordinates.
The nature of these agreements can be encoded by an \emph{agreement hypergraph} on $L+1$ vertices that contains $n$ hyperedges, one for each coordinate.
The hyperedge for a coordinate is simply the set of indices of codewords agreeing with $y$ on that coordinate.
It follows that if $c_{1},\ldots,c_{L+1}$ agree with $y$ on many coordinates, then the sum of the sizes of the hyperedges is large.
Let these agreements be described by an agreement hypergraph denoted by $\cH$.
Upon invoking the sampling property of the graph $G$,
it is observed that the precise pattern of these agreements is ``ported over'' to a significant number of vertices on the left. As a result, the projections of codewords $c_1,\ldots,c_{L+1}$ and $y$ onto several left vertices have agreements that closely resemble the agreement pattern described by $\cH$.
Consequently, the codeword projections agree with the projection of $y$ at a large number of coordinates within the inner code.
But as the codeword projections are also codewords of the inner code, and because the inner code has good list-decoding parameters, its codewords do not have a lot of agreements with the local projection of $y$.
Therefore, we have seemingly arrived at a contradiction.

However, there is a flaw with this argument: the guarantee of the inner code applies only when the (inner) codewords are pairwise distinct.
Since it cannot be guaranteed that the codeword projections are pairwise distinct, this argument fails.
In order to overcome this obstacle, we use a concept known as a \emph{weakly-partition-connected hypergraph}.
This object was investigated in the context of list-decoding in \cite{AGL24}, where the authors established that (i) any agreement hypergraph with a sufficiently large hyperedge size sum contains a weakly-partition-connected hypergraph, and (ii) there exist linear codes that contain \emph{no non-trivial} codeword sets satisfying weakly-partition-connected hypergraphs.
By \emph{non-trivial} codeword sets, we mean sets of codewords that are not all equal.

We now search for a linear code of constant block length satisfying the above property, and use it as an inner code.
The analysis now proceeds in a manner analogous to the first approach.
If $L+1$ codewords from $\cael$ have a large number of agreements with some received word $y$, then the agreement hypergraph contains a weakly-partition-connected hypergraph $\cH$.
Consequently, a subset of those codewords satisfy the constraints as set forth by $\cH$.
We then use the sampling property of the graph $G$ to port over information about the nature of the agreements to the left side, with the result that the inner codewords at several vertices on the left satisfy the constraints described by $\cH$.
From the parameters that we eventually set, it is seen that the number of such parts exceeds $n(1-\dout)$ (where $\dout$ is the minimum distance of the outer code).
Consequently, at least one set of inner codewords from such a part must be non-trivial, thereby contradicting our assumption on the inner code.

\begin{figure}
\centering
\begin{tikzpicture}[
    vertex/.style={
        circle, 
        draw=blue!50, 
        fill=blue!20, 
        fill opacity=0.2, 
        minimum size=3mm
    }
]

\definecolor{clr3}{HTML}{00798c};
\definecolor{clr2}{HTML}{d1495b};
\definecolor{clr1}{HTML}{d3942e};

\def\vspacing{1.1}
\def\vspacingrect{1.1}
\def\vertoffsetleftrect{0.40}
\def\vspacingleftrect{1.12}
\def\graphoffset{0.7}
\def\xoffset{3} 
\def\rectoffset{0.7} 

\begin{scope}[rect/.style={
        rectangle, 
        draw=gray, 
        minimum width=0.6cm, 
        minimum height=1.6cm 
    }
]

\end{scope}

\begin{scope}[rect/.style={
        rectangle, 
        draw=gray, 
        minimum width=0.6cm, 
        minimum height=0.08cm 
    }
]

    \node[rect, anchor=north east, draw=clr1, fill=clr1!20] at (0, -\vertoffsetleftrect-1*0.2) {};
    \node[rect, anchor=north east, draw=clr2, fill=clr2!20] at (0, -\vertoffsetleftrect-2*0.2) {};
    \node[rect, anchor=north east, draw=clr3, fill=clr3!20] at (0, -\vertoffsetleftrect-3*0.2) {};
    \node[rect, anchor=north east, draw=clr3, fill=clr3!20] at (0, -\vertoffsetleftrect-4*0.2) {};

    \node[rect, anchor=north east, draw=clr3, fill=clr3!20] at (0, -\vertoffsetleftrect-1*0.2-\vspacingleftrect) {};
    \node[rect, anchor=north east, draw=clr2, fill=clr2!20] at (0, -\vertoffsetleftrect-2*0.2-\vspacingleftrect) {};
    \node[rect, anchor=north east, draw=clr3, fill=clr3!20] at (0, -\vertoffsetleftrect-3*0.2-\vspacingleftrect) {};
    \node[rect, anchor=north east, draw=clr1, fill=clr1!20] at (0, -\vertoffsetleftrect-4*0.2-\vspacingleftrect) {};

    \node[rect, anchor=north east, draw=clr2, fill=clr2!20] at (0, -\vertoffsetleftrect-1*0.2-2*\vspacingleftrect) {};
    \node[rect, anchor=north east, draw=clr2, fill=clr2!20] at (0, -\vertoffsetleftrect-2*0.2-2*\vspacingleftrect) {};
    \node[rect, anchor=north east, draw=clr1, fill=clr1!20] at (0, -\vertoffsetleftrect-3*0.2-2*\vspacingleftrect) {};
    \node[rect, anchor=north east, draw=clr2, fill=clr2!20] at (0, -\vertoffsetleftrect-4*0.2-2*\vspacingleftrect) {};

    \node[rect, anchor=north east, draw=clr2, fill=clr2!20] at (0, -\vertoffsetleftrect-1*0.2-3*\vspacingleftrect) {};
    \node[rect, anchor=north east, draw=clr3, fill=clr3!20] at (0, -\vertoffsetleftrect-2*0.2-3*\vspacingleftrect) {};
    \node[rect, anchor=north east, draw=clr1, fill=clr1!20] at (0, -\vertoffsetleftrect-3*0.2-3*\vspacingleftrect) {};
    \node[rect, anchor=north east, draw=clr3, fill=clr3!20] at (0, -\vertoffsetleftrect-4*0.2-3*\vspacingleftrect) {};

    \node[rect, anchor=north east, draw=clr3, fill=clr3!20] at (0, -\vertoffsetleftrect-1*0.2-4*\vspacingleftrect) {};
    \node[rect, anchor=north east, draw=clr3, fill=clr3!20] at (0, -\vertoffsetleftrect-2*0.2-4*\vspacingleftrect) {};
    \node[rect, anchor=north east, draw=clr1, fill=clr1!20] at (0, -\vertoffsetleftrect-3*0.2-4*\vspacingleftrect) {};
    \node[rect, anchor=north east, draw=clr2, fill=clr2!20] at (0, -\vertoffsetleftrect-4*0.2-4*\vspacingleftrect) {};

    \node[rect, anchor=north east, draw=clr2, fill=clr2!20] at (0, -\vertoffsetleftrect-1*0.2-5*\vspacingleftrect) {};
    \node[rect, anchor=north east, draw=clr3, fill=clr3!20] at (0, -\vertoffsetleftrect-2*0.2-5*\vspacingleftrect) {};
    \node[rect, anchor=north east, draw=clr1, fill=clr1!20] at (0, -\vertoffsetleftrect-3*0.2-5*\vspacingleftrect) {};
    \node[rect, anchor=north east, draw=clr3, fill=clr3!20] at (0, -\vertoffsetleftrect-4*0.2-5*\vspacingleftrect) {};

    \node[rect, anchor=north east, draw=clr3, fill=clr3!20] at (0, -\vertoffsetleftrect-1*0.2-6*\vspacingleftrect) {};
    \node[rect, anchor=north east, draw=clr3, fill=clr3!20] at (0, -\vertoffsetleftrect-2*0.2-6*\vspacingleftrect) {};
    \node[rect, anchor=north east, draw=clr1, fill=clr1!20] at (0, -\vertoffsetleftrect-3*0.2-6*\vspacingleftrect) {};
    \node[rect, anchor=north east, draw=clr3, fill=clr3!20] at (0, -\vertoffsetleftrect-4*0.2-6*\vspacingleftrect) {};

    \node[rect, anchor=north east, draw=clr3, fill=clr3!20] at (0, -\vertoffsetleftrect-1*0.2-7*\vspacingleftrect) {};
    \node[rect, anchor=north east, draw=clr2, fill=clr2!20] at (0, -\vertoffsetleftrect-2*0.2-7*\vspacingleftrect) {};
    \node[rect, anchor=north east, draw=clr1, fill=clr1!20] at (0, -\vertoffsetleftrect-3*0.2-7*\vspacingleftrect) {};
    \node[rect, anchor=north east, draw=clr3, fill=clr3!20] at (0, -\vertoffsetleftrect-4*0.2-7*\vspacingleftrect) {};
\end{scope}

\foreach \i in {1,...,8}
    \node[vertex] (L\i) at (\graphoffset, -\i*\vspacing) {};

\foreach \i in {1,...,8}
    \node[vertex] (R\i) at (\graphoffset+\xoffset, -\i*\vspacing) {};

\draw (L1) -- (R2);
\draw (L1) -- (R5);

\draw (L2) -- (R1);
\draw (L2) -- (R4);
\draw (L2) -- (R6);

\draw (L3) -- (R1);
\draw (L3) -- (R6);

\draw (L4) -- (R2);
\draw (L4) -- (R7);
\draw (L4) -- (R5);

\draw (L5) -- (R6);
\draw (L5) -- (R8);

\draw (L6) -- (R3);
\draw (L6) -- (R7);

\draw (L7) -- (R2);
\draw (L7) -- (R8);

\draw (L8) -- (R7);
\draw (L8) -- (R6);

\begin{scope}[rect/.style={
        rectangle, 
        draw=black, 
        minimum width=0.7cm, 
        minimum height=1.1cm 
    }
]

    \node[rect, anchor=north west, draw=clr3, fill=clr3!20] at (\graphoffset+\xoffset+\rectoffset, 0.55-1*\vspacingrect) {};
    \node[rect, anchor=north west, draw=clr2, fill=clr2!20] at (\graphoffset+\xoffset+\rectoffset, 0.55-2*\vspacingrect) {};
    \node[rect, anchor=north west, draw=clr1, fill=clr1!20] at (\graphoffset+\xoffset+\rectoffset, 0.55-3*\vspacingrect) {};
    \node[rect, anchor=north west, draw=clr3, fill=clr3!20] at (\graphoffset+\xoffset+\rectoffset, 0.55-4*\vspacingrect) {};
    \node[rect, anchor=north west, draw=clr3, fill=clr3!20] at (\graphoffset+\xoffset+\rectoffset, 0.55-5*\vspacingrect) {};
    \node[rect, anchor=north west, draw=clr2, fill=clr2!20] at (\graphoffset+\xoffset+\rectoffset, 0.55-6*\vspacingrect) {};
    \node[rect, anchor=north west, draw=clr3, fill=clr3!20] at (\graphoffset+\xoffset+\rectoffset, 0.55-7*\vspacingrect) {};
    \node[rect, anchor=north west, draw=clr1, fill=clr1!20] at (\graphoffset+\xoffset+\rectoffset, 0.55-8*\vspacingrect) {};
\end{scope}
\end{tikzpicture}
\caption{Several copies of the inner code witness vector sets satisfying the constraints in roughly the same proportion as the vector set on the right. Each type of constraint has a different color.} \label{fig:sampling}
\end{figure}
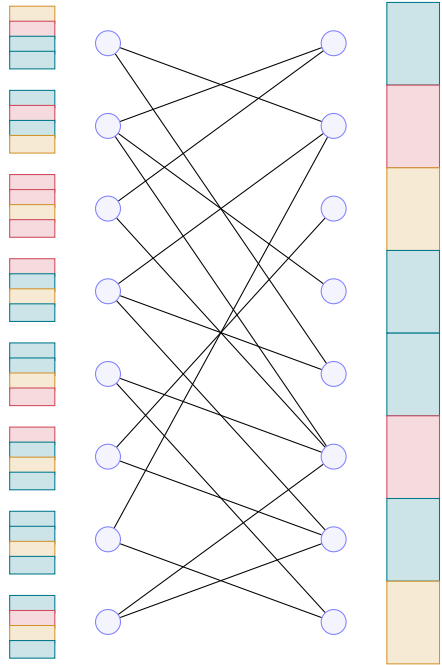

Before giving a proof sketch for the main result (\cref{thm:inf-the-big-one}), we present a high-level overview of LCL properties.
The central idea is that the collection of vector sets corresponding to an LCL property $\cP$ can be defined by specifying a family of linear constraint sets, known as local profiles.
For a block length $n$ and a locality parameter $L$, a local profile specifies a set of linear constraints, on vectors of length $L$, for each coordinate.
The collection of vector sets associated with $\cP$ is then defined as the column set of all matrices of dimension $n \times L$ whose rows satisfy the linear constraints described by at least one local profile associated with $\cP$.
A code is said to satisfy $\cP$ if it contains a vector set from the collection associated with $\cP$.

Every linear constraint set can be written down as a matrix of dimension $L \times L$.
Two linear constraint sets $S$ and $T$ are said to be of the same type if there exists a full rank linear transformation that maps the matrix associated with $S$ to that associated with $T$.
Even though the number of types can be exponential in $L^2$, it is independent of the block length of $\cael$.
This fact ensures that if there are $L$ pairwise distinct codewords in $\cael$ that satisfy the linear constraints described by some local profile $M$ associated with $\cP$, then upon arranging them in an $n \times L$ matrix, a significant fraction of the coordinates can be partitioned according to the constraint type satisfied by the row whose index is equal to the coordinate.
Moreover, the fractional size of each of these sets will be a constant.
Thus, the sampling property of the graph $G$ can be utilized to port over the proportion of these types onto several projections on the left.
The local projections, which are also codewords belonging to the inner code, consequently have codewords that satisfy (a close approximation of) the local profile $M$.
If the AEL procedure were instantiated with an inner code that avoids vector sets satisfying close approximations of the local profiles associated with $\cP$, one might hope to derive a contradiction.  
However, we encounter the same obstacle as before: the guarantee for the inner code applies only when the inner codewords are pairwise distinct, and we lack a mechanism to ensure that this condition is met.

We circumvent this obstacle by employing the concept of \emph{implied types}, first introduced in \cite{MosheiffRRSW19} in the context of local properties.
Implied types were used in their work for the purpose of pinning down the exact threshold rates of local properties of random linear codes in the low alphabet regime.
They also appear implicitly in the work of \cite{LMS25}.
Informally speaking, implied vector sets can be thought of as  ``compressed'' representations of the vector sets corresponding to a local property.
For LCL properties, the corresponding notion is that of implied local profiles.
Consider the implied local profile $I(M)$ corresponding to a local profile $M$.
Denote by $V_{I(M)}, V_M$ the set of matrices associated with $I(M), M$ respectively.
That is, $V_{I(M)}$ (respectively, $V_M$) is the set of matrices whose rows satisfy the constraints specified by $I(M)$ (respectively, $M$).
Then, $V_{I(M)}$ can be obtained by applying an appropriate linear map on the rows of every matrix in $V_M$.
Consequently, we see that if a linear code $\cC$ contains a vector set that satisfies $M$, then by linearity, $\cC$ also contains a vector set that satisfies $I(M)$.

Recall that we are interested in codes that satisfy the complement of an LCL property $\cP$.
We shall prove in \cref{sec:ael-loc-prop} that if random linear codes of rate $R$ satisfy the complement of $\cP$, then random linear codes of rate $R$ also satisfy the following property: with high probability, they do not contain any \emph{non-zero matrices} satisfying any implied local profile $I(M)$, where $M$ is any local profile associated with $\cP$.
Upon instantiating the AEL procedure with an inner code satisfying the aforementioned property, the analysis can be carried out in a manner analogous to that of the naive approach detailed above.
If $L$ pairwise distinct codewords in $\cael$ satisfy a local profile $M$ associated with $\cP$, then upon arranging these codewords in a $n \times L$ matrix and applying the appropriate linear map to each of its rows, the matrix obtained satisfies the constraints set forth by the implied local profile $I(M)$.
Upon porting the constraint types to the left as before, we see that the projected matrices on to several left vertices are codewords belonging to the inner code, and moreover, they satisfy $I(M)$\footnote{In reality, the projected matrices satisfy a close approximation of $I(M)$, instead of $I(M)$. Nevertheless, the inner codes we employ are chosen to be sufficiently robust to accommodate this technical subtlety.}.
The parameters are instantiated in a way so as to ensure that this occurs on more than $(1-\dout)n$ left vertices, and therefore there is one projected matrix that is non-zero.
But because the columns of the projected matrices are codewords of the inner code, and because the inner code avoids containing all non-zero matrices satisfying $I(M)$ for any $M$ associated with the LCL property $\cP$, we successfully arrive at a contradiction.

\section{Preliminaries}\label{sec:prelim}
Let $[n]$ denote the set $\inset{1,\ldots,n}$.
For a vector $x \in (\Sigma^d)^n$ and $i \in [n]$, we denote $x(i) \in \Sigma^d$ as the $i$th entry of $x$.
Furthermore, for $j \in [d]$, we denote $x(i)[j] \in \Sigma$ as the $j$th entry of $x(i)$.
For two vectors $x,y \in \Sigma^n$, the \emph{Hamming distance} is defined as $d(x, y):=\inabset{i \in [n]\colon x(i) \neq y(i)}$, that is, the number of indices at which $x$ and $y$ differ.
For a set $X$, we denote $2^X$ to be the power set of $X$.
For a matrix $G$, we use the notation $G[i][j]$ to index the entries of $G$, and use $G[i][]$ and $G[][j]$ to refer to the $i$th row and $j$th column of $G$, respectively.
For a prime power $q$, let $\F_q$ be the finite field of order $q$.
The notation $x \sim X$ means that the random variable $x$ is being sampled uniformly from the set $X$.

For a vector space $V$, let $\cL(V)$ denote the set of all subspaces of $V$. Furthermore, if $V$ is an $L$-dimensional space, then we define $\cL_{\dist(V)}$ as
\[
    \cL_{\dist}(V) := \inset{ U \in \cL(V) \mid \forall 1 \le i < j \le L, \exists u \in U \text{ such that }u[i] \neq u[j]}.
\]
That is, $\cL_{\dist}(V)$ is the set of all subspaces that for each pair of distinct coordinates, contains at least one vector whose entries differ at those coordinates.
We use $\ker \psi, \ker M$ to denote the kernel of a linear map $\psi$ and the kernel of the linear map given by matrix $M$ in the standard basis, respectively.
The notation $\0$ is used to denote the all zeroes matrix.

\subsection{Error-Correcting Codes}
An error-correcting code $\cC$ over alphabet $\Sigma$ is a subset of $\Sigma^n$ where every pair of distinct vectors in $\cC$ has large Hamming distance.
We denote by $n$ the block length of the code.
We say that a code has \emph{(relative) distance} $\delta$ if:
\[
    \min_{\substack{x, y \in \cC \\x \neq y}} \frac{d(x, y)}{n} \ge \delta.
\]
The \emph{rate} of the code $\cC$, usually denoted by $R$, is defined as:
\[
    R := \frac{\log_{\Sigma}\inabs{\cC}}{n}.
\]
In this paper, we concern ourselves with linear codes.
For a finite field $\F_q$, a linear code $\cC$ is a code that is a linear subspace of $\F_q^n$.
For linear codes, the dimension of the corresponding subspace and the rate are related in the following manner:
\[
    R = \frac{\dim \cC}{n}.
\]

We say that a code $\cC$ is a $[N, \delta, R]_\Sigma$ code if it has block length $N$, distance $\delta$, rate $R$, and is over the alphabet $\Sigma$.
For constants $\rho \in [0, 1]$, $L \in \N$, a code $\cC \subseteq \Sigma^n$ is $(\rho, L)$-list decodable if for every vector $y \in \Sigma^n$, we have
\[
    \inabset{c \in \cC \mid d(c, y) \le \rho n} \le L.
\]

A \emph{random linear code} (RLC) $\cC \subseteq \F_q^n$ of rate $R$ is the kernel of a uniformly random matrix in $\F_q^{(1-R)n \times n}$.

\paragraph{Concatenated Codes.} For integers $N, d$, $d < N$, let $\cout$ be a $[N, \rout, \dout]_\sigout$ code and let $\cin$ be a $[d, \rin, \din]_{\sigin}$ code, satisfying
\begin{equation}\label{eq:concat-eq}
    \inabs{\sigout}=\inabs{\cin} = \inabs{\sigin}^{\rin\cdot d}.
\end{equation}
Note that \cref{eq:concat-eq} allows us to construct an encoding function for $\cin$ that is a bijection from $\sigout$ to $\cin$, denoted by $\phi \colon \sigout \rightarrow \cin$.
Then the concatenated code $\cout \circ \cin$ is a $[Nd, \rout\cdot \rin, \dout\cdot\din]_{\sigin}$ code defined as
\[
    \cout \circ \cin = \inset{v \in \sigin^{([N]\times [d])} \mid \exists~c \in \cout, \bigl( \forall i \in [N], \phi(c[i]) \in \cin \bigr) \land (\phi(c[i]))_{i \in [N]} = v}.
\]
That is, for a codeword $c \in \cout$, we encode $c[i]$ for every $i \in [N]$ with the encoding map $\phi$, thus producing $(\phi(c[i]))_{i \in [N]}$.
This vector is the one obtained by concatenating the codewords of $\cin$ corresponding to each entry of $c$.
We perform this procedure for every codeword of $\cout$, and collect them in the set $\cout \circ \cin$.

It is easy to see that the rate and distance of $\cout \circ \cin$ is equal to $\rout \cdot \rin$ and $\dout \cdot \din$, respectively.

\subsection{Alon-Edmonds-Luby (AEL) Construction}
Let $G=(V_L \cup V_R, E)$ be a bipartite graph satisfying $\inabs{V_L}=\inabs{V_R}=N$, with both vertex sets having degree equal to $d$.
For a vertex $v \in V_L \cup V_R$, denote $\Gamma(v)$ to be the neighborhood of $v$.
For every vertex $v$, we will have an arbitrary, but fixed ordering on the edges incident on $v$.
This allows us to define the $i$th neighbor of $v$: for every $i \in [d]$, $\G_i(v)=w$ if $e=(v, w)$ is the $i$th edge of $v$.
We will also fix an ordering, according to $[N]$, on $V_L$ and $V_R$.
With a slight abuse of notation, we will use $\ell$ (respectively, $r$) to refer to a vertex in $V_L$ (respectively, $V_R$), and also an index in $[N]$.

Observe that the structure of $G$ implies a bijection $\varphi_G \colon (V_L \times [d]) \rightarrow (V_R \times [d])$.
Namely, $\varphi_G(\ell, i)=(r, j)$ for $i, j \in [d], \ell \in V_L$ and $r \in V_R$ if $\Gamma_i(\ell)=r$ and $\Gamma_j(r)=\ell$ both hold.
For our applications, we will be dealing with matrices whose rows are indexed by $(r, j)$ where $r \in [N]$ and $j \in [d]$.
We shall now develop some notation that allows us to ``project'' the rows of these matrices from the right to the left.

\begin{definition}[Projection Operation]\label{def:proj}
    For an integer $L$ and a matrix $A \in \sigin^{([N] \times [d]) \times L}$, define the matrix $A^\proj \in \sigin^{([N] \times [d]) \times L}$ as follows: $\forall \ell \in [N\dd], \forall i \in [d], \forall r \in [N\dd], \forall j \in [d]$, 
    \[
        A^\proj[(\ell, i)][]=A[(r, j)][] \iff \varphi_G(\ell, i) = (r, j).
    \]
\end{definition}
We think of $A$ as a matrix that creates a label (from its rows) for each outgoing edge from $V_R$, and $A^\proj$ as ``collecting'' the labels on the incoming edges into $V_L$. The graph $G$ plays the role of ``shuffling'' these edge labelings.
Thus, $A^\proj$ is created from $A$ by permuting its rows, according to $G$.

\begin{definition}[Projection on Vertices]\label{def:proj-on-vertices}
For $\ell \in [N]$, define the matrix $A^\proj(\ell) \in \sigin^{[d] \times L}$ to be the submatrix of $A^\proj$ satisfying
\[
    \forall j \in [d], (A^\proj(\ell))[j][]=A^\proj[(\ell, j)][].
\]
\end{definition}

\begin{definition}[Flattening Operation]\label{def:flatten}
    For an integer $n$ and a vector $h \in (\sigin^d)^n$, we define $h_\fl \in \sigin^{([n] \times [d])}$ to be the \emph{flattened} vector corresponding to $h$: that is, $h_\fl[(i,j)]=h[i](j)$ for all $i \in [n], j \in [d]$.
\end{definition}

Note that $h_\fl$ can also be viewed as a matrix having just a single column, and therefore \cref{def:proj}
applies to $h_\fl$ as well.
For codes $\cout, \cin$ where $\cout$ is a $[N, \rout, \dout]_\sigout$ code and $\cin$ is a $[d, \rin, \din]_{\sigin}$ code, the code $\cael(\cout, \cin , G) \subseteq \inbrak{\sigin^d}^{N\dd}$ is defined as follows
\[
    \cael(\cout, \cin , G) := \inset{ h \in \inbrak{\sigin^d}^{N\dd} : h^\proj_\fl \in \cout \circ \cin}.
\]

In all our constructions, the underlying graph $G$ is obtained from the following result.
\begin{claim}[Lemma 2.7, \cite{KMRS17}. Also see Claim E.1, \cite{KRZSW23}.]\label{clm:kmrs}
    Let $\beta, \eta, \zeta \in [0, 1]$.
    For infinitely many integers $N$, there is a $d=O(1/\zeta\eta^2)$, so that the following holds.
    There exists a bipartite $\lambda$-spectral expander graph $G=(V_L, V_R, E)$ that can be constructed in time $\poly(N\dd)$, with $N\dd$ vertices on each side, degree $d$ on both sides, and with the following property: for any set $Y \subseteq V_R$ of right-hand vertices with $|Y| = \beta N\dd$, we have 
    \[
        |\{ v \in V_L : |\Gamma(v) \cap Y| < (\beta-\eta)d\}| \leq \zeta N.
    \]
\end{claim}

The expander graph constructed in the preceding claim is a sampler in the sense that, for every sufficiently large subset of right vertices $Y$, a substantial fraction of left vertices have neighborhoods intersecting with $Y$ in a proportion that is roughly equal to the density of $Y$.

\begin{observation}[Explicitness of AEL Procedure, LDPC property]\label{obs:ael-explicit}
    We observe that $\cael(\cout, \cin , G) \subseteq \inbrak{\sigin^d}^{N\dd}$ can be constructed in time $\poly(N\dd)$---that is, in time polynomial in the block length of the code---provided that each of the three components, namely, $\cin$, $\cout$, and $G$, can themselves be constructed in time $\poly(N\dd)$.

    Moreover, $\cael(\cout, \cin , G)$ is an LDPC code if $\cout$ is an LDPC code.
\end{observation}

\section{Warm Up: Construction of List-Decodable Codes}\label{sec:warmup}
In this section, we give a proof for \cref{thm:inf-list-dec-ael}.
We define the required concepts and definitions in \cref{subsec:hypergs-rims}, and state the proof in \cref{subsec:proof-list-dec}.

\subsection{Preliminaries}\label{subsec:hypergs-rims}
A hypergraph $\cH=(V, \cE)$ consists of a vertex set $V$, and a collection $\cE$ of subsets of $V$.
The subsets are known as hyperedges.
For every hyperedge $e$, we define the weight of $e$ as $\wt(e) := \max \inbrak{\inabs{e}-1, 0}$.
Furthermore, the weight of the set of hyperedges is defined as the sum of hyperedge weights.
That is, $\wt(\cE) := \sum_{e \in \cE} wt(e)$.

\begin{definition}[Agreement Hypergraph]
    For an integer $n \in \N$ and a set of vectors $y, c_1,\ldots,c_t \in \Sigma^n$, we define an agreement hypergraph $\cH(y,c_1,\ldots,c_t)=([t], \cE)$ as follows:
    For each $i \in [n]$, construct hyperedge $\e_i \subseteq [t]$ by including the indexes of all vectors that agree with $y$ at the $i$th coordinate.
    That is, $e_i := \inset{j \in [t] \colon c_j[i]=y[i]}$.
\end{definition}

\begin{definition}[Weak Partition Connectivity]
    For an integer $n \in \N$, and $R \in [0, 1]$, we say that a hypergraph $H=(V, \cE)$ is $(R, n)$-weakly-partition-connected if for every partition $\cP$ of the vertex set $V$, the following holds:
    \begin{equation} \label{eq:wpc-cond}
        \sum_{e \in \cE} \max\inset{\inabs{\cP(e)}-1, 0} \ge R n(\inabs{\cP}-1),
    \end{equation}
    where $\inabs{\cP}$ denotes the number of parts in $\cP$ and $\inabs{\cP(e)}$ denotes the number of parts that intersect non-trivially with $e$.
\end{definition}

From Lemma 2.3 in \cite{AGL24}, we see that the agreement hypergraph corresponding to a bad list-decoding configuration contains a weakly-partition-connected sub-hypergraph.
\begin{lemma}[Lemma 2.3, \cite{AGL24}]\label{lem:bad-list-dec-conf}
    Suppose that for vectors $c_1,\ldots,c_{L+1} \in \Sigma^n$, the average Hamming distance of these vectors from a vector $y \in \Sigma^n$ is at most $\frac{L}{L+1}(1-R)n$.
    Then, for some subset $J \subseteq [L+1]$ where $\inabs{J}\ge 2$, the agreement hypergraph corresponding to vectors $y$ and $\inset{c_j \colon j \in J}$ is $(R, n)$-weakly-partition-connected.
\end{lemma}
\begin{proof}
    Let $\cH=([L+1], \cE)$ be the agreement hypergraph corresponding to vectors $y$ and $c_1,\ldots,c_{L+1} \in \Sigma^n$.
    Since
    \begin{equation}\label{eq:hyperedge-wt-large}
        \sum_{i \in [n]} \frac{(L+1)-\inabs{e_i}}{L+1} = \sum_{j \in [L+1]} \frac{ \sum_{i \in [n]} \mathbbm{1}[c_j(i) \neq y(i)]}{L+1} = \sum_{j \in [L+1]} \frac{d(y, c_i)}{L+1} \le \frac{L}{L+1}(1-R)n,
    \end{equation}
    we have
    \[
        \wt(\cE) = \sum_{i \in [n]} \wt(e_i) \ge -n + \sum_{i \in [n]} \inabs{e_i} \ge LRn.
    \]
    The last inequality follows from \cref{eq:hyperedge-wt-large}.

    Let $J \subseteq [L+1]$ be an inclusion minimal subset with $|J| \ge 2$ such that
    \begin{equation}\label{eq:inc-min-subset}
        \sum_{i \in [n]} \wt(e_i \cap J) \ge LRn.
    \end{equation}
    The existence of such a $J$ follows from the fact that $J=[L+1]$ satisfies \cref{eq:inc-min-subset}.
    Let $\cH'=(J, \cE')$ be the hypergraph with vertex set $J$ and edge set $\cE' := \inset{J \cap e \mid e \in \cE}$.

    We now prove that $\cH'$ is $(R, n)$-weakly-partition-connected.
    Observe that \cref{eq:wpc-cond} follows from \cref{eq:inc-min-subset} when $\cP$ is the trivial partition with a single part.
    Now, consider a non-trivial partition $\cP=P_1 \sqcup \ldots \sqcup P_p$.
    We have
    \begin{align*}
        \sum_{e \in \cE'} \max \inset{\inabs{\cP(e)}-1, 0} &= \sum_{\substack{e \in \cE' \\ e \neq \emptyset}} \Bigl( -1 + \sum_{b \in [p]} \mathbbm{1}[\inabs{e \cap P_b} > 0]\Bigr) \\
        &= \sum_{\substack{e \in \cE' \\ e \neq \emptyset}} \Bigl( (|e|-1) - \sum_{b \in [p]} \inbrak{\inabs{e \cap P_b} - \mathbbm{1}[\inabs{e \cap P_b} > 0]} \Bigr) \\
        &= \sum_{\substack{e \in \cE' \\ e \neq \emptyset}} \Bigl( \max (|e|-1, 0) - \sum_{b \in [p]} \max \inbrak{\inabs{e \cap P_b}-1, 0} \Bigr) \\
        &= \sum_{e \in \cE'} \wt(e) - \sum_{b \in [p]} \sum_{e \in \cE'} \wt(e \cap P_b) \\
        &\ge (\inabs{J}-1)Rn-\sum_{b \in [p]} (\inabs{P_b}-1)Rn \\
        &= (p-1)Rn=(\inabs{\cP}-1)Rn,
    \end{align*}
    where the last inequality follows from the inclusion minimality of set $J$, and the fact that every $P_b$ is a strict subset of $J$.
\end{proof}

Lemma 2.14 of \cite{AGL24} proves the robustness of weakly-partition-connected hypergraphs to hyperedge deletions:
\begin{lemma}[Robustness of weakly-partition-connected hypergraphs (Lemma 2.14, \cite{AGL24})]\label{lem:wpc-robust}
    Let $\cH=([t], \cE)$ be a $(R+\eps, n)$-weakly-partition-connected hypergraph. Then for all sets $\cE' \subseteq \cE$, $|\cE'| \le \eps n$, the hypergraph $\cH'=([t], \cE \setminus \cE')$ is $(R, n)$-weakly-partition-connected.
\end{lemma}
\begin{proof}
    Consider any partition $\cP$ of $[t]$. We have
    \begin{align*}
        \sum_{e \in \cE \setminus \cE'} \max \inbrak{\inabs{\cP(e)}-1, 0} &= \sum_{e \in \cE} \inbrak{\inabs{\cP(e)}-1, 0} - \sum_{e \in \cE'}  \inbrak{\inabs{\cP(e)}-1, 0} \\
        &\ge (R+\eps)(\inabs{\cP}-1)-\inabs{\cE'}(\inabs{\cP}-1) \\
        &= Rn(\inabs{\cP}-1).
    \end{align*}
\end{proof}

We now proceed to define another object used in the proofs of \cite{AGL24}: a Reduced Intersection Matrix ($\RIM$).
Although several versions of the $\RIM$ appeared in other works such as \cite{ST20, GLS24, BGM23}, this variant was first introduced in \cite{GZ23}.

In order to define the $\RIM$, we need to set up some notation.
For any integers $k, m$ with $k \le m$ and a finite field $\F_q$, define the symbolic matrix $\cG \in \F_q(X_{1,1},\ldots,X_{k,n})^{k \times n}$ as
\[
    \cG := 
    \begin{bmatrix}
        X_{1,1} & \ldots & X_{1,n}\\
        \vdots  & \ddots & \vdots \\
        X_{1,k} & \ldots & X_{k,n}
    \end{bmatrix}.
\]
The $i$th column of $\cG$ is denoted by $\cG_i=[X_{1,i},\ldots,X_{k, i}]$.
\begin{definition}[Reduced Intersection Matrix]
    The Reduced Intersection Matrix $\RIM_\cH$ associated with a hypergraph $\cH=([t], \cE=(e_1,\ldots,e_n))$ is a $\wt(\cE) \times (t-1)k$ matrix with entries from $\F_q(X_{1,1},\ldots,X_{k,n})^{k \times n}$.
    We construct $\RIM_\cH$ as follows: for every hyperedge $e_i \in \cE$ containing vertices $j_1<j_2<\ldots<j_{|e_i|}$, add $\wt(e_i)=|e_i|-1$ rows for each $\ell=2,\ldots,|e_i|$.
    Each row has $(t-1)$ segments, each of length $k$, of the form $r_{i,\ell}=(r^{(1)}, \ldots, r^{(t-1)})$. The segments are defined as follows:
    \begin{itemize}
        \item If $j=j_1$, then $r^{(j)}=\cG_i^\top=[X_{1,i},\ldots,X_{k,i}]$.
        \item If $j=j_\ell$, and $j_\ell \neq t$, then $r^{(j)}=\cG_i^\top=-[X_{1,i},\ldots,X_{k,i}]$.
        \item Otherwise, $r^{(j)}=0^k$.
    \end{itemize}
\end{definition}

\begin{definition}(Substituted $\RIM$)
    For any matrix $G \in \F_q^{k \times n}$ and a $\RIM_\cH$ associated with the hypergraph $\cH=([t], \cE)$, the \emph{substituted Reduced Intersection Matrix}, denoted by $\RIM_{\cH}(G)$ is defined to be the matrix in $\F_q^{\wt(\cE) \times (t-1)k}$ obtained by substituting every indeterminate symbol in $\RIM_\cH$ with the corresponding entry from $G$.
    That is, we obtain $\RIM_{\cH}(G)$ from $\RIM_\cH$ by replacing, for every $i \in [k], j \in [n]$, $X_{i,j}$ with $G[i][j]$.
\end{definition}

\begin{lemma}[Column Rank of Reduced Intersection Matrices (Lemma 4.2, \cite{AGL24})]\label{lem:col-rank-rim}
    Let $\cH$ be the agreement hypergraph
    For a vector $y \in \F_q^n$, and a matrix $G \in \F_q^{k \times n}$, suppose that the corresponding agreement hypergraph for $y$ and codewords $c_1,\ldots,c_t$ generated by $G$ is equal to $\cH$.
    Moreover, let $c_1,\ldots,c_t$ satisfy the property that they are not all equal.
    Then, the substituted reduced intersection matrix $\RIM_{\cH}(G)$ does not have full column rank.
\end{lemma}
\begin{proof}
    Let $m_1,\ldots,m_t$ be the message vectors for codewords $c_1,\ldots,c_t$.
    That is, $c_i=m_i\cdot G$ for all $i \in [t]$.
    Then, we see that
    \begin{equation}\label{eq:rim-not-full-rank}
        \RIM_{\cH}(G) \cdot \begin{bmatrix}
                                m_1-m_t \\
                                \vdots  \\
                                m_{t-1}-m_t 
                            \end{bmatrix} = 0.
    \end{equation}
    Because $c_1,\ldots,c_t$ were not all equal, the same applies for $m_1,\ldots,m_t$.
    Therefore the vector multiplied to $\RIM_{\cH}(G)$ in \cref{eq:rim-not-full-rank} is non-zero.
    Because the agreement hypergraph corresponding to $y$ and codewords $c_1,\ldots,c_t$ is the same as $\cH$, every constraint in $\RIM_{\cH}(G)$ is satisfied, and hence, $\RIM_{\cH}(G)$ does not have full column rank.
\end{proof}

The following lemma, which is the crux of the main result of \cite{AGL24}, states that for a uniformly random matrix $G \in \F_q^{k \times n}$, the corresponding substituted reduced intersection matrices associated with all weakly partition-connected agreement hypergraphs have full column rank.

\begin{lemma}[Existence of Linear Codes whose RIMs have Full Column Rank (Theorem 1.3, Lemma 4.6, \cite{AGL24})]\label{lem:good-rim-subs}
    For integers $n, L$, where $L$ is a constant independent of $n$, rate $R \in [0, 1]$ and a sufficiently small $\eps>0$, alphabet size $q\geq 2^{10L/\eps}$, with probability at least $1-2^{-Ln}$, a uniformly random matrix $\mG \in \F_q^{Rn \times n}$ has the following property: for every agreement hypergraph $\cH$ on a vertex set of size $\leq L+1$ that is $(R+\eps/2, n)$-weakly-partition-connected, the substituted reduced intersection matrix $\RIM_{\cH}(\mG)$ has full column rank.
\end{lemma}
The proof of \cref{lem:good-rim-subs} follows from the proof of Theorem 1.3 in \cite{AGL24}. We record an important corollary of the lemma.

\begin{corollary}\label{cor:good-rim-subs}
    For integers $n, L$ such that $n>1/L$, rate $R \in [0, 1]$ and a sufficiently small $\eps>0$, there exist $\F_q$-linear codes $\cC \subseteq \F_q^n$, where $q=2^{10L/\eps}$, with the following property: for every agreement hypergraph $\cH$ on a vertex set of size $\leq L+1$ that is $(R+\eps/2, n)$-weakly-partition-connected, the only set of codewords in $\cC$ that satisfy every agreement in $\cH$ is the trivial set of codewords that are all equal.
\end{corollary}
\begin{proof}
    By \cref{lem:good-rim-subs}, we see that with a non-zero probability, a uniformly random matrix $\mG \in \F_q^{Rn \times n}$ has the property of having full column rank on every $\RIM_{\cH}(G)$, for all agreement hypergraphs $\cH$ that are $(R+\eps/2, n)$-weakly-partition-connected.
    Thus, there exists at least one matrix $G \in \F_q^{Rn \times n}$ satisfying the property.
    The contrapositive of \cref{lem:col-rank-rim} then implies that the only set of codewords in $G$ that simultaneously satisfy every agreement in $\cH$ is the set of codewords that are all equal.
    The corollary follows by considering the code $\cC$ generated by the rows of $G$.
\end{proof}

\subsection{Proof for List-Decoding}\label{subsec:proof-list-dec}
Recall that the bipartite graph $G$ from \cref{clm:kmrs} has $N\dd$ vertices on each side.  
For the sake of simplifying the exposition, we omit floor and ceiling notation in the proof.  
In order to utilize the results listed in \cref{subsec:hypergs-rims}, we require the inner code to be defined over a finite field $\F_q$, and so we take $\sigin = \F_q$.
Consequently, we may interpret $\cael \subseteq \inbrak{\F_q^d}^N$ as a code over the extension field $\F_Q$, where $Q = q^d$.

Fix a list size parameter $L$, rate $R \in [0, 1]$, and slack $\eps \in [0, 1]$.
Let $q=2^{10L/\eps}$.
Recall that $\dout$ is defined to be the distance of the outer code $\cout$, and take $\eta=\eps/2^{(L+3)}$ and $\zeta < \dout/2^{(L+1)}$.
Take $d=\max (O(1/\zeta \eta^2), 1/L)$.
We take $\cin \subseteq \F_q^d$ to be a $\F_q$-linear code having rate $\rin = R$.
Let $G$ be the bipartite graph from \cref{clm:kmrs} having $N$ vertices on each side, and every vertex having degree $d \ge O(1/\zeta \eta^2)$.
Let $\cout \subseteq (\F_q^{\rin d})^N$ be the outer code.
Note that by our definition, $\phi$ is now a map of the form $\phi \colon \F_q^{\rin d} \rightarrow \cin$.
We require $\cout$ to be an $\F_q$-linear code and $\phi$ to be a $\F_q$-linear map.

We now prove the following result:
\begin{theorem}\label{thm:list-dec-ael}
    Let $\cin, \cout, G$ be as defined above.
    Furthermore, if $\cin$ is a code satisfying the property in \cref{cor:good-rim-subs}, then $\cael \subseteq \inbrak{\F_q^d}^N$ is an $\F_q$-linear code that is $\inbrak{\frac{L}{L+1}(1-R-\eps), L}$ average-radius list-decodable.
\end{theorem}
Since both $\cout$ and the map $\phi$ are $\F_q$-linear, it follows that $\cael$ is itself an $\F_q$-linear code.

For the sake of contradiction, assume there is a vector $y \in (\F_q^d)^N$ and pairwise distinct codewords $c_1,\ldots,c_{L+1} \in \cael$ such that their average Hamming distance from $y$ is less than $\frac{L}{L+1}(1-R-\eps)N$. That is
\[
    \sum_{i \in [L+1]}\frac{d(y, c_i)}{L+1} < \frac{L}{L+1}(1-R-\eps)N.
\]
Then, \cref{lem:bad-list-dec-conf} implies that there is a subset $J \subseteq [L+1]$, $\inabs{J}\ge 2$ such that the agreement hypergraph corresponding to vectors $y$ and $\inset{c_j \colon j \in J}$ is $(R+\eps, n)$-weakly-partition-connected.
We state a lemma which proves that the local projections of the vectors $y$ and $\inset{c_j \colon j \in J}$ onto many left vertices also yields a weakly-partition-connected hypergraph.
Let $y^\proj:= y^\proj_\fl$ and $c_j^\proj := (c_j)^\proj_\fl$ for every $j \in J$ denote the flattened projections of $y$ and the codewords $\inset{c_j \colon j \in J}$, respectively.
These vectors are obtained by performing a flattening operation (see \cref{def:flatten}), followed by a projection operation (see \cref{def:proj}).

\begin{lemma}[Local Projections are Weakly-Partition-Connected]\label{lem:loc-proj-wpc}
    If the agreement hypergraph corresponding to a vector $y \in (\F_q^d)^N$ and codewords $c_1,\ldots,c_t \in \cael$ is $(R+\eps, N)$-weakly-partition-connected, there exists a set $L^* \subseteq V_L$ satisfying $\inabs{L^*} > (1-\dout)N$ such that for every $\ell \in L^*$, the agreement hypergraph corresponding to the local projections $y^\proj(\ell) \in \F_q^d$, $c_1^\proj(\ell),\ldots,c_t^\proj(\ell) \in \cin$ is $(R+\frac{\eps}{2}, d)$-weakly-partition-connected.
\end{lemma}

We first prove \cref{thm:list-dec-ael} using this lemma.
\begin{proof}[Proof of \cref{thm:list-dec-ael} using \cref{lem:loc-proj-wpc}]
    Fix a pair of codewords $c_i, c_j$, such that $i, j \in J$, and $i, j$ are distinct (such a pair of indices exist since $\inabs{J} \ge 2$).
    Define
    \[
        S:= \inset{\ell \in V_L : c_i^\proj(\ell) \neq c_j^\proj(\ell)}.
    \]
    Because $c_i$ and $c_j$ are distinct codewords belonging to $\cael$, we know that by construction, $|S| \ge \dout N$.
    This is because $c_i^\proj(\ell) \neq c_j^\proj(\ell)$ if and only if $\phi^{-1}(c_i^\proj(\ell)) \neq \phi^{-1}(c_j^\proj(\ell))$, and this holds for exactly those vertices on which the codewords in $\cout$ corresponding to $c_i, c_j$ differ.
    Upon applying \cref{lem:loc-proj-wpc} to the agreement hypergraph corresponding to vectors $y$ and $\inset{c_j \colon j \in J}$, there exists at least one left vertex $\ell \in S \cap L^*$.
    Because $\ell \in S$, $c_i^\proj(\ell) \neq c_j^\proj(\ell)$, and therefore, the local codewords $\inset{c_j^\proj(\ell) : j \in J}$ are not all equal.
    Since $\ell \in L^*$, the agreement hypergraph corresponding to the local projections $y^\proj(\ell) \in \F_q^d$, $c_1^\proj(\ell),\ldots,c_T^\proj(\ell) \in \cin$ is $(R+\frac{\eps}{2}, d)$-weakly-partition-connected.
    This contradicts the property of $\cin$ as described in \cref{cor:good-rim-subs}.
\end{proof}

\begin{proof}[Proof of \cref{lem:loc-proj-wpc}]
    We will associate a subset of $[t]$ with every $r \in V_R$, which we shall refer to as the \ty of $r$.
    We say that $r$ is of \ty $T$ for some subset $T \subseteq [t]$ if the set of indices of all codewords agreeing with $y$ at $r$ is equal to $T$.
    More formally, define $\type \colon V_R \rightarrow 2^{[t]}$, with 
    \[
        \type(r):=\inset{i \in [t] : c_i(r)=y(r)}.
    \]
    For $\beta \in [0, 1]$, we say that a \ty $T \subseteq [t]$ is $\beta$-dense if $\type(r)=T$ for more than $\beta N$ vertices $r \in V_R$.
    We shall only consider types that are $\beta$-dense for $\beta :=\eps/2^{(t+2)}$.
    Define $\db \subseteq 2^{[t]}$ to be the set of all subsets of $[t]$ that are $\beta$-dense.
    Then the number of right vertices whose type is \emph{not} $\beta$-dense is at most 
    \[
        \inabset{r \in V_R : \type(r) \not\in \db} \leq |2^{[t]}\setminus \db|\cdot \beta N \leq \frac{\eps N}{4}.
    \]
    By the robustness of weakly-partitioned-hypergraphs (c.f. \cref{lem:wpc-robust}), we see that the agreement hypergraph created by deleting all hyperedges corresponding to types that are not $\beta$-dense is still $(R+\frac{3\eps}{4}, N)$-weakly-partition-connected.
    Denote this agreement hypergraph by $\cH'$.
    Combining the fact that (i) $\cH'$ is $(R+3\eps/4, N)$-weakly-partition-connected and (ii) all hyperedges in $\cE_{\cH'}$ are associated with vertices whose type belongs to $\db$, we see that for every partition $\cP$ of $[t]$,
    \[
        \sum_{T \in \db} \sum_{\substack{r \in V_R \\ \type(r)=T}} \max\inset{\inabs{\cP(T)}-1, 0} = \sum_{e \in \cE_{\cH'}} \max\inset{\inabs{\cP(e)}-1, 0} \ge \inbrak{R+\frac{3\eps}{4}} N(\inabs{\cP}-1).
    \]
    Denote the set of all vertices of type $T$ by $S_T \subseteq V_R$, and denote its density by $\mu(S_T) := \inabs{S_T}/|V_R|=\inabs{S_T}/N \ge \beta$.
    Then,
    \begin{equation}\label{eq:wpc-global}
        \sum_{T \in \db} \mu(S_T) \cdot \max\inset{\inabs{\cP(T)}-1, 0} \ge \inbrak{R+\frac{3\eps}{4}}(\inabs{\cP}-1).
    \end{equation}
    
    Fix some $\beta$-dense type $T \subseteq [t]$.
    Using the sampling property of the graph $G$, we now show that for a large number of left vertices $\ell \in V_L$, the fraction of edges entering $\ell$ that arise from right vertices of type $T$ is roughly the same as the fraction of type $T$ vertices on the right side.
    Quantitatively, by \cref{clm:kmrs}, we see that
    \[
        \inabset{ \ell \in V_L : |\Gamma(\ell) \cap S_T|/d \leq \mu(S_T)-\eta} \leq \zeta N.
    \]
    By applying a simple union bound argument over all $\beta$-dense types, the following holds
    \[
        \inabset{ \ell \in V_L : \exists T \in \db : |\Gamma(\ell) \cap S_T|/d \leq \mu(S_T)-\eta} \leq \inabs{\db} \zeta N.
    \]
    Thus for at least $(1-\inabs{\db} \zeta)N > (1-\dout)N$ left vertices $\ell \in V_L$,
    \begin{equation}\label{eq:types-prop}
        \forall T \in \db: |\Gamma(\ell) \cap S_T| > (\mu(S_T)-\eta)d.
    \end{equation}
    Denote this set by $L^* \subseteq V_L$.
    For a vertex $\ell \in L^*$ and a type $T \in \db$, observe that for all indices $i \in [d]$ for which $\type(\Gamma_i(\ell))$ belongs to $T$, the local codewords $\inset{c_j^\proj(\ell) : j \in T}$ agree with $y^\proj(\ell)$ at coordinate $i$.
    Therefore, we can speak of types for local coordinates as well, and by a slight abuse of notation, define $\type(i):=\type(\Gamma_i(\ell))$.
    
    Additionally, the proportion of those coordinates is roughly equal to the proportion of right vertices on which the codewords $\inset{c_j : j \in T}$ agree with $y$.
    Thus, we see that the agreements corresponding to all hyperedges that occur on more than $\beta$ fraction of the right vertices are ``ported over'' to the vertices in $L^*$.
    Informally, this says that the agreement hypergraph corresponding to the local projections at every $\ell \in L^*$ is roughly equivalent to $\cH'$.
    We will now prove this in a formal manner, in order to conclude that these local agreement hypergraphs are weakly-partition-connected.

    Turning our attention over to a fixed $\ell \in L^*$ and denoting the set of hyperedges in the local agreement hypergraph of $\ell$ by $\cH_\ell=([t], \cE_\ell)$, we see by \cref{eq:types-prop} that for every partition $\cP$ of $[t]$,
    \begin{align*}
        \sum_{e \in \cE_\ell} \max\inset{\inabs{\cP(e)}-1, 0} &\ge \sum_{T \in \db} \sum_{\substack{i \in [d] \\ \type(i)=T}} \max\inset{\inabs{\cP(T)}-1, 0}\\
        &> \sum_{T \in \db} (\mu(S_T)-\eta)d\cdot \max\inset{\inabs{\cP(T)}-1, 0}.
    \end{align*}
    The last term can be expanded as
    \[
        \sum_{T \in \db} \mu(S_T)d\cdot \max\inset{\inabs{\cP(T)}-1, 0}-\sum_{T \in \db}\eta d\cdot \max\inset{\inabs{\cP(T)}-1, 0}.
    \]
    The first term is at least $(R+3\eps/4)d(|\cP|-1)$ by virtue of \cref{eq:wpc-global}, and the second term is at most $\frac{\eps}{4}d(|\cP|-1)$, as $|\db| \le 2^{L+1}$ and $\eta=\eps/2^{(L+3)}$. Putting everything together, we get
    \[
        \sum_{e \in \cE_\ell} \max\inset{\inabs{\cP(e)}-1, 0} \ge (R+\eps/2)d(|\cP|-1).
    \]
    Thus, $\cH_\ell$ is $(R+\eps/2, d)$-weakly-partition-connected.
\end{proof}

\begin{corollary}\label{cor:list-dec-ael}
    Let $\cin, \cout, G$ be as defined in \cref{thm:list-dec-ael}.
    Furthermore, let $\cout$ be a code with rate $\rout=1-\eps$ and distance $\dout\ge \eps^3$.
    Denote $Q:=q^d$.
    Then $\cael \subseteq \inbrak{\F_Q}^N$ is an $\F_q$-linear LDPC code that is $\inbrak{\frac{L}{L+1}(1-R-\eps), L}$ average-radius list-decodable, with rate $\rael \ge R-\eps$, where $d$ satisfies
    \[
        d \le O(2^{3L}/\eps^5).
    \]
    Thus, $Q = \exp(2^{O(L)}/\eps^5)$.
    Moreover, $\cael$ is constructible in time $\poly(N)$.
\end{corollary}
\begin{proof}
    The $\F_q$-linearity of $\cael$ its list-decodability parameters are proven in \cref{thm:list-dec-ael}.
    The rate is given by
    \[
        \rael=\frac{\log_Q(\inabs{\cael})}{N}.
    \]
    Indeed, it is easy to see that $\inabs{\cael}=\inabs{\cout}=q^{\rin\rout dN}$.
    Because $Q=q^d$, a simple calculation gives $\rael=\rin\cdot \rout=R\cdot (1-\eps)> R-\eps$.
    The value for $d$ is obtained by recalling that $d=\max(O(1/\zeta \eta^2), 1/L)$, $\eta=\eps/(2^{L+3})$, $\zeta=\dout/2^{(L+1)}$, and plugging in the value for $\dout$ in $\zeta$.
    The value for $Q$ is obtained by recalling the fact that $q=2^{10L/\eps}$ from \cref{cor:good-rim-subs}.

    We note that explicit constructions of $\F_q$-linear codes having rate $1-\eps$ and distance $\eps^3$ that are constructible in time $\poly (N)$ can be obtained by using Tanner codes (see Corollary 11.4.8 in \cite{GRS23}).
    We note that Tanner codes are LDPC codes.
    Moreover, the graph $G$ can be constructed in time $\poly(N)$, by \cref{clm:kmrs}.
    By \cref{cor:good-rim-subs}, $\cin$ exists and has a block length independent of $N$, therefore it can be found through brute force in constant time.
    Upon invoking \cref{obs:ael-explicit}, we see that $\cael(\cin, \cout, G)$ is an LDPC code and can be constructed in time $\poly(N\dd)$.
\end{proof}

\section{Constructions for Local Properties}\label{sec:ael-loc-prop}
A \emph{code property} $\cP$ can informally be defined as a family of codes sharing some common characteristics.
We focus on code properties that are
\begin{enumerate}[(i)]
    \item local, and
    \item monotone-increasing.
\end{enumerate}
A local code property, informally speaking, is defined by the inclusion of bad sets of vectors.
A \emph{monotone-increasing} code property is one for which the following is true: if $\cC$ is in $\cP$, then every $\cC'$ for which $\cC' \supseteq \cC$ holds, also lies in $\cP$.
An example of local, monotone-increasing code property is the complement of $(\rho, L)$-list-decodability. 

Before we give a formal definition of local properties studied in this paper, we discuss aspects of similar definitions in previous works.
Local properties were first defined in \cite{MRRSW20} in order to prove threshold type results for random linear codes, and to prove that LDPC codes achieve the same parameters as random linear codes.
Their work focused on proving results for the small alphabet regime, and given their definition of local properties, it was not possible to extend the results to the large alphabet regime.
This was accomplished in \cite{LMS25}, where the authors provided a new definition suitable for the large alphabet regime.

Since our work studies codes in the latter regime, we choose to adapt the definition in \cite{LMS25}, and give a brief overview of the same before discussing our modifications.
For a locality parameter $L \in \N$ and block length $n$, a \emph{$L$-local coordinate wise linear} ($L$-LCL) property $\cP$ is defined as a collection of local profiles.
A local profile is an ordered tuple of subspaces $\cV= (\cV_1, \ldots, \cV_n)$, where $\cV_i \in \cL (\F_q^L)$ for each $i \in [n]$.
A matrix $A \in \F_q^{n \times L}$ is said to be \emph{contained in} $\cV$ if the $i$th row of A belongs to $\cV_i$, for all $i$.
We say that a code satisfies property $\cP$ if there exists a matrix $A \in \F_q^{n \times L}$ such that the columns of $A$ are pairwise distinct codewords in $\cC$, and $A$ is contained in some local profile belonging to the collection of local profiles associated with $\cP$.

We now state our modifications, and the justifications for introducing them.
In a nutshell, our modifications are concerned with reconciling the (seemingly) different definitions of LCL properties for codes having differing field sizes and block lengths.
The modifications are necessary in order to talk about the LCL properties being satisfied by the constant-sized inner code, while also being satisfied by the infinite code family produced by the AEL construction.
Recall that in addition to having differing block lengths, these two codes also have different alphabet (field) sizes.

Our first modification is to define local profiles as tuples of matrices, instead of subspaces.
We then require that in order for a matrix $A$ to be contained in a local profile, each row of $A$ should lie in the kernel of the matrix corresponding to the row index.
This is necessary in order to address the problem of differing field sizes.
Recall that the alphabet of $\cael$ is $\sigin^d$, where $\sigin$ is the alphabet of the inner code.
Thus, if our inner code is over a field $\F_q$, one can view the alphabet of $\cael$ as being equal to $\F_q^d$. Note that one can naturally view the vectors in $\F_q^d$ as elements in the extension field $\F_Q$, where $Q=q^d$.
This fact ensures that the rows of matrices in $\F_q^{L \times L}$, when viewed as linear constraints,
will be applicable to vectors in $\F_q^L$ and $\F_Q^L$ simultaneously, as $\F_Q$ is an extension field of $\F_q$.

Our second modification is to construct local profiles using a list of fractions, each corresponding to a matrix in $\mathbb{F}_q^{L \times L}$ and denoting the fraction of coordinates on which that matrix appears.  
This representation enables us to describe local properties in a manner that is independent of the block length of the codes.  
We remark that this modification is similar in spirit to the definition of local properties in \cite{MRRSW20}, where the forbidden matrices were described by specifying the frequency of each vector from $\mathbb{F}_q^L$ in such matrices.
In the sequel, we refer to this representation as a \emph{local profile description}, and note that each such description defines a collection of local profiles rather than a single one.

\subsection{Preliminaries}
Fix a locality parameter $L \in \N$.
Additionally, consider finite fields $\Fqz, \Fq$ where $q_0, q$ are powers of the same prime that satisfy $q_0 \le q$.
It follows that $\Fq$ is an extension field of $\Fqz$.
Throughout the paper, we fix a $q_0$ and define LCL properties using elements from $\Fqz$, and work with $\F_q$-linear codes over the field $\Fq$.
The precise value of $q$ will be fixed later.
  
Accordingly, the term ``linear'' will henceforth always refer to $\Fq$-linear.
\begin{definition}[Local Profile Description]
    An \emph{$L$-local profile description} is an unordered tuple of tuples of the form
    \[
        \mathcal{V} = \bigl( (f_1, \mM_1), \ldots, (f_T, \mM_T) \bigr),
    \]
    where for each $t \in [T]$, we have $f_t \in [0,1]$, $\sum_{t \in [T]} f_t = 1$, and $\mM_t \in \Fqz^{L \times L}$.
    The matrices $\mM_t$ are not required to be pairwise distinct.
\end{definition}

\begin{definition}[Local Profile]
    Fix a block length $n \in \N$, and a $L$-local profile description $\cV$ such that every $f_t$ in $\cV$ is a multiple of $1/n$.
    We define an \emph{$L$-local profile $M_n(\cV)$ created according to $\cV$} as an ordered tuple of matrices 
    \[
        M_n(\cV) = (M_1, \ldots, M_n),
    \]
    where $M_i \in \Fqz^{L \times L}$ for each $i \in [n]$.
    Moreover, as prescribed by $\cV$, for each $t \in [T]$, the matrix $\mM_t$ appears in exactly $f_t\cdot n$ coordinates in $M_n(\cV)$.
\end{definition}

We emphasize that $M_n(\cV)$ is not unique for a local profile description $\cV$. In fact, it is easily seen that all permutations of the entries of $M_n(\cV)$ are valid local profiles that can be created using $\cV$. We will denote the set of all local profiles that can be created from $\cV$ (for a block length $n$) by $\cV_n$.

For the rest of this subsection, fix a block length $n$, and an $L$-local profile description $\cV=((f_1, \mM_1), \ldots, (f_t, \mM_{T}))$, where every fraction $f_t$ is a multiple of $1/n$.

\begin{definition}[Satisfying Local Profile Descriptions]\label{def:sat-loc-prof-desc}
    For a matrix $A \in \F_q^{n \times L}$, if there exists a local profile $M_n(\cV) = (M_1, \ldots, M_n) \in \cV_n$ such that $A[i][] \in \ker M_i$ for all $i$, then we say that $A$ \emph{satisfies} $\cV$, and that $M_n(\cV)$ is a \emph{witness} for $A$ satisfying $\cV$.
\end{definition}

\begin{definition}[Containing Matrices]\label{def:cont-matr}
We say that a matrix $A$ is \emph{contained} in a code $\cC \subseteq \F_q^n$ if the columns of $A$ are codewords of $\cC$.
Equivalently, we will use the shorthand $A \subseteq \cC$.
\end{definition}

\begin{definition}[Containing Local Profiles]\label{def:con-loc-prof}
A code $\cC \subseteq \F_q^n$ is said to \emph{contain} $\cV$ if there exists a matrix $A \in \F_q^{n \times L}$ such that
\begin{enumerate}
    \item $A \subseteq \cC$,
    \item $A$ satisfies $\cV$, and
    \item $A$ has pairwise distinct columns.
\end{enumerate}
\end{definition}

\begin{definition}[Local Coordinate wise Linear (LCL) Property]\label{def:lcl-prop}
    We define a \emph{$L$-local coordinate wise linear} ($L$-LCL) property $\cP$ to be a set of $L$-local profile descriptions $\cV$.
\end{definition}
By abuse of notation, we will use the term $\cP$ to refer to both the property itself, as well as the set of local profile descriptions that specify it.

\begin{definition}[Satisfying LCL Properties]\label{def:sat-lcl-prop}
    We say that a code $\cC \subseteq \F_q^n$ \emph{satisfies} $\cP$ if there is a $\cV \in \cP$ such that $\cC$ contains $\cV$.
\end{definition}

\begin{definition}[Code contains $(\cV, U)$]
    For a subspace $U \in \cL(\F_q^L)$, we say that a code $\cC \subseteq \F_q^n$ \emph{contains $(\cV, U)$} if there is a matrix $A \in \F_q^{n \times L}$ such that
    \begin{enumerate}
        \item $A \subseteq \cC$,
        \item $A$ satisfies $\cV$, and
        \item  the row span of $A$ is equal to $U$.
    \end{enumerate}
    Furthermore, if $M_n(\cV) \in \cV_n$ is a witness for $A$ satisfying $\cV$, then we say that $M_n(\cV)$ is a witness for $\cC$ containing $(\cV, U)$.
\end{definition}

\begin{observation}\label{obs:code-contains-cv-eq}
    A code $\cC \in \F_q^n$ contains $\cV$ if and only if $\cC$ contains $(\cV, U)$ for some $U \in \cL_\dist(\F_q^L)$.
\end{observation}

We now state a simple fact regarding the inclusion of matrices in a random linear code.
\begin{fact}
    For a matrix $A \in \F_q^{n \times L}$, the probability that $A$ is contained in rate $R$ RLC $\cC \subseteq \F_q^n$ is equal to
    \[
        \Pr_\cC [A \subseteq \cC] = q^{-(1-R)n \rank A}.
    \]
\end{fact}
\begin{proof}
    Since $\cC$ by our definition is the kernel of a uniformly random matrix $K \in \F_q^{(1-R)n \times n}$, we can write
   \[
        \Pr_\cC [A \subseteq \cC] = \Pr_{K \sim \F_q^{(1-R)n \times n}}[K \cdot A=\0 ]= \prod_{i \in [(1-R)n]} \Pr_{K \sim \F_q^{(1-R)n \times n}}[K[i][] \cdot A=\0 ] = q^{-(1-R)n \rank A}.
    \]
\end{proof}

\begin{definition}[Potential]
    For $R \in [0, 1]$, and a subspace $U \in \cL(\F_q^L)$, define the potential $\pdeg(\cV, U, R)$ as
    \[
        \pdeg(\cV, U, R) := \sum_{t \in [T]}f_t\cdot\dim(\ker(\mM_t) \cap U) - (1-R) \dim U.
    \]
    Here, $\ker(\mM_t)$ is a subspace in $\Fq^L$.
\end{definition}

The following lemma provides some motivation as to why we require the definition.
\begin{lemma}\label{lem:prob-cont-vu}
    For an $L$-local profile $M_n(\cV)$ created according to $\cV$ and a subspace $U \in \F_q^L$, the probability that an RLC $\cC \subseteq \F_q^n$ of rate $R$ contains $(\cV, U)$ with $M_n(\cV)$ as a witness is at most $q^{\pdeg(\cV, U, R)n}$.
\end{lemma}
\begin{proof}
    The proof is given in section 4 of \cite{LMS25}, but we include it here for completeness. Define the set
    \[
        \cM_{(M_n(\cV), U)} := \inset{A \in \F_q^{n \times L} \mid \forall i \in [n], A[i][] \in \ker(M_i) \land \rspn(A)=U}.
    \]
    This is the set of all matrices that satisfy $\cV$ by ``complying'' with the constraints specified by $M_n(\cV)$, while also having $\rspn(A)=U$.
    We now define a set similar to the one above, except we now require $\rspn(A) \subseteq U$.
    \[
        \cM^*_{(M_n(\cV), U)} := \inset{A \in \F_q^{n \times L} \mid \forall i \in [n], A[i][] \in \ker(M_i) \land \rspn(A) \subseteq U}.
    \]
    It is easy to see that $\cM_{(M_n(\cV), U)} \subseteq \cM^*_{(M_n(\cV), U)}$.

    By a union bound, the probability that a rate $R$ RLC $\cC \subseteq \F_q^n$ contains $(\cV, U)$ with $M_n(\cV)$ as a witness is at most:
    \begin{equation}\label{eq:prob-rlc-contains-a}
        \sum_{A \in \cM_{(M_n(\cV), U)}} q^{-(1-R)n \cdot \rank(A)} = \inabs{\cM_{(M_n(\cV), U)}} \cdot q^{-(1-R)n \cdot \rank(A)} \le \inabs{\cM^*_{(M_n(\cV), U)}} \cdot q^{-(1-R)n \cdot \dim U}.
    \end{equation}
    We now proceed to estimate $\inabs{\cM^*_{(M_n(\cV), U)}}$. Upon observing that this set is a linear subspace of $\F_q^{n \times L}$, we see that
    \[
        \log_q \inbrak{\inabs{\cM^*_{(M_n(\cV), U)}}} = \sum_{i \in [n]} \dim (\ker M_i \cap U) = \sum_{t \in [T]}f_t \cdot n \cdot\dim(\ker(\mM_t) \cap U).
    \]
    By plugging the value of $|\cM^*_{(M_n(\cV), U)}|$ in \cref{eq:prob-rlc-contains-a} we see that the probability of $\cC$ containing $(\cV, U)$ with $M_n(\cV)$ as a witness is at most $q^{\pdeg(\cV, U, R)n}$.
\end{proof}

\begin{definition}\label{def:thresh-rate}
    We define $R_{\cV, U}$ to be 
    \[
        R_{\cV, U} = \min \inset{R \in [0, 1] \mid \pdeg(\cV, U, R) \ge \pdeg(\cV, W, R) \text{ for every linear subspace } W \subseteq U}.
    \]
\end{definition}
Note that the minimum exists, as $\pdeg(\cV, U, 0) \ge \pdeg(\cV, W, 0)$ for every subspace $W \subseteq U$, by the definition of the potential $\pdeg$.

Let us provide some motivation for the definition.
We first state Proposition 4.3 from \cite{LMS25}.
\begin{proposition}[RLC thresholds for local profiles (Proposition 4.3, \cite{LMS25})]\label{prop:rlc-thresh}
    Let $\cC \subseteq \F_q^n$ be a RLC of rate $R \in [0, 1]$.
    Fix some $M_n(\cV) \in \cV_n$, and a $U \in \cL(\F_q^L) \setminus \inset{\inset{0}}$. Let
    \[
        \gamma := \min_{\substack{W \in \cL(\F_q^L) \\ W \subsetneq U }} \inset{ \pdeg(\cV, U, R) - \pdeg(\cV, W, R)}.
    \]
    The following then holds.
    \begin{enumerate}
        \item If $\gamma < 0$, then $\Pr_\cC[\cC \text{ contains } (\cV, U)\text{ with }M_n(\cV)\text{ as a witness}] \le q^{\gamma n}$.
        \item If $\gamma > 0$, then $\Pr_\cC[\cC \text{ contains } (\cV, U)\text{ with }M_n(\cV)\text{ as a witness}] \ge 1-q^{-{\gamma n}+L^2}$.
    \end{enumerate}
\end{proposition}

Observe that the value of $\gamma$ depends on the rate $R$, the local profile description $\cV$, and subspace $U$, but is independent of the local profile $M_n(\cV)$.
From \cref{prop:rlc-thresh}, we see that upon union bounding over all local profiles $M_n(\cV) \in \cV_n$, we get:
\begin{corollary}\label{cor:rlc-thresh}
    Let $\cC \subseteq \F_q^n$ be a RLC of rate $R \in [0, 1]$.
    For a $U \in \cL(\F_q^L) \setminus \inset{\inset{0}}$, let $\gamma$ be as defined in \cref{prop:rlc-thresh}.
    Then the following holds.
    \begin{enumerate}
        \item If $\gamma < 0$, then $\Pr_\cC[\cC \text{ contains } (\cV, U)] \le \inabs{\cV_n} \cdot q^{\gamma n}$.
        \item If $\gamma > 0$, then $\Pr_\cC[\cC \text{ contains } (\cV, U)] \ge 1-q^{-{\gamma n}+L^2}$.
    \end{enumerate}    
\end{corollary}

Thus, \cref{cor:rlc-thresh} implies that for $R=R_{\cV, U}-\eps$, where $0<\eps<R_{\cV, U}$ is a constant, a rate $R$ RLC contains $(\cV, U)$ with exponentially low probability (provided that $\inabs{\cV_n}$ is sufficiently small), while for $R=R_{\cV, U}+\eps$, a rate $R$ RLC contains $(\cV, U)$ with probability exponentially close to $1$.

We now define a threshold rate with respect to an RLC containing $\cV$. Invoking \cref{obs:code-contains-cv-eq}, we define the threshold rate corresponding to $\cV$ as
\[
    R_\cV := \min_{U \in \cL_\dist(\F_q^L)} R_{\cV, U}.
\]

For an $L$-LCL property $\cP$, recall that a code $\cC$ satisfies $\cP$ if $\cC$ contains some $\cV \in \cP$.
We thus define:
\begin{definition}\label{def:rate-thresh-prop}
    The threshold rate of an $L$-LCL property $\cP$ is defined to be
\begin{equation}\label{eq:rate-thresh-prop}
    R_\cP := \min_{\cV \in \cP} R_\cV = \min_{\substack{ \cV \in \cP \\ U \in \cL_\dist(\F_q^L)}} R_{\cV, U}.
\end{equation}
\end{definition}

Theorem 4.4 from \cite{LMS25} proves that $R_\cP$ as defined above is indeed a threshold rate.
\begin{theorem}[Theorem 4.4, \cite{LMS25}]\label{thm:rp-is-tight}
    Let $\cP$ be an $L$-LCL property.
    Let $\cC \subseteq \F_q^n$ be an RLC of rate $R$, and let $R_\cP$ be as defined in \cref{def:rate-thresh-prop}.
    Let $\eps>0$ be a sufficiently small constant.
    The following now holds
    \begin{enumerate}
        \item If $R \ge R_\cP + \eps$, then $\Pr_\cC [\cC \text{ satisfies }\cP] \ge 1-q^{-\eps n + L^2}$.
        \item If $R \le R_\cP - \eps$, then $\Pr_\cC [\cC \text{ satisfies }\cP] \le \sum_{\cV \in \cP} \inabs{\cV_n} \cdot q^{-\eps n + L^2}$.
    \end{enumerate}
\end{theorem}

At this point we provide a brief discussion about the LCL properties $\cP$ considered henceforth.
In general, we allow the fractions $f_1,\ldots,f_T$ to range over some fixed sub-intervals of $[0, 1]$, and regard all local profile descriptions with such fractions as belonging to $\cP$.
Consequently, $\cP$ is uncountably infinite.
This, however, does not pose a difficulty, because for a block length $n$, the set of local profile descriptions satisfying $\inabs{\cV_n}>0$ is itself finite, as by definition, each fraction in such local profile descriptions must be an exact multiple of $1/n$.
Furthermore, for any such local profile description, the corresponding set of local profiles $\cV_n$ is also finite.
Specifically, $\inabs{\cV_n}$ is equal to the number of ways we can arrange $n$ objects in a row, where we have $f_t\cdot n$ objects of type $t$, for $t \in T$.
It follows, therefore, that 
\[
    \sum_{\cV \in \cP} \inabs{\cV_n}
\]
is finite.

Since our definition of LCL properties is derived from the notion of local profile descriptions, the properties we consider are coordinate-permutation invariant.
An LCL property~$\cP$ is said to be coordinate-permutation invariant if, for any local profile $M = (M_1, \ldots, M_n)$ associated with~$\cP$, all local profiles obtained by permuting the order of the matrices in~$M$ are also associated with~$\cP$.

At first glance, this might appear to restrict the class of properties we can consider; however, it does not.
Indeed, for a local profile description~$\cV$, the threshold rate~$R_\cV$ depends only on the frequency of the matrices appearing in the description, and not on their order.
Consequently, the threshold rate remains the same whether we consider a single local profile described by~$\cV$ or all of them.
Thus, if we initially intended to include in our LCL property~$\cP$ only a subset of the local profiles described by~$\cV$, we may in fact include all of them without any loss in threshold rate.
This does result in an increase in alphabet size, though it remains constant.

We require Lemma 4.5 from \cite{LMS25}:
\begin{lemma}[Lemma 4.5, part (3), \cite{LMS25}]\label{lem:argmax-not-dist}
    For $R \in [0, 1]$, denote
    \[
        \argmax \inset{\pdeg(\cV, *, R)} := \inset{W \in \cL(\F_q^L) \mid \pdeg(\cV, W, R) = \max_{U \in \cL(\F_q^L)} \pdeg(\cV, U, R)}.
    \]
    Then, $\argmax \pdeg(\cV, *, R_\cV)$ contains at least one element from $\cL(\F_q^L) \setminus \cL_\dist (\F_q^L)$, and at least one element from $\cL_\dist (\F_q^L)$.
\end{lemma}

\begin{claim}\label{clm:canonical-w}
    There is a canonical $W_\cV \in \cL(\F_q^L) \setminus \cL_\dist(\F_q^L)$ such that for all $R \le R_\cV$ and for every $U \in \cL(\F_q^L)$ satisfying $\dim U \ge \dim W_\cV$, the inequality
    \[
        \pdeg(\cV, U, R) \le \pdeg(\cV, W_\cV, R)
    \]
    is true.
\end{claim}
\begin{proof}
    By \cref{lem:argmax-not-dist}, the set  
    \[
        \argmax \pdeg(\cV, *, R_\cV)\;\cap\;\bigl(\cL(\F_q^L)\setminus \cL_\dist(\F_q^L)\bigr)
    \]
    is non-empty.  
    Fix an arbitrary total ordering on the subspaces of $\cL(\F_q^L)$, and select the first subspace in this ordering from the above set.
    Denote this subspace by $W_\cV$.

    Fix a subspace $U \in \cL(\F_q^L)$ satisfying $\dim U \ge \dim W_\cV$, and some $R \le R_\cV$. Then,
    \begin{align*}
        \pdeg(\cV, U, R) - \pdeg(\cV, W_\cV, R) &= \sum_{t \in [T]} f_t \cdot (\dim(\ker(\mM_t) \cap U)- \dim(\ker(\mM_t) \cap W_\cV)) \\
        &\quad - (1-R)(\dim U - \dim W_\cV) \\
        &\le \sum_{t \in [T]} f_t \cdot (\dim(\ker(\mM_t) \cap U)- \dim(\ker(\mM_t) \cap W_\cV)) \\
        &\quad - (1-R_\cV)(\dim U - \dim W_\cV) \\
        &= \pdeg(\cV, U, R_\cV) - \pdeg(\cV, W_\cV, R_\cV) \\
        &\le 0,
    \end{align*}
    where the last inequality holds because $W_\cV \in \argmax \pdeg(\cV, *, R_\cV)$, and the first one holds because $R \le R_\cV$ and $\dim U \ge \dim W_\cV$.
\end{proof}

\subsection{Implied Local Profile Descriptions}
Fix an $L$-local profile description $\cV=((f_1, \mM_1), \ldots, (f_{T}, \mM_{T}))$.

\begin{definition}[Implied Local Profile Description]\label{def:imp-loc-prof-desc}
    Let $\psi$ be a linear map $\psi: \F_q^L \rightarrow \F_q^{L-K}$ for some integer $K \le L$.
    Then, the $(L-K)$-\emph{implied local profile description of $\cV$ with respect to $\psi$}, denoted by $\cV^{\psi}$, is an $(L-K)$-local profile description defined as
    \[
        \cV^\psi:=((f_1, \mM^\psi_1), \ldots, (f_{T}, \mM^\psi_{T})).
    \]
    Here, for each $t \in [T]$, $\mM^\psi_t$ is a matrix satisfying $\ker \mM^\psi_t=\psi( \ker \mM_t)$.
\end{definition}

We now give some intuition for \cref{def:imp-loc-prof-desc}.
\begin{observation}\label{obs:imp-local-prof-desc}
    For a linear map $\psi: \F_q^L \rightarrow \F_q^{L-K}$, let $\cV^{\psi}$ be the corresponding $(L-K)$-\emph{implied local profile description} of $\cV$.
    If a linear code $\cC \subseteq \F_q^n$ contains $\cV$, then there is a matrix $B \in \F_q^{n \times (L-K)}$ (with possibly non-distinct columns) such that
    \begin{enumerate}
        \item $B \subseteq \cC$, and
        \item $B$ satisfies $\cV^\psi$.
    \end{enumerate}
\end{observation}
\begin{proof}
    By \cref{def:con-loc-prof}, if $\cC$ contains $\cV$, then there is a matrix $A \in \F_q^{n \times L}$ that (i) has pairwise distinct columns, (ii) is contained in $\cC$, and (iii) satisfies $\cV$.
    Now consider the matrix $B \in \F_q^{n \times (L-K)}$ constructed row by row, such that $B[i][] := \psi(A[i][])$.
    By the linearity of $\cC$, $A \subseteq \cC$ implies $B \subseteq \cC$, and moreover, it is easy to verify that $B$ satisfies $\cV^\psi$.
\end{proof}

Let $W_\cV$ be the canonical subspace guaranteed by \cref{clm:canonical-w} for $\cV$.
Recall that for all $R \le R_\cV$ and for every $U \in \cL(\F_q^L)$ satisfying $\dim U \ge \dim W_\cV$, $W_\cV$ satisfies
\[
    \pdeg(\cV, U, R) \le \pdeg(\cV, W_\cV, R).
\]
Fix a full rank linear map $\psi_\cV \colon \F_q^L \rightarrow \F_q^{L-\dim W_\cV}$ such that $\ker \psi_\cV = W_\cV$.
\begin{definition}[Canonical Implied Local Profile Description]
    Let $\psi_\cV$ be as defined above.
    The $(L-\dim W_\cV)$-\emph{canonical implied local profile description} of $\cV$, denoted by $\cV^{\imp}$, is an $(L-\dim W_\cV)$-local profile description defined as
    \[
        \cV^\imp:=((f_1, \mM^{\imp}_1), \ldots, (f_{T}, \mM^{\imp}_T)).
    \]
    Here, $\mM^\imp_t$ is a matrix satisfying $\ker \mM^\imp_t=\psi_\cV( \ker \mM_t)$ for each $t \in [T]$.
\end{definition}

\paragraph{Robust Local Profile Descriptions.} We will now describe a general method to create robust counterparts to local profile descriptions.

\begin{definition}[Robust Local Profile Description]
    Let $\D \in [0, 1]$ be a constant, and $\cV=((f_1, \mM_1), \ldots, (f_{T}, \mM_{T}))$ be an $L$-local profile description. 
    Consider the set of all $L$-local profile descriptions of the form
    \[
        \Bigl( (f_1 - \D_1,\;\mM_1),\; \ldots,\; (f_T - \D_T,\;\mM_T),\;
       \bigl(\textstyle\sum_{t\in[T]}\D_t,\;\0\bigr) \Bigr)
    \]

    where for every $t \in [T]$, we have $0 \le \D_t \le f_t$ and $\sum_{t \in [T]}\D_t \le \D$.
    Recall that $\0$ is the all zeroes matrix.
    We denote this set by $\Rob_\D(\cV)$, and an element from this set is known as a \emph{$\D$-robust local profile description} of $\cV$.
\end{definition}

Informally speaking, any code $\cC$ that avoids containing matrices that satisfy local profile descriptions from $\Rob_\D(\cV)$ will consequently avoid containing matrices that only satisfy a $(1-\D)$ fraction of the constraints set forth by $\cV$.

\begin{observation}\label{obs:robust-loc-prof-desc}
    For every constant $\D \in [0, 1]$, a matrix satisfying $\cV$ also satisfies every local profile description from $\Rob_\D(\cV)$.
\end{observation}

We now state and prove a lemma asserting that if the rate $R$ is bounded away from the threshold rate $R_\cV$ by $\eps$, then the potential of the robust counterparts of $\cV^\imp$ is bounded above by $-\eps$, up to a small additive factor.

\begin{lemma}\label{lem:pdeg-rob-implied-ub}
    For $\eps>0$, let $R \le R_\cV-\eps$.
    Let $\cV^\imp_\D \in \Rob_\D(\cV^\imp)$ be some $\D$-robust local profile description of $\cV^\imp$.
    Then for every $U' \in \cL(\F_q^{(L-\dim W_\cV)}), U' \neq \inset{0}$, we have
    \[
        \pdeg(\cV_\D^\imp, U', R) \le -\eps + \D \cdot L.
    \]
\end{lemma}
\begin{proof}

    By the guarantee for $W_\cV$ (see \cref{clm:canonical-w}), we have for all subspaces $U \in \cL(\F_q^L)$ satisfying $\dim U \ge \dim W_\cV$,
    \begin{align} 
        \pdeg(\cV, U, R) - \pdeg(\cV, W_\cV, R) &= \pdeg(\cV, U, R_\cV) - \pdeg(\cV, W_\cV, R_\cV) - \eps(\dim U - \dim W_\cV) \nonumber \\ 
        &\leq -\eps (\dim U - \dim W_\cV). \label{eq:wv-pdeg-max}
    \end{align}
    Let $U' \in \cL(\F_q^{(L-\dim W_\cV)})$ be such that $U' \neq \inset{0}$, and let $U \in \cL(\F_q^L)$ be such that $W_\cV \subset U$ and $\psi_\cV(U)=U'$, that is, $U=\psi_\cV^{-1}(U')=\inset{u \in \F_q^L \mid \psi_\cV(u) \in U'}$.
    By the rank-nullity theorem, $\dim U' = \dim U - \dim W_\cV$.

    Consequently,
    \begin{align}
        \pdeg(\cV_\D^\imp, U', R) &= \bigl(\sum_{t \in [T]} (f_t-\D_t)\cdot \dim(\ker (\mM^\imp_t) \cap U')\bigr) \nonumber\\
                            & \quad + \sum_{t \in [T]}\D_t \cdot\dim(\ker \0^\imp) - (1-R)\dim U' \nonumber\\
                           &= \bigl(\sum_{t \in [T]} (f_t-\D_t) \cdot \dim(\psi_\cV(\ker (\mM_t) \cap U))\bigr) \label{eq:subsp-int} \\
                           & \quad + \sum_{t \in [T]}\D_t \cdot \dim(\psi_\cV(\ker \0)) - (1-R)(\dim U-\dim W_\cV)  \nonumber\\
                           &= \bigl(\sum_{t \in [T]} (f_t-\D_t) \cdot \inbrak{\dim(\ker(\mM_t) \cap U)- \dim(\ker(\mM_t) \cap W_\cV)}\bigr)  \label{eq:app-rank-nullity}\\
                           & \quad + \sum_{t \in [T]}\D_t \cdot \dim(\psi_\cV(\ker \0)) - (1-R)(\dim U-\dim W_\cV)  \nonumber\\
                           &\le \pdeg(\cV, U, R) - \pdeg(\cV, W_\cV, R) + \D \cdot L \nonumber\\
                           &\le -\eps + \D \cdot L, \nonumber
    \end{align}
    where \cref{eq:subsp-int} follows from \cref{clm:subsp-int}, stated and proved below.
    Moreover, \cref{eq:app-rank-nullity} follows from an application of the rank-nullity theorem.
    The final inequality follows from \cref{eq:wv-pdeg-max}, owing to the fact that $U \supset W_\cV$, given that $U'$ is a non-trivial subspace.
\end{proof}

\begin{claim}\label{clm:subsp-int}
    Let $W_\cV, \psi_\cV, U'$ and $U$ be as defined in \cref{lem:pdeg-rob-implied-ub}.
    Then, for every subspace K in $\F_q^L$, we have
    \[
        \psi_\cV(K) \cap U' = \psi_\cV(K) \cap \psi_\cV(U) = \psi_\cV(K \cap U).
    \]
\end{claim}
\begin{proof}
    First, $\psi_\cV(K \cap U) \subseteq \psi_\cV(K) \cap \psi_\cV(U)$ is immediate: if $x \in \psi_\cV(K \cap U)$, then there is a $x' \in K \cap U$ such that $\psi_\cV(x')=x$.
    Then, $x' \in K$ implies $x=\psi_\cV(x') \in \psi_\cV(K)$, and $x' \in U$ implies $x=\psi_\cV(x') \in \psi_\cV(U)$.

    For the other inclusion, take $y \in \psi_\cV(K) \cap \psi_\cV(U)$.
    Then, there exist $k \in K$ and $u \in U$ such that $\psi_\cV(u)=y=\psi_\cV(k)$.
    Hence
    \[
        k-u \in \ker \psi_\cV = W_\cV.
    \]
    But recall that $W_\cV \subset U$, and thus every element of $\ker \psi_\cV$ lies in $U$.
    Therefore, $k=u+(k-u) \in U$.
    Since $k \in K$ as well, $k \in K \cap U$, and so $y = \psi_\cV(k) \in \psi_\cV(K \cap U)$.
    Thus, $\psi_\cV(K) \cap \psi_\cV(U) \subseteq \psi_\cV(K \cap U)$.
    
\end{proof}

\begin{lemma}\label{lem:imp-prof-ub}
    For $0 \le \D \le 1$, consider some $\cV^\imp_\D \in \Rob_\D(\cV^\imp)$.
    For a block length $n$, the number of local profiles associated with $\cV_\D^\imp$ is at most $2^n \inabs{\cV_n}$.
\end{lemma}
\begin{proof}
Let $\cV^\imp_\D$ be of the form
\[
    \cV^\imp_\D = \Bigl( (f_1 - \D_1,\;\mM_1), \ldots,(f_T - \D_T,\;\mM_T),
       \bigl(\textstyle\sum_{t\in[T]}\D_t,\;\0\bigr) \Bigr),
\]
where for every $t \in [T]$, we have $0 \le \D_t \le f_t$ and $\sum_{t \in [T]}\D_t \le \D$.
Each local profile associated with $\cV_\D^\imp$ can be created by following this recipe: fix some local profile $M_n(\cV)$ created according to $\cV$. Then,
\begin{enumerate}[(i)]
    \item For each coordinate, change the associated matrix from $\mM_t$ to $\mM_t^\imp$.
    \item for each $t \in [T]$, select some $\D_t \cdot n$ coordinates whose associated matrix is $\mM_t$ and change the associated matrix to $\0$.
\end{enumerate}
The first step does not cause an increase in the number of local profiles.
The number of ways to perform the second step is upper bounded by $\binom{n}{ \le \D \cdot n} \le 2^n$.
\end{proof}

We now state and prove a lemma that establishes an upper bound on the probability with which an RLC of rate below the rate threshold $R_\cP$ contains non-zero matrices satisfying a robust local profile description of $\cV^\imp$.
\begin{lemma}\label{lem:ub-bad-prob-lcl}
    Let $\cP$ be an $L$-LCL property.
    Let $R = R_\cP-\eps$ for a constant $\eps \in [0, R_\cP)$, and let $\D$ be a constant in $[0, 1]$.
    Then the probability that an RLC $\cC \subseteq \F_q^n$ of rate $R$ contains a non-zero matrix satisfying a local profile description from the set
    \[
        \bigcup_{\cV \in \cP} \Rob_\D(\cV^\imp)
    \]
    is at most $\sum_{\cV \in \cP} |\cV_n|\cdot 2^{n} \cdot q^{L^2-\eps n + \D \cdot L n}$.
\end{lemma}
\begin{proof}
    By \cref{lem:imp-prof-ub}, we see that for an $L$-LCL property $\cP$, the total number of local profiles that can be created according to local profile descriptions from the set $\bigcup_{\cV \in \cP} \Rob_\D(\cV^\imp)$ is at most
    \[
        \sum_{\cV \in \cP} |\cV_n| \cdot 2^{n}.
    \]
    For some $\cV \in \cP$, consider a local profile $M_n(\cV_\Dp^\imp)$ created according to $\cV_\Dp^\imp \in \Rob_\D(\cV^\imp)$.
    Recall that $\cV_\Dp^\imp$ is a $(L- \dim W_\cV)$-canonical implied local profile description.
    Fix a non-zero subspace $U \in \F_q^{(L- \dim W_\cV)}$.
    From \cref{lem:prob-cont-vu}, and the fact that
    \[
        R \le R_\cP-\eps \le R_\cV-\eps \le R_{\cV, U}-\eps,
    \]
    which follows from \cref{eq:rate-thresh-prop}, we see that the probability that there is a non-zero matrix $A \in \F_q^{n \times (L-\dim W_\cV)}$ such that $A \subseteq \cC$, and $A$ satisfies $(\cV_\Dp^\imp, U)$ with $M_n(\cV_\Dp^\imp)$ as a witness is at most
    \[
        q^{\pdeg(\cV_\Dp^\imp, U, R)n} \le q^{(-\eps+ \D \cdot L)n},
    \]
    where the inequality follows from \cref{lem:pdeg-rob-implied-ub}.
    The result follows by union bounding over all non-zero subspaces $U$, and all local profiles created according to local profile descriptions from the set $\bigcup_{\cV \in \cP} \Rob_\D(\cV^\imp)$. 
\end{proof}

We now turn to analyze the upper bound on the probability in \cref{lem:ub-bad-prob-lcl}.

\begin{definition}\label{def:kappa}
    For a prime power $q$ and all sufficiently large $n$, we define $\kappa_{q}(\cP) \in \mathbb{R}$ as
    \[
        \kappa_{q}(\cP) := \inf_\kappa \inset{ \kappa : \sum_{\cV \in \cP} |\cV_n| \le q^{\kappa n} }.
    \]
\end{definition}
The probability upper bound is non-trivial if and only if $\sum_{\cV \in \cP} |\cV_n|\cdot 2^{n} \cdot q^{L^2-\eps n}\cdot q^{\D \cdot L n} < 1$.
This is achieved when
\[
    \kappa_{q}(\cP) + \log_q 2+\frac{L^2}{n}-\eps + \D \cdot L < 0,
\]
Upon restricting $\Dp$ to be at most $\eps/2L$ and $n$ to be at least $4L^2/\eps$, it suffices to have
\begin{equation}\label{eq:cond-for-reas}
    \kappa_{q}(\cP) + \log_q 2-\frac{\eps}{4} < 0.
\end{equation}
Here, we are interested in the smallest prime power $q$ that satisfies the above inequality.
Such a $q$ exists.
Indeed, this is seen because the total number of local profiles associated with $\cP$---equal to $q^{\kappa_q(\cP)n}$---is at most $q_0^{L^2n}$.
Recall that $\Fqz$ is the field over which constraints of $\cP$ were defined.
Thus, $\kappa_q(\cP) \le \log_q (q_0) \cdot L^2$, which implies
\[
    \lim_{q \rightarrow \infty} (\kappa_q(\cP)+\log_q 2)=0.
\]
We choose the smallest $q$ so that $\kappa_{q}(\cP) + \log_q 2$ is lesser than $\eps/4$, and thus \cref{eq:cond-for-reas} is satisfied.

We summarize the above discussion as the following fact:
\begin{fact}\label{fct:good-in-code}
    For an $L$-LCL property $\cP$ and $R \le R_\cP-\eps$, there exists a minimum prime power $q$ that is solely a function of $\eps, L,q_0,$ and $\cP$, such that the following holds:
    There exists a linear code $\cC \subseteq \F_q^n$ for every block length $n \ge 4L^2/\eps$ that contains no non-zero matrices satisfying local profile descriptions from the set $\bigcup_{\cV \in \cP} \Rob_\D(\cV^\imp)$.
\end{fact}

\begin{lemma}\label{lem:good-in-code-brute-force}
    For a block length $n \ge 4L^2/\eps$ and an $L$-LCL property $\cP$, let $\cC$ be a linear code $\cC$ as described in \cref{fct:good-in-code}. Then,
    \begin{enumerate}
        \item the code $\cC$ can be found in time $q^{\poly(n, L)}$, and
        \item $\cC$ does not satisfy $\cP$.
    \end{enumerate}
\end{lemma}
\begin{proof}
    We first prove that $\cC$ can be found in time $q^{\poly(n, L)}$.
    By our definition of random linear codes, a non-zero probability is placed on at most $q^{(1-R)n}$ linear subspaces of $\F_q^n$.
    For each such code, we perform a check of the following form: determine whether it contains a non-zero matrix satisfying a local profile description from the set $\bigcup_{\cV \in \cP} \Rob_\D(\cV^\imp)$.
    We return the first code $\cC$ that does not contain any matrices of the described form.
    Such a code is guaranteed to exist by \cref{fct:good-in-code}.
    Let us analyze the runtime of this algorithm.
    
    The number of non-zero matrices satisfying local profile descriptions from the set $\bigcup_{\cV \in \cP} \Rob_\D(\cV^\imp)$ is at most the total number of matrices of dimensions $n \times L$, which is equal to $q_0^{nL} \le q^{nL}$.
    Given the parity check matrix of a linear code $\cC \subseteq \F_q^n$, we can check whether every column of a matrix $A \in \F_q^{n \times L}$ lies in $\cC$ in time at most $n^2L$.
    Therefore, for each linear code, we can perform the aforementioned check in time at most $n^2L \cdot q^{nL}$.
    Consequently, the total runtime is at most
    \[
        q^{(1-R)n}\cdot n^2L \cdot q^{nL} \le q^{\poly (n, L)}.
    \]

    We now prove that $\cC$ does not satisfy $\cP$.
    Assume for contradiction that it does.
    Then, there is a $\cV \in \cP$ such that $\cC$ contains $\cV$.
    By \cref{obs:imp-local-prof-desc}, $\cC$ contains a matrix $B$ that satisfies $\cV^\imp$.
    The matrix $B$ was obtained from some matrix $A \subseteq \cC$ having pairwise distinct columns, by applying a map $\psi$ to every row of $A$.
    By the definition of $\cV^\imp$, there exist $i, j \in [L]$, where $i \neq j$, such that the kernel of $\psi$ consists of a subspace whose vectors agree at coordinates $i, j$.
    This implies that $B$ is non-zero.
    Furthermore, by \cref{obs:robust-loc-prof-desc}, $B$ also satisfies every $\cV^\imp_\D \in \Rob_\D(\cV)$.
    This contradicts the fact that $\cC$ is of the form as specified in \cref{fct:good-in-code}.
\end{proof}

Recall that the code produced by the AEL construction, $\cael$, has alphabet $\sigin^{d}$, where $\sigin$ is the alphabet of the inner code.
We take the inner code to be over $\F_{q}$, and therefore interpret $\cael$ as a code over the extension field $\F_{Q}$, where $Q = q^{d}$.
Furthermore, by inspection, the definitions of containment and satisfiability (\cref{def:sat-loc-prof-desc}, \cref{def:cont-matr}, \cref{def:con-loc-prof}, \cref{def:sat-lcl-prop}) also apply to matrices and codes over $\F_Q$. In particular, they apply to $\cael \subseteq \F_Q^{N\dd}$.

\subsection{Main Result}
We now state and prove the central theorem of this section.
Fix an $L$-LCL property $\cP$, and $\rin \le R_\cP-\eps$ for some $\eps>0$.
Let $\Dp = \eps/2L$ be as defined above.
Let $T_\cP$ denote the maximum number of distinct matrices appearing in any local profile description $\cV \in \cP$.
Let $G$ be the bipartite graph obtained from \cref{clm:kmrs} by setting 
\[
    \eta=\frac{\Dp}{4T_\cP},
\]
and
\[
    \zeta=\frac{\dout}{2T_\cP},
\]
having degree
\[
    d=O\inbrak{\frac{1}{\zeta\eta^2}} = O\inbrak{\frac{T_\cP^3}{\Dp^2\cdot \dout}}.
\]
For any integer $N$ satisfying $N > T_\cP/\dout$, take $G$ to have $N\dd$ vertices on both sides.
Let $\cin \subseteq \F_q^{d}$ be an $\F_q$-linear code of rate $\rin$, where $q$ is the minimum prime power guaranteed by \cref{fct:good-in-code}, and let $\cout \subseteq (\F_q^{\rin d\dd})^N$ be an $\F_q$-linear code having distance $\dout$.
By definition, $\phi$ is now an $\F_q$-linear map of the form $\phi \colon \F_q^{\rin d \dd} \rightarrow \cin$.

\begin{theorem}\label{thm:the-big-one}
Let $\cin, \cout$, and $G$ be as defined above.
Furthermore, assume that $\cin$ is a linear code whose existence is guaranteed by \cref{fct:good-in-code}. 
Then, $\cael(\cin, \cout, G) \subseteq \F_Q^{N\dd}$ is an $\F_q$-linear code that \textbf{does not} satisfy $\cP$.
\end{theorem}
\begin{proof}
    Since both $\cout$ and the map $\phi$ are $\F_q$-linear, it follows that $\cael$ is itself an $\F_q$-linear code.

    Assume for contradiction that $\cael$ satisfies $\cP$.
    Then, there is a $\cV=(f_1, \mM_1),\ldots,(f_{T},\mM_{T})) \in \cP$ such that $\cael$ contains $\cV$, which implies the existence of a matrix $A \in \F_Q^{N\dd \times L}$, such that
    \begin{enumerate}
        \item $A \subseteq \cael$,
        \item $A$ satisfies $\cV$, and
        \item $A$ has pairwise distinct columns.
    \end{enumerate}

    Let the columns of $A$ be pairwise distinct codewords $c_1,\ldots,c_L \in \cael \subseteq \F_Q^{N\dd}$.
    Note that we can reinterpret these codewords as vectors in $(\F_q^d)^{N\dd}$, and we choose to do so.
    We now perform three operations.
    Firstly, we create $A_\fl \in \F_q^{([N] \times [d]) \times L}$ from $A$ by \emph{flattening} every column of $A$ (see \cref{def:flatten}) and then collecting them in a matrix $A_\fl$.
    Therefore, for every $i \in [L]$, $A_\fl[][i]=(A[][i])_\fl$.
    Secondly, we create $A^\proj_\fl \in \F_q^{([N] \times [d]) \times L}$ by performing the \emph{projection} operation on $A_\fl$ (see \cref{def:proj}).

    We now instantiate the objects required to perform the third operation.
    As $R \le R_\cP-\eps \le R_\cV-\eps$, by \cref{clm:canonical-w} there exists a (canonical) subspace $W_\cV$ for $\cV$.
    In particular, because
    \[
        W_\cV \in (\cL(\F_q^L) \setminus \cL_\dist(\F_q^L)),
    \]
    $\dim W_\cV \neq L$ holds, as the only subspace of dimension $L$, $\F_q^L$, lies in $\cL_\dist(\F_q^L)$.
    Define a linear map $\psi \colon \F_q^L \rightarrow\F_q^{(L-\dim W_\cV)}$ satisfying $\ker \psi = W_\cV$.
    Note that because $\dim W_\cV < L$, this map is well-defined.
    The third operation consists of applying $\psi$ on each row of $A^\proj_\fl$.
    We denote the resultant matrix by $A^\proj_\psi \in \F_q^{([N] \times [d]) \times (L-\dim W_\cV)}$.
    That is, $A^\proj_\psi$ satisfies
    \[
        A^\proj_\psi[(\ell, i)][]=\psi(A^\proj_\fl[(\ell, i)][])
    \]
    for all $\ell \in [N]$ and $i \in [d]$.
    
    We recall \cref{def:proj-on-vertices}, and note that $A^\proj_\psi(\ell)$ denotes the submatrix of $A^\proj_\psi$ consisting of rows indexed by vertex-edge pairs corresponding to the vertex $\ell \in [N]$.
    This submatrix lives in $\F_q^{[d] \times (L-\dim W_\cV)}$.

    We now prove a claim stating that the submatrices $A^\proj_\psi(\ell)$ corresponding to most left vertices in fact satisfy $\cV_\Dp^\imp$, for some $\cV^\imp_\D \in \Rob_\D(\cV^\imp)$.

    \begin{claim}\label{clm:many-parts-sat-cvrob}
        There exists a set of indices $L^* \subseteq [N]$ satisfying $\inabs{L^*} > (1-\dout)N$ such that for every $\ell \in L^*$, there exists a $\cV^\imp_\D \in \Rob_\D(\cV^\imp)$ such that the matrix $A^\proj_\psi(\ell)$ satisfies $\cV^\imp_\D$.
    \end{claim}

    We state the rest of the proof assuming \cref{clm:many-parts-sat-cvrob}, whose proof is given below.

    We now prove that because the columns of $A$ are pairwise distinct, $A^\proj_\psi$ contains at least one non-zero row.
    This is seen as follows: $A \in (\F_q^d)^{N\dd \times L}$ has pairwise distinct columns, which implies that $A_\fl$ has pairwise distinct columns as well.
    Because $A^\proj_\fl$ is obtained by applying a permutation on the rows of $A_\fl$, the same applies for $A^\proj_\fl$.
    Therefore, for every $i_1,i_2 \in [L]$, there is at least one row in $A^\proj_\fl$ that has differing entries on $i_1$ and $i_2$.
    In particular, this is true for a pair of indices of $[L]$ on which every vector in $W_\cV$ has identical entries, and such a pair of indices exists by virtue of $W_\cV \in (\cL(\F_q^L) \setminus \cL_\dist(\F_q^L))$.
    This implies that there exists a $\ell \in [N]$ and a $j \in [d]$ such that $A^\proj_\fl[(\ell, j)][] \not\in W_\cV$.
    Thus, it is seen that $A^\proj_\psi[(\ell, j)][] \neq \0^{L-\dim W_\cV}$.

    This immediately implies that one of the columns of $A^\proj_\psi$ is non-zero.
    Denote the index of this column by $e \in [L-\dim W_\cV]$.
    Observe that the columns of $A^\proj_\psi$ are codewords in $\cout \circ \cin$.
    Indeed, the columns of $A_\fl$ are (flattened) codewords of $\cael$, and so the columns of $A^\proj_\fl$ are codewords in $\cout \circ \cin$.
    It follows directly from the $\F_q$-linearity of $\cout$ and $\cin$ that $\cout \circ \cin$ is also $\F_q$-linear. Moreover, since $A^\proj_\psi$ is obtained by applying an $\F_q$-linear map to the rows of $A^\proj_\fl$, the columns of $A^\proj_\psi$ are $\F_q$-linear combinations of columns of $A^\proj_\fl$, and therefore are also codewords of $\cout \circ \cin$.
    In particular, this holds for the non-zero column $A^\proj_\psi[][e]$, and therefore is of the following form
    \[
        (\phi(c[1]),\ldots,\phi(c[N]))^\top
    \]
    where $c$ is a codeword in $\cout$.
    This in turn implies that at least $\dout N$ submatrices of the form $A^\proj_\psi(\ell)$ are non-zero, because the  column indexed by $e$ is non-zero for all such submatrices. Denote this set of indices by $S \subseteq [N]$.

    By \cref{clm:many-parts-sat-cvrob}, and the fact that $|S| \ge \dout N$, there is at least one index $\ell^* \in (S \cap L^*)$.
    Since the columns of $A^\proj_\psi$ are codewords of $\cout \circ \cin$, it follows that every column of the submatrix $A^\proj_\psi(\ell^*)$ is a codeword of $\cin$.
    Therefore, $A^\proj_\psi(\ell^*)$ has the following properties:
    \begin{enumerate}
        \item $A^\proj_\psi(\ell^*) \subseteq \cin$,
        \item there exists a $\cV^\imp_\D \in \Rob_\D(\cV^\imp)$ such that $A^\proj_\psi(\ell^*)$ satisfies $\cV_\Dp^\imp$ (by \cref{clm:many-parts-sat-cvrob} and the fact that $\ell^* \in L^*$), and
        \item $A^\proj_\psi(\ell^*)$ is non-zero (as $\ell^* \in S$).
    \end{enumerate}

    This contradicts the fact that $\cin$ is a code satisfying the guarantee mentioned in \cref{fct:good-in-code}.

\end{proof}

\begin{proof}[Proof of \cref{clm:many-parts-sat-cvrob}]
    Since $A$ satisfies $\cV=(f_1, \mM_1),\ldots,(f_{T},\mM_{T}))$, there is a \emph{witness} $M_{N\dd}(\cV) \in \cV_{N\dd}$ such that
    \begin{equation}\label{eq:a-row-sat}
        \forall r \in [N\dd], A[r][] \in \ker M_{r}.
    \end{equation}
    For every $t \in [T]$, define 
    \[
        S_t := \inset{r \in [N\dd] \mid M_{r} = \mM_t}.
    \]
    This is the set of right vertices whose corresponding matrix in $M_{N\dd}(\cV)$ is equal to $\mM_t$. By definition of $\cV$, $\inabs{S_t}=f_t N\dd$ for every $t \in [T]$.
    Denote
    \[
        T_\beta := \inset{t \in [T] \mid f_t > \beta}.
    \]    
    We now set $\beta := \D/(2T_\cP)$, and restrict our attention to elements in $T_\beta$.
    This is possible because the $f_t$ corresponding to $t \not\in T_\beta$ contribute only a small amount of mass. Precisely,
    \begin{equation} \label{eq:not-dense-contrib}
        \sum_{t \in [T] \setminus T_\beta} f_t \le \inabs{[T] \setminus T_\beta}\cdot \frac{\D}{2T_\cP} \le \frac{\D}{2},
    \end{equation}
    where the first inequality holds because $\inabs{[T] \setminus T_\beta} \le T$, and $T \le T_\cP$ by definition.
    
    Fix some $t \in T_\beta$.
    By the guarantee on graph $G$ from \cref{clm:kmrs}, we see that
    \[
        \inabset{\ell \in V_L \mid \inabs{\Gamma(\ell) \cap S_t} < (f_t - \eta)d} \le \zeta N\dd = \frac{\dout}{2T_\cP}\cdot N < \frac{\dout}{T}\cdot N.
    \]
    By applying a union bound argument over all $t \in T_\beta$,
    \begin{equation}\label{eq:good-parts-prop}
        \inabset{\ell \in V_L \mid \forall t \in T_\beta, \inabs{\Gamma(\ell) \cap S_t} \ge \inbrak{f_t-\eta}d} > (1-\dout)N.
    \end{equation}
    Denote this set by $L^* \subseteq [N]$, and note that $\inabs{L^*} > (1-\dout)N$.
    
    We now argue that for every $\ell \in L^*$, the matrix $A_\psi^\proj(\ell)$ satisfies some $\cV_\Dp^\imp \in \Rob_\D(\cV^\imp)$.
    Informally, this can be seen as follows: fix some part $\ell \in L^*$.
    Then, the proportion of incoming edges that are adjacent to some right vertex set $S_t$ is roughly equal to the fractional size of $S_t$, which is equal to $f_t$.
    This implies that the proportion of rows in $A^\proj(\ell)$ that lie in $\ker \mM_t$ is also roughly $f_t$.
    Since $(A^\proj(\ell))[j][] \in \ker \mM_t \implies (A_\psi^\proj(\ell))[j][] \in \ker \mM_t^\imp$,
    we see that for all $t$, roughly $f_t$ fraction of rows of $A_\psi^\proj(\ell)$ lie in $\ker \mM_t^\imp$.
    Because we need to account for the fact that these fractions may not be exact, we see that the matrix satisfies a robust local profile description of $\cV^\imp$.
    
    We now formalize this argument.
    Firstly, note that \cref{eq:a-row-sat} implies the following
    \begin{equation}\label{eq:a-fl-row-sat}
        \forall r \in [N\dd], \forall j \in [d], A_\fl[(r, j)][] \in \ker M_r.
    \end{equation}
    This can be seen by applying the following simple fact to each row of $A$.
    \begin{fact}
        If for vectors $u \in \F_Q^L$ and $v \in \F_q^L$, $\langle u, v \rangle=0$ holds, then $u_\fl \cdot v^\top = \0^d$ holds as well. 
    \end{fact}

    Observe that \cref{eq:a-fl-row-sat}, along with the definition of $\varphi_G$ now implies
    \begin{equation}\label{eq:the-big-one}
        \forall \ell, r \in [N\dd], \forall i, j \in [d], (\varphi_G(\ell, i)=(r,j) \land A_\fl[(r, j)][] \in \ker M_r) \iff A^\proj[(\ell, i)][] \in \ker M_r.
    \end{equation}
    Fix an $\ell \in L^*$.
    From \cref{eq:good-parts-prop,eq:the-big-one}, we see that
    \[
        \forall t \in T_\beta, \inabset{i \in [d] \mid (A^\proj(\ell))[i][] \in \ker \mM_t} \ge \inbrak{f_t - \eta}d.
    \]
    From the definition of $A_\psi^\proj$,
    \begin{equation}\label{eq:roughly-ft-rows}
        \forall t \in T_\beta, \inabset{i \in [d] \mid (A_\psi^\proj(\ell))[i][] \in \ker \mM^\imp_t} \ge \inbrak{f_t - \eta}d.
    \end{equation}
    Furthermore, we will denote the set
    \[
        Z_t := \inset{i \in [d] \mid A_\psi^\proj(\ell)[i][] \in \ker \mM^\imp_t}.
    \]
    Observe that $Z_t \subseteq [d]$, and moreover,
    \begin{equation}\label{eq:size-of-zt}
        \inabs{Z_t} \ge \left\lfloor \inbrak{f_t - \eta}d \right\rfloor.
    \end{equation}
    Because $f_t>\D/(2T_\cP)$ and $d \ge 4T_\cP/\D$, we see that \cref{eq:roughly-ft-rows} implies $\inabs{Z_t} \ge 1$ for all $t \in T_\beta$.

    Now consider
    \[
        \cV^\imp_\Dp := \Bigl((f_t-(\D/(4T_\cP)), \mM_t^\imp)_{t \in T_\beta}, (0, \mM_t^\imp)_{t \not\in T_\beta}, (\inabs{T_\beta}\cdot \D)/4T_\cP + \textstyle{\sum_{t \not\in T_\beta}} f_t, \0)\Bigr).
    \]
    As we have
    \[
        \frac{\inabs{T_\beta}\cdot \D}{4T_\cP} + \sum_{t \not\in T_\beta} f_t \le \frac{\D}{4}+\frac{\D}{2} \le \D,
    \]
    where the first inequality follows from \cref{eq:not-dense-contrib}, we see that $\cV^\imp_\Dp$ belongs to $\Rob_\D(\cV^\imp)$.
    We now prove that $A_\psi^\proj(\ell)$ satisfies $\cV^\imp_\Dp$.
    We do so by creating a local profile $M_{d}(\cV^\imp_\Dp) \in (\cV^\imp_\Dp)_{d}$ and then showing that $A_\psi^\proj(\ell)$ satisfies $\cV^\imp_\Dp$ with $M_{d}(\cV^\imp_\Dp)$ as a witness.
    The local profile $M_{d}(\cV^\imp_\Dp)$ is created as follows:
    \begin{enumerate}[(i)]
        \item For every $t \in T_\beta$, define $\mM^\imp_t$ to be the corresponding matrix for some 
        \[
            \left\lfloor \inbrak{f_t - \frac{\Dp}{4T_\cP}}d \right\rfloor
        \]
        rows indexed by $i \in Z_t$.
        These rows are guaranteed to exist by \cref{eq:size-of-zt}, and the fact that $\eta = \Dp/4T_\cP$.
        \item For the rest of the at most 
        \begin{align*}
            d - \sum_{t \in T_\beta} \left\lfloor \inbrak{f_t - \frac{\Dp}{4T_\cP}}d \right\rfloor
            &\le d - \sum_{t \in T_\beta}  \inbrak{\inbrak{f_t - \frac{\Dp}{4T_\cP}}d -1} \\
            &= d - \sum_{t \in T_\beta} f_t\cdot d + \frac{\inabs{T_\beta} \cdot \D \cdot d}{4T_\cP} + \inabs{T_\beta} \\
            &\le \Bigl( 1-\sum_{t \in T_\beta} f_t \Bigr) \cdot d + \frac{\D \cdot d}{4} + T_\cP \\
            &\le \sum_{t \not\in T_\beta} f_t \cdot d + \frac{\D \cdot d}{2} \\
            & \le \D \cdot d
        \end{align*}
        rows, let the corresponding matrix  be $\0$.
        Note that the last inequality follows from \cref{eq:not-dense-contrib}.
    \end{enumerate}
    It is now seen that $A_\psi^\proj(V_L^p)$ satisfies $\cV^\imp_\Dp$, with $M(\cV^\imp_\Dp)_{d\dd}$ as a witness.
    Indeed, this is seen to follow from \cref{eq:roughly-ft-rows} for the rows corresponding to item (i).
    For the other rows, note that every row in $A_\psi^\proj(\ell)$ belongs to $\psi(\F_q^L) \subseteq \F_q^{(L-\dim W_\cV)}=\ker \0^{(L-\dim W_\cV)}$.

    Thus, we see that for each $\ell \in L^*$, $A_\psi^\proj(\ell)$ satisfies $\cV^\imp_\Dp$.
\end{proof}

\section{Consequences}\label{sec:conseq}
We state an important corollary of \cref{thm:the-big-one}.
\begin{corollary}\label{cor:the-big-one}
    For an $L$-LCL property $\cP$, let $\rin = R_\cP-\eps$ for some $\eps>0$.
    Let $\cin, \cout$, and $G$ be as defined in \cref{thm:the-big-one}.
    Furthermore, let $\cout \subseteq (\F_q^{\rin d\dd})^N$ be an $\F_q$-linear code with rate $\rout=1-\eps$ and distance $\dout\ge \eps^3$.
    Then, $\cael(\cin, \cout, G) \subseteq \F_Q^{N\dd}$ is an $\F_q$-linear LDPC code that \textbf{does not} satisfy $\cP$, has rate $\rael > R_\cP-2\eps$, and 
    \begin{equation}\label{eq:ub-on-d}
        d=O\left(\frac{ L^2 \cdot T_\cP^3}{\eps^5}\right).
    \end{equation}
    Additionally, $\cael(\cin, \cout, G)$ can be constructed in time $\poly(N\dd)$.
\end{corollary}
\begin{proof}
    The $\F_q$-linearity of $\cael$, and the fact that it does not satisfy $\cP$ are established in \cref{thm:the-big-one}.
    The rate is given by
    \[
        \rael=\frac{\log_Q(\inabs{\cael})}{N\dd}.
    \]
    Indeed, it is easy to see that $\inabs{\cael}=\inabs{\cout}=q^{\rin\rout N d \dd}$.
    Recalling that $Q=q^d$, a simple calculation gives $\rael=\rin\cdot \rout=(R_\cP-\eps)(1-\eps)> R_\cP-2\eps$.
    The value for $d$ is obtained by recalling that $d=O(1/\zeta \eta^2)$, $\eta=\Dp/(4T_\cP)$, $\zeta=\dout/2T_\cP, \D=\eps/2L$, and plugging in the value for $\dout$ in $\zeta$.

    We note that explicit constructions of $\F_q$-linear codes having rate $1-\eps$ and distance $\eps^3$ that are constructible in time $\poly (N\dd)$ can be obtained by using Tanner codes (see Corollary 11.4.8 in \cite{GRS23}).
    We note that Tanner codes are LDPC codes.
    Moreover, the graph $G$ can be constructed in time $\poly(N\dd)$, by \cref{clm:kmrs}.
    By \cref{lem:good-in-code-brute-force}, $\cin$ can be constructed in time $q^{\poly(d, L)} \le \poly(N\dd)$.
    Upon invoking \cref{obs:ael-explicit}, we see that $\cael(\cin, \cout, G)$ is an LDPC code that can be constructed in time $\poly(N\dd)$.
\end{proof}

Our constructions also allow for combining multiple LCL properties.
Given a list of $L$-LCL properties $\cP_1, \cP_2, \ldots$ where the constraints are defined over the same field $\Fqz$, we can take the union of the local profile descriptions to form $\cP$:
\[
    \cP := \bigcup_i \cP_i = \bigcup_i \inset{\cV \mid \cV \in \cP_i}.
\]
Since the set $\cP$ is comprised of $L$-local profile descriptions, $\cP$ is an $L$-LCL property as well.
Thus, we can apply \cref{cor:the-big-one} to $\cP$ and obtain an explicit construction that satisfies the complement of $\cP$.
We record this discussion as the following corollary:

\begin{corollary}[Explicit construction for combinations of LCL properties]\label{cor:combine-lcl-prop}
    For $L$-LCL properties $\cP_1, \cP_2, \ldots$, define $\cP$ as above.
    Then, the explicit construction detailed in \cref{cor:the-big-one} for $\cP$ avoids satisfying $\cP_1, \cP_2, \ldots$ simultaneously.
\end{corollary}

\subsection{List Recovery}
We present a general definition of list-recovery, from which the special cases of interest can be derived.

\begin{definition}[List Recovery with Erasures]\label{def:list-rec-eras}
    For integers $\ell, L$, where $\ell \le L$, and constants $0 \le \sigma \le \rho \le 1$, we say that a code $\cC \subseteq \Sigma^n$ is $(\rho, \sigma, \ell, L)$-\emph{average radius list-recoverable with erasures} if the following holds.
    For every collection of sets $S_1,\ldots,S_n \subseteq (\Sigma \cup \bot)$ satisfying
    \begin{enumerate}
        \item $\forall i \in [n], \inabs{S_i} \le \ell$, and
        \item $\exists J \subseteq [n]$ satisfying $|J| \le \sigma n$, such that $S_j=\inset{\bot}$ for all $j \in J$,
    \end{enumerate}
    we have for all pairwise distinct codewords $c_1,\ldots,c_{L+1} \in \cC$:
    \begin{equation} \label{eq:list-rec-cond}
        \frac{\sum_{k \in [L+1]} \sum_{i \in [n] \setminus J}  \mathbbm{1}[c_k[i] \in S_i]}{L+1} \le (1-\rho-\sigma)n.
    \end{equation}

    We say that $\cC$ is $(\rho, \sigma, \ell, L)$-\emph{list-recoverable with erasures} if
    \[
        \min_{k \in [L+1]} \inabset{i \in ([n] \setminus J) \mid c_k[i] \in S_i} \le (1-\rho-\sigma)n.
    \]

    Clearly, a code that is $(\rho, \sigma, \ell, L)$-average radius list-recoverable with erasures is also $(\rho, \sigma, \ell, L)$-list-recoverable with erasures.
\end{definition}

\begin{definition}[List Recoverability]
    For integers $\ell, L$, where $\ell \le L$, and constant $0 \le \rho \le 1$, we say that $\cC$ is $(\rho, \ell, L)$-\emph{average radius list-recoverable} if $\cC$ is $(\rho, 0, \ell, L)$-average radius list-recoverable with erasures.
\end{definition}

\begin{definition}[List Decodability]
    For an integer $L$ and constants $0 \le \sigma \le \rho \le 1$, we say that a code $\cC \subseteq \Sigma^n$ is $(\rho, \sigma, L)$-\emph{average radius erasure list-decodable} if $\cC$ is $(\rho, \sigma, 1, L)$-average radius list-recoverable with erasures.

    We say that $\cC$ is $(\rho, L)$-\emph{average radius list-decodable} if $\cC$ is $(\rho, 1, L)$-average radius list-recoverable.
\end{definition}

\begin{definition}[Zero Error List Recoverability]
    For integers $\ell, L$, where $\ell \le L$, we say that a code $\cC \subseteq \Sigma^n$ is $(\ell, L)$-\emph{zero error list-recoverable} if $\cC$ is $(0, \ell, L)$-average radius list-recoverable.
\end{definition}

\begin{definition}[List Recovery from Erasures]
    For integers $\ell, L$, where $\ell \le L$, and a constant $0 \le \sigma \le 1$, we say that a code $\cC \subseteq \Sigma^n$ is $(\sigma, \ell, L)$-\emph{erasure list-recoverable} if $\cC$ is $(0, \sigma, \ell, L)$-list recoverable with erasures.
\end{definition}

\begin{definition}[Perfect Hash Matrix]
    For integer $t \ge 2$, a code $\cC \subseteq \Sigma^n$ of rate $R$ is defined to be a $(n, \inabs{\Sigma}^{Rn}, t)$-\emph{perfect hash matrix} if it is $\left(0, t-1, t-1\right)$-list-recoverable.
\end{definition}

For integers $\ell, L$, where $\ell \le L$, and constants $0 \le \sigma \le \rho \le 1$, let $\cP(\rho, \sigma, \ell, L)$ be the property of \textbf{not} being $(\rho, \sigma, \ell, L)$-average radius list-recoverable with erasures.
\begin{claim}\label{clm:list-rec-reas}
    Property $\cP(\rho, \sigma, \ell, L)$ is an $(L+1)$-LCL property.
    Additionally, we have
    \[
        \kappa_{q} (\cP(\rho, \sigma, \ell, L)) \le \log_q (\ell+1)^{(L+1)},
    \]
     and
    \[
        T_{\cP(\rho, \sigma, \ell, L)} \le (\ell+1)^{(L+1)}.
    \]
\end{claim}

\begin{proof}
    For convenience, we use the shorthand $\cP$ to denote $\cP(\rho, \sigma, \ell, L)$ throughout this proof.
    In order to prove this claim, we need to construct a set $\cP$ of $(L+1)$-local profile descriptions such that a code $\cC$ is not $(\rho, \sigma, \ell, L)$-average radius list-recoverable with erasures if and only if it contains some local profile description from $\cP$.
    Note that the local profile descriptions we create will depend on the characteristic of the field over which the codes exist, and we will denote that characteristic by $p$.

    For every subset $K \subseteq [L+1]$, denote the set of partitions of $K$ consisting of at most $\ell$ parts by $P_K$.
    For every subset $K \in (2^{[L+1]} \setminus \emptyset)$ and partition $P \in P_K$, create the matrix $\mM(K, P)$, whose rows belong to $\F_p^{L+1}$, and are linear constraints that, for any prime power $q$ of $p$, are satisfied by exactly the vectors in the following set:
    \begin{equation}\label{eq:vecs-sat-mkp}
        \inset{v \in \F_q^{L+1} \mid \forall i, j \in K, i, j\text{ belong to the same part of }P \implies v[i]=v[j]}.
    \end{equation}
    Next, create the matrix $\0$, whose kernel is all of $\F_q^{L+1}$, and associate with it a fraction $f_\0$.

    We denote the fraction associated with the matrix $\mM(K, P)$ by $f_{K, P}$.
    Create a local profile description for every set of associated fractions that satisfy
    \begin{equation}\label{eq:list-rec-sat-1}
        \forall K \in \bigl(2^{[L+1]} \setminus \emptyset\bigr), \forall P \in P_K, 0 \le f_{K, P} \le 1,
    \end{equation}
    \begin{equation}\label{eq:list-rec-sat-2}
        f_\0 \ge \sigma,
    \end{equation}
    \begin{equation}\label{eq:list-rec-sat-3}
        \sum_{K \in (2^{[L+1]} \setminus \emptyset)} \sum_{P \in P_K} |K|\cdot f_{K, P} > (1-\rho-\sigma)(L+1),
    \end{equation}
    and
    \begin{equation}\label{eq:list-rec-sat-4}
        f_\0 + \sum_{K \in (2^{[L+1]} \setminus \emptyset)} \sum_{P \in P_K} f_{K, P} = 1.
    \end{equation}
    Lastly, we collect these local profile descriptions in a set $\cP$.

    We first prove that if a code $\cC \subseteq \F_q^n$ contains some
    \[
        \cV=\bigl((f_{K, P_K},\mM(K, P_K))_{K \in (2^{[L+1]} \setminus \emptyset), P \in P_K}, (f_\0, \0)\bigr) \in \cP,
    \]
    then it is \textbf{not} $(\rho, \sigma, \ell, L)$-average radius list-recoverable with erasures.
    The code $\cC$ containing a $\cV \in \cP$ implies the existence of a matrix $A \in \F_q^{n \times (L+1)}$ such that
    \begin{enumerate}
        \item $A \subseteq \cC$,
        \item $A$ satisfies $\cV$, and
        \item $A$ has pairwise distinct columns.
    \end{enumerate}
    By (1) and (3), the columns of $A$ are pairwise distinct codewords in $\cC$.
    By (2), we know that for every non-empty $K$ and partition $P \in P_K$, there are a $f_{K, P_K}$ fraction of rows in $A$ whose entries agree according to the constraints set forth by $\mM(K, P_K)$.
    The constraints say that the number of distinct elements appearing in the entries specified by $K$ is at most $\ell$.
    Thus for each coordinate $i \in [n]$ corresponding to the matrix $\mM(K, P_K)$, we can create a subset $S_i \subseteq \F_q$ such that $\inabs{S_i} \le \ell$, and for all $k \in K$, we have $A[i][k] \in S_i$.
    According to \cref{eq:list-rec-sat-4}, such a subset $S_i$ can be created for $(1-f_\0)n$ coordinates.
    For the remaining $f_\0 n$ coordinates, we create the set $\inset{\bot}$.
    Denote this set of coordinates by $J$.
    According to \cref{eq:list-rec-sat-3},
    \[
        \sum_{i \in [n] \setminus J} \sum_{k \in [L+1]} \mathbbm{1} [A[i][k] \in S_i] \ge \sum_{K \in (2^{[L+1]} \setminus \emptyset)} \sum_{P \in P_K} |K|\cdot f_{K, P} \cdot n > (1-\rho-\sigma)(L+1)n.
    \]
    Therefore, we see that for a set of $L+1$ pairwise distinct codewords in $\cC$, \cref{eq:list-rec-cond} is not satisfied, hence $\cC$ is \textbf{not} $(\rho, \sigma, \ell, L)$-average radius list-recoverable with erasures.

    We now prove the other direction.
    Suppose that the code $\cC \subseteq \F_q^n$ is \textbf{not} $(\rho, \sigma, \ell, L)$-average radius list-recoverable with erasures.
    Then, there exists a collection of sets $S_1, \ldots, S_n \subseteq (\F_q \cup \bot)$ satisfying
    \begin{enumerate}
        \item $\forall i \in [n], \inabs{S_i} \le \ell$, and
        \item $\exists J \subseteq [n]$, satisfying $\inabs{J} \le \sigma n$, such that $S_j = \inset{\bot}$ for all $j \in J$.
    \end{enumerate}
    Furthermore, there exists a matrix $A \in \F_q^{n \times (L+1)}$ satisfying
    \begin{enumerate}
        \item $A \subseteq \cC$,
        \item $A$ has pairwise distinct columns, and
        \item $\sum_{i \in [n] \setminus J} \sum_{k \in [L+1]} \mathbbm{1} [A[i][k] \in S_i] > (1-\rho-\sigma)(L+1)n$.
    \end{enumerate}
    For every coordinate $i \in [n] \setminus J$, we can assign $i$ a type $(K, P)$, where $K$ is a subset of $[L+1]$ and $P$ is some partition in $P_K$.
    The type assigned to $i$ will be $(K=\inset{k \in [L+1] \mid A[i][k] \in S_i}, P)$, where $P$ is a partition of $K$ such that for every $k_1, k_2$ belonging to the same part, $A[i][k_1]=A[i][k_2]$ holds.
    We now create a matrix $\mM(K, P)$ for each type $(K, P)$, whose entries belong to $\F_p$ and whose rows consist of linear constraints that are satisfied by vectors in $\F_q^{L+1}$ that belong to the set described in \cref{eq:vecs-sat-mkp}.
    Clearly, row $A[i][]$ satisfies the constraints of $\mM(K, P)$ if its type is $(K, P)$.
    The fraction of rows having type $(K, P)$ is denoted by $f_{K, P}$.
    Therefore, from (3), we see that 
    \begin{equation} \label{eq:rows-of-a-agree}
        \sum_{K \in (2^{[L+1]} \setminus \emptyset)} \sum_{i \in [n] \setminus J} K \cdot \mathbbm{1}[\inset{k \in [L+1] \mid A[i][k] \in S_i}=K] > (1-\rho-\sigma)(L+1)n.
    \end{equation}
    We also see that if the equality $\inset{k \in [L+1] \mid A[i][k] \in S_i}=K$ is satisfied for some coordinate $i$ and $K$, then there is a matrix $\mM(K, P)$ for some $P \in P_K$ such that row $A[i][]$ satisfies the constraints set by $\mM(K, P)$.
    Thus, we see that \cref{eq:rows-of-a-agree}, upon being divided by the block length $n$, is exactly the same as \cref{eq:list-rec-sat-3}.

    For those coordinates that are present in $J$, and also those coordinates for which $K=\emptyset$, we assign the matrix $\0$, and note that $f_\0$ is at least $\sigma$.
    We then see that \cref{eq:list-rec-sat-1}, \cref{eq:list-rec-sat-2}, and \cref{eq:list-rec-sat-4} are also satisfied, and thus $A$ satisfies a local profile description that belongs to $\cP$. 

    Lastly, we prove the upper bounds on $\kappa_q(\cP)$ and $T_\cP$.
    This is seen by observing that the number of matrices of the form $\mM(K, P)$, as constructed above, is at most $(\ell+1)^{(L+1)}$, which is an upper bound on the number of partitions consisting of at most $\ell$ parts, of every subset of $[L+1]$.
    Thus, $T_\cP \le (\ell+1)^{(L+1)}$, and the number of local profiles associated with this property is at most
    \[
        (\ell+1)^{(L+1)n} = q^{\log_q (\ell+1)^{(L+1)}\cdot n}
    \]
    and we see that $\kappa_q(\cP) \le \log_q (\ell+1)^{(L+1)}$.    
\end{proof}

We then obtain the following corollary.
\begin{corollary}\label{cor:list-rec-params}
    For $\cP(\rho, \sigma, \ell, L)$, let $\cin, \cout, G$ be as defined in \cref{thm:the-big-one}.
    Furthermore, let $\cout \subseteq (\F_q^{\rin d\dd})^N$ be a code with rate $\rout=1-\eps$ and distance $\dout\ge \eps^3$.
    Then, $\cael(\cin, \cout, G) \subseteq \F_Q^{N\dd}$ is a $\F_q$-linear code that is $(\rho, \sigma, \ell, L)$-average radius list-recoverable with erasures, has rate $\rael > R_\cP-2\eps$, and can be constructed in time $\poly(N)$.
    Additionally, we have $q \le (\ell+1)^{(L+1)\cdot \frac{8}{\eps}}$, and $Q=q^d$, where
    \[
        d =  O\inbrak{\frac{L^2 \cdot (\ell+1)^{3(L+1)}}{\eps^5}}.
    \]
\end{corollary}
\begin{proof}
    We set a value for $q$ so that \cref{eq:cond-for-reas} is satisfied.
    For $q=(\ell+1)^{(L+1)\cdot \frac{8}{\eps}}$, we see
    \begin{align*}
        \log_q (\ell+1)^{(L+1)}+ \log_q 2 - \frac{\eps}{4} = \frac{\eps}{8} + \frac{1}{\log_2 q} - \frac{\eps}{4}<0
    \end{align*}
    Therefore, \cref{cor:the-big-one} implies that $\cael$ does not satisfy $\cP(\rho, \sigma, \ell, L)$, and thus, is $(\rho, \sigma, \ell, L)$-average radius list-recoverable with erasures.

    Finally, \cref{eq:ub-on-d} from \cref{cor:the-big-one}, along with the value of $T_{\cP(\rho, \sigma, \ell, L)}$ from \cref{clm:list-rec-reas} implies
    \[
        d=O\inbrak{\frac{L^2 \cdot (\ell+1)^{3(L+1)}}{\eps^5}}.
    \]
    Plugging the value into $Q=q^d$ gives us the final alphabet size.
\end{proof}

We now record several results, corresponding to the special cases defined above.
The proofs for them follow from \cref{cor:list-rec-params}.
The first details parameters for list recovery at capacity.

\begin{corollary}[Explicitly Achieving Capacity for List Recovery]\label{cor:cap-ach-list-rec}
    For a fixed input list size $\ell$, and a fixed radius $\rho>0$, let $L_{\rho, \ell}$ be the smallest output list size such that $R_{\cP(\rho, 0 , \ell, L_{\rho, \ell})} \ge 1-\rho$.
    For $\eps>0$, let $R:=1-\rho-\eps \le R_{\cP(\rho, 0 , \ell, L_{\rho, \ell})}-\eps$.
    The code $\cael$ as in \cref{cor:list-rec-params} is $(1-R-\eps, \ell, L_{\rho, \ell})$-average radius list-recoverable, with rate $\rael > R_{\cP(\rho, 0 , \ell, L_{\rho, \ell})}-2\eps$, and alphabet size at most $\exp \bigl( \inbrak{L_{\rho, \ell}/\eps}^{O(L_{\rho, \ell})} \bigr)$.
\end{corollary}

\begin{corollary}[Explicit Zero Error List Recovery]\label{cor:exp-zero-err-list-rec}
    For $\eps>0$ and a fixed input list size $\ell$ and rate $R$, let $L_{R, \ell}$ be the smallest output list size such that $R_{\cP(0, 0, \ell, L_{R, \ell})} -\eps \ge R$.
    Then, the code $\cael$ as in \cref{cor:list-rec-params} is $(\ell, L_{R, \ell})$-zero error list recoverable, with rate $\rael > R_{\cP(0, 0, \ell, L_{R, \ell})} -2\eps$, and alphabet size at most $\exp \bigl( \inbrak{L_{R, \ell}/\eps}^{O(L_{R, \ell})} \bigr)$.
\end{corollary}

\begin{corollary}[Explicit List Recovery from Erasures]\label{cor:exp-eras-list-rec}
    For constants $\eps>0$, $0 \le \sigma \le 1$, a fixed input list size $\ell$ and rate $R$, let $L_{\sigma, R, \ell}$ be the smallest output list size such that $R_{\cP(0, \sigma, \ell, L_{\sigma, R, \ell})} -\eps \ge R$.
    Then, the code $\cael$ as in \cref{cor:list-rec-params} is $(\sigma, \ell, L_{\sigma, R, \ell})$-erasure list recoverable, with rate $\rael > R_{\cP(0, \sigma, \ell, L_{\sigma, R, \ell})} -2\eps$, and alphabet size at most $\exp \bigl( \inbrak{L_{\sigma, R, \ell}/\eps}^{O(L_{\sigma, R, \ell})} \bigr)$.
\end{corollary}

\begin{corollary}[Explicit Perfect Hash matrices]\label{cor:exp-per-hash-mat}
    For an integer $t \ge 2$, let $\cael$ be the code as in \cref{cor:list-rec-params} for the property $\cP(0, 0, t-1, t-1)$.
    Then, $\cael$ is a code of rate $R=R_{\cP(0, 0, t-1, t-1)}-2\eps$ that is $(0, t-1, t-1)$-list recoverable.
    The alphabet size is at most $\exp((t/\eps)^{O(t)})$ and moreover, the codewords of $\cael$, when arranged as columns in a $N\dd \times Q^{R N}$ matrix, form a perfect hash matrix.
\end{corollary}
\begin{proof}
    The list-recoverability of $\cael$ follows from \cref{cor:list-rec-params}.
    We prove that the set of codewords of a $(0, t-1, t-1)$-list-recoverable code can be arranged as columns in a matrix to form a perfect hash matrix.
    Indeed, $(0, t-1, t-1)$-list recoverability implies that the number of pairwise distinct codewords having at most $t-1$ distinct entries in every row is at most $t-1$. This implies that for any set of $t$ codewords, there exists at least one index on which the $t$ codewords have pairwise distinct entries.
\end{proof}

\subsection{Efficient Decoding Algorithms}

We remark our derandomization of linear codes admit efficient (list) decoding algorithms via the recently developed efficient (list) decoding algorithms for AEL~\cite{JMST25,JS25,ST25}. There are two kinds of such decoders: one based on weak-regularity lemmas~\cite{JS25,ST25}, and another based on the Sum-of-Squares SDP hierarchy~\cite{JMST25}. The former runs in near-linear time in the block length, whereas the latter has a larger polynomial running time. The crucial property needed by these AEL decoders is that the expander graph used in the construction be sufficiently expanding, namely, sufficiently small second largest singular value $\lambda$ of their (bi)adjacency matrix.

For the sake of completeness, we recall the main decoding results of~\cite{JS25} and~\cite{ST25} and show how they can be applied to our derandomization of linear codes via the  AEL procedure to yield near-liner time list decoding and list recovery algorithms. We start with the list-decoding 
case.

\begin{theorem}[Implicit version of Theorem 4.1 from~\cite{JS25}]    
    Let $\cin$ be a $(\delzero, L_{\textnormal{in}})$-list-decodable code with block length $d$ and $C_{out}$ be uniquely decodable from $\delone$ fraction of errors, having block length $N$.
    Then the \text{AEL} code $\cael=C(\cin, \cout, G)$ is $(\rho-\varepsilon, L)$-list-decodable for any constant $0<\varepsilon<\rho$, if $\tfrac{\lambda}{d} < O\left(\tfrac{\varepsilon^2\delone^2}{L_{in}^4}\right)$, where $L$ is an upper bound on the list size of $\cael$.
    Furthermore, if $C_{out}$ is uniquely decodable from $\delone$ fraction of errors in time $\mathcal{T}$, then the decoding can be performed in $\widetilde{\cal{O}}_{\varepsilon,\ell}(N+\mathcal{T})$ time.
\end{theorem}

Let $\cP$ be an $(L+1)$-LCL property that is the complement of being list-decodable at the $\eps$-relaxed Generalized Singleton Bound.
\cref{lem:good-in-code-brute-force} guarantees that a code $\cin$ satisfying the complement of $\cP$ can be found in time that is constant in the block length $N$.
For $C_{out}$, we choose the same code family from~\cite{JS25}, dating back to~\cite{GI05}, that has linear time decoders. As a consequence, we obtain near-linear time list-decodable codes up to the $\varepsilon$-relaxed singled codes.
Moreover, by \cref{cor:combine-lcl-prop}, $\cael$ can also satisfy the complement of any number of additional LCL properties.

We now turn our attention to efficient list-recovery.

\begin{theorem}[Implicit version of Theorem 1.3 from~\cite{ST25}] 
    Let $\cin$ be a $(\rho, \ell, L_{\textnormal{in}})$-list-recoverable code with block length $d$ and $\cout$ be a code with block length $N$ that is uniquely decodable from $\delone$ fraction of errors.
    Then the \text{AEL} code $\cael=C(\cin, \cout, G)$ is $(\rho-\varepsilon, \ell, L)$-list-recoverable for any constant $0<\varepsilon<\rho$, if $\tfrac{\lambda}{d} < O\left(\tfrac{\varepsilon^2\delone^2}{(\ell \cdot L_{\textnormal{in}})^2}\right)$, where $L$ is an upper bound on the output list size of $\cael$. Furthermore, if $\cout$ is uniquely decodable from $\delone$ fraction of errors in time $\mathcal{T}$, then the decoding can be performed in $\widetilde{\cal{O}}_{\varepsilon,\ell,L_{in}}(n+\mathcal{T})$ time.
\end{theorem}

Similar to the above theorem, we can take $\cin$ to be a code guaranteed by \cref{lem:good-in-code-brute-force}, that achieves the same list-size tradeoffs as random linear codes. The choice of $\cout$ is the same as above.
As a consequence, we obtain near-linear time list-recoverable codes up to the optimal~\cite{BCDZ25b} bound.
Moreover, by \cref{cor:combine-lcl-prop}, $\cael$ can also satisfy the complement of any number of additional LCL properties.

\section{Acknowledgements}
We thank Jonathan Mosheiff for providing feedback and comments on a draft of the paper, and also thank Yeyuan Chen for providing comments on an earlier version of the paper.

\bibliographystyle{alpha}
\bibliography{macros,references}

\newcommand{\etalchar}[1]{$^{#1}$}
\begin{thebibliography}{MRRZ{\etalchar{+}}19}

\bibitem[AD24]{AD24}
Ron Asherov and Irit Dinur.
\newblock Bipartite unique neighbour expanders via ramanujan graphs.
\newblock {\em Entropy}, 26(4):348, 2024.

\bibitem[AEL95]{AEL95}
N.~Alon, J.~Edmonds, and M.~Luby.
\newblock Linear time erasure codes with nearly optimal recovery.
\newblock In {\em Proceedings of IEEE 36th Annual Foundations of Computer
  Science}, pages 512--519, 1995.

\bibitem[AGL24]{AGL24}
Omar Alrabiah, Venkatesan Guruswami, and Ray Li.
\newblock Randomly punctured reed-solomon codes achieve list-decoding capacity
  over linear-sized fields.
\newblock In Bojan Mohar, Igor Shinkar, and Ryan O'Donnell, editors, {\em
  Proceedings of the 56th Annual {ACM} Symposium on Theory of Computing, {STOC}
  2024, Vancouver, BC, Canada, June 24-28, 2024}, pages 1458--1469. {ACM},
  2024.

\bibitem[AN96]{AN96}
Noga Alon and Moni Naor.
\newblock Derandomization, witnesses for boolean matrix multiplication and
  construction of perfect hash functions.
\newblock {\em Algorithmica}, 16(4/5):434--449, 1996.

\bibitem[{Ari}09]{Arikan08}
E.~{Arikan}.
\newblock Channel polarization: A method for constructing capacity-achieving
  codes for symmetric binary-input memoryless channels.
\newblock {\em IEEE Transactions on Information Theory}, 55(7):3051--3073, July
  2009.

\bibitem[BCDZ25a]{BCDZ25b}
Joshua Brakensiek, Yeyuan Chen, Manik Dhar, and Zihan Zhang.
\newblock Combinatorial bounds for list recovery via discrete brascamp--lieb
  inequalities, 2025.

\bibitem[BCDZ25b]{BCDZ25a}
Joshua Brakensiek, Yeyuan Chen, Manik Dhar, and Zihan Zhang.
\newblock From random to explicit via subspace designs with applications to
  local properties and matroids, 2025.

\bibitem[BGM23]{BGM23}
Joshua Brakensiek, Sivakanth Gopi, and Visu Makam.
\newblock Generic reed-solomon codes achieve list-decoding capacity.
\newblock In {\em Proceedings of the 55th Annual ACM Symposium on Theory of
  Computing}, pages 1488--1501, 2023.

\bibitem[Bla00]{Bla00}
Simon~R. Blackburn.
\newblock Perfect hash families: Probabilistic methods and explicit
  constructions.
\newblock {\em J. Comb. Theory {A}}, 92(1):54--60, 2000.

\bibitem[BLVW19]{BLVW19}
Zvika Brakerski, Vadim Lyubashevsky, Vinod Vaikuntanathan, and Daniel Wichs.
\newblock Worst-case hardness for lpn and cryptographic hashing via code
  smoothing.
\newblock In {\em Advances in Cryptology – EUROCRYPT 2019}, 2019.

\bibitem[BW98]{BW98}
Simon~R. Blackburn and Peter~R. Wild.
\newblock Optimal linear perfect hash families.
\newblock {\em J. Comb. Theory {A}}, 83(2):233--250, 1998.

\bibitem[CCS25]{CCS25}
Yeyuan Chen, Mahdi Cheraghchi, and Nikhil Shagrithaya.
\newblock Optimal erasure codes and codes on graphs.
\newblock {\em CoRR}, abs/2504.03090, 2025.

\bibitem[Che25]{Che25}
Yeyuan Chen.
\newblock Unique-neighbor expanders with better expansion for polynomial-sized
  sets.
\newblock In Yossi Azar and Debmalya Panigrahi, editors, {\em Proceedings of
  the 2025 Annual {ACM-SIAM} Symposium on Discrete Algorithms, {SODA} 2025, New
  Orleans, LA, USA, January 12-15, 2025}, pages 3335--3362. {SIAM}, 2025.

\bibitem[CZ25]{CZ25}
Yeyuan Chen and Zihan Zhang.
\newblock {Explicit Folded Reed-Solomon and Multiplicity Codes Achieve Relaxed
  Generalized Singleton Bounds}.
\newblock In {\em Proceedings of the 57th ACM Symposium on Theory of
  Computing}, 2025.

\bibitem[DEL{\etalchar{+}}22]{DELL22}
Irit Dinur, Shai Evra, Ron Livne, Alexander Lubotzky, and Shahar Mozes.
\newblock Locally testable codes with constant rate, distance, and locality.
\newblock In Stefano Leonardi and Anupam Gupta, editors, {\em {STOC} '22: 54th
  Annual {ACM} {SIGACT} Symposium on Theory of Computing, Rome, Italy, June 20
  - 24, 2022}, pages 357--374. {ACM}, 2022.

\bibitem[DMS03]{DMS03}
Ilya Dumer, Daniele Micciancio, , and Madhu Sudan.
\newblock Hardness of approximating the minimum distance of a linear code.
\newblock {\em {IEEE} Transactions on Information Theory}, 49(1):22--37, 2003.

\bibitem[DW22]{DW22}
Dean Doron and Mary Wootters.
\newblock High-probability list-recovery, and applications to heavy hitters.
\newblock In Mikolaj Bojanczyk, Emanuela Merelli, and David~P. Woodruff,
  editors, {\em 49th International Colloquium on Automata, Languages, and
  Programming, {ICALP} 2022, July 4-8, 2022, Paris, France}, volume 229 of {\em
  LIPIcs}, pages 55:1--55:17. Schloss Dagstuhl - Leibniz-Zentrum f{\"{u}}r
  Informatik, 2022.

\bibitem[FKS82]{FKS82}
Michael~L. Fredman, J{\'{a}}nos Koml{\'{o}}s, and Endre Szemer{\'{e}}di.
\newblock Storing a sparse table with {O(1)} worst case access time.
\newblock In {\em 23rd Annual Symposium on Foundations of Computer Science,
  Chicago, Illinois, USA, 3-5 November 1982}, pages 165--169. {IEEE} Computer
  Society, 1982.

\bibitem[FM04]{FM04}
Uriel Feige and Daniele Micciancio.
\newblock The inapproximability of lattice and coding problems with
  preprocessing.
\newblock {\em J. Comput. Syst. Sci.}, 69(1):45--67, 2004.

\bibitem[GHK11]{GHK11}
Venkatesan Guruswami, Johan H{\aa}stad, and Swastik Kopparty.
\newblock On the list-decodability of random linear codes.
\newblock {\em {IEEE} Trans. Inf. Theory}, 57(2):718--725, 2011.

\bibitem[GHSZ02]{GHSZ02}
Venkatesan Guruswami, Johan H{\aa}stad, Madhu Sudan, and David Zuckerman.
\newblock Combinatorial bounds for list decoding.
\newblock {\em {IEEE} Trans. Inf. Theory}, 48(5):1021--1034, 2002.

\bibitem[GI02]{GI02}
Venkatesan Guruswami and Piotr Indyk.
\newblock Near-optimal linear-time codes for unique decoding and new
  list-decodable codes over small alphabets.
\newblock In {\em Proceedings of the 34th ACM Symposium on Theory of
  Computing}, pages 812--821, 2002.

\bibitem[GI03]{GI03}
Venkatesan Guruswami and Piotr Indyk.
\newblock Linear time encodable and list decodable codes.
\newblock In {\em Proceedings of the 35th ACM Symposium on Theory of
  Computing}, 2003.

\bibitem[GI05]{GI05}
Venkatesan Guruswami and Piotr Indyk.
\newblock Linear-time encodable/decodable codes with near-optimal rate.
\newblock {\em {IEEE} Trans. Inf. Theory}, 51(10):3393--3400, 2005.

\bibitem[Gil52]{G52}
E.N. Gilbert.
\newblock A comparison of signalling alphabets.
\newblock {\em Bell System Technical Journal}, 31:504--522, 1952.

\bibitem[GL89]{GoldreichL89}
O.~Goldreich and L.~Levin.
\newblock A hard-core predicate for all one-way functions.
\newblock In {\em Proceedings of the 21st ACM Symposium on Theory of
  Computing}, pages 25--32, 1989.

\bibitem[GLM{\etalchar{+}}22]{GLMRSW22}
Venkatesan Guruswami, Ray Li, Jonathan Mosheiff, Nicolas Resch, Shashwat Silas,
  and Mary Wootters.
\newblock Bounds for list-decoding and list-recovery of random linear codes.
\newblock {\em {IEEE} Trans. Inf. Theory}, 68(2):923--939, 2022.

\bibitem[GLS{\etalchar{+}}24]{GLS24}
Zeyu Guo, Ray Li, Chong Shangguan, Itzhak Tamo, and Mary Wootters.
\newblock Improved list-decodability and list-recoverability of reed--solomon
  codes via tree packings.
\newblock {\em SIAM Journal on Computing}, 53(2):389--430, 2024.

\bibitem[GM22]{GM22}
Venkatesan Guruswami and Jonathan Mosheiff.
\newblock Punctured low-bias codes behave like random linear codes.
\newblock In {\em 63rd {IEEE} Annual Symposium on Foundations of Computer
  Science, {FOCS} 2022, Denver, CO, USA, October 31 - November 3, 2022}, pages
  36--45. {IEEE}, 2022.

\bibitem[GMR{\etalchar{+}}22]{GMRSW22}
Venkatesan Guruswami, Jonathan Mosheiff, Nicolas Resch, Shashwat Silas, and
  Mary Wootters.
\newblock Threshold rates for properties of random codes.
\newblock {\em {IEEE} Trans. Inf. Theory}, 68(2):905--922, 2022.

\bibitem[Gol24]{Gol24}
Louis Golowich.
\newblock New explicit constant-degree lossless expanders.
\newblock In David~P. Woodruff, editor, {\em Proceedings of the 2024 {ACM-SIAM}
  Symposium on Discrete Algorithms, {SODA} 2024, Alexandria, VA, USA, January
  7-10, 2024}, pages 4963--4971. {SIAM}, 2024.

\bibitem[GR06a]{GuruswamiR06}
Venkatesan Guruswami and Atri Rudra.
\newblock Explicit capacity-achieving list-decodable codes.
\newblock In {\em Proceedings of the 38th ACM Symposium on Theory of
  Computing}, pages 1--10, 2006.

\bibitem[GR06b]{GR06}
Venkatesan Guruswami and Atri Rudra.
\newblock Explicit capacity-achieving list-decodable codes.
\newblock In Jon~M. Kleinberg, editor, {\em Proceedings of the 38th Annual
  {ACM} Symposium on Theory of Computing, Seattle, WA, USA, May 21-23, 2006},
  pages 1--10. {ACM}, 2006.

\bibitem[GRS23]{GRS23}
Venkatesan Guruswami, Atri Rudra, and Madhu Sudan.
\newblock Essential coding theory.
\newblock Book, 2023.

\bibitem[GS98]{GS98}
Venkatesan Guruswami and Madhu Sudan.
\newblock Improved decoding of reed-solomon and algebraic-geometric codes.
\newblock In {\em 39th Annual Symposium on Foundations of Computer Science,
  {FOCS} 1998, Palo Alto, California, USA, November 8-11, 1998}, pages 28--39.
  {IEEE} Computer Society, 1998.

\bibitem[GUV09]{GUV09}
Venkatesan Guruswami, Christopher Umans, and Salil~P. Vadhan.
\newblock Unbalanced expanders and randomness extractors from parvaresh-vardy
  codes.
\newblock {\em J. {ACM}}, 56(4):20:1--20:34, 2009.

\bibitem[GZ23]{GZ23}
Zeyu Guo and Zihan Zhang.
\newblock Randomly punctured reed-solomon codes achieve the list decoding
  capacity over polynomial-size alphabets.
\newblock In {\em Proceedings of the 64rd IEEE Symposium on Foundations of
  Computer Science}, 2023.

\bibitem[HLM{\etalchar{+}}25a]{HLMO25}
Jun{-}Ting Hsieh, Ting{-}Chun Lin, Sidhanth Mohanty, Ryan O'Donnell, and
  Rachel~Yun Zhang.
\newblock Explicit two-sided vertex expanders beyond the spectral barrier.
\newblock In Michal Kouck{\'{y}} and Nikhil Bansal, editors, {\em Proceedings
  of the 57th Annual {ACM} Symposium on Theory of Computing, {STOC} 2025,
  Prague, Czechia, June 23-27, 2025}, pages 833--842. {ACM}, 2025.

\bibitem[HLM{\etalchar{+}}25b]{HLMR25}
Jun{-}Ting Hsieh, Alexander Lubotzky, Sidhanth Mohanty, Assaf Reiner, and
  Rachel~Yun Zhang.
\newblock Explicit lossless vertex expanders.
\newblock {\em CoRR}, abs/2504.15087, 2025.

\bibitem[HMMP24]{HMMP24}
Jun{-}Ting Hsieh, Theo McKenzie, Sidhanth Mohanty, and Pedro Paredes.
\newblock Explicit two-sided unique-neighbor expanders.
\newblock In Bojan Mohar, Igor Shinkar, and Ryan O'Donnell, editors, {\em
  Proceedings of the 56th Annual {ACM} Symposium on Theory of Computing, {STOC}
  2024, Vancouver, BC, Canada, June 24-28, 2024}, pages 788--799. {ACM}, 2024.

\bibitem[HRW20]{HRW20}
Brett Hemenway, Noga Ron{-}Zewi, and Mary Wootters.
\newblock Local list recovery of high-rate tensor codes and applications.
\newblock {\em {SIAM} J. Comput.}, 49(4), 2020.

\bibitem[JMST25]{JMST25}
Fernando~Granha Jeronimo, Tushant Mittal, Shashank Srivastava, and Madhur
  Tulsiani.
\newblock Explicit codes approaching generalized singleton bound using
  expanders.
\newblock In {\em Proceedings of the 57th ACM Symposium on Theory of
  Computing}, 2025.

\bibitem[JS25]{JS25}
Fernando~Granha Jeronimo and Aman Singh.
\newblock List decoding expander-based codes via fast approximation of
  expanding csps: I, 2025.

\bibitem[KMRS17]{KMRS17}
Swastik Kopparty, Or~Meir, Noga Ron{-}Zewi, and Shubhangi Saraf.
\newblock High-rate locally correctable and locally testable codes with
  sub-polynomial query complexity.
\newblock {\em J. {ACM}}, 64(2):11:1--11:42, 2017.

\bibitem[KRSW23]{KRZSW23}
Swastik Kopparty, Noga Ron{-}Zewi, Shubhangi Saraf, and Mary Wootters.
\newblock Improved list decoding of folded reed-solomon and multiplicity codes.
\newblock {\em {SIAM} J. Comput.}, 2023.

\bibitem[LMS25]{LMS25}
Matan Levi, Jonathan Mosheiff, and Nikhil Shagrithaya.
\newblock Random reed-solomon codes and random linear codes are locally
  equivalent, 2025.

\bibitem[LNNT19]{LNNT19}
Kasper~Green Larsen, Jelani Nelson, Huy~L. Nguyen, and Mikkel Thorup.
\newblock Heavy hitters via cluster-preserving clustering.
\newblock {\em Commun. {ACM}}, 62(8):95--100, 2019.

\bibitem[LP20]{LP20}
Ben Lund and Aditya Potukuchi.
\newblock On the list recoverability of randomly punctured codes.
\newblock In Jaroslaw Byrka and Raghu Meka, editors, {\em Approximation,
  Randomization, and Combinatorial Optimization. Algorithms and Techniques,
  {APPROX/RANDOM} 2020, August 17-19, 2020, Virtual Conference}, volume 176 of
  {\em LIPIcs}, pages 30:1--30:11. Schloss Dagstuhl - Leibniz-Zentrum f{\"{u}}r
  Informatik, 2020.

\bibitem[LPB06]{LPB06}
Yi~Lu, Balaji Prabhakar, and Flavio Bonomi.
\newblock Perfect hashing for network applications.
\newblock In {\em Proceedings 2006 {IEEE} International Symposium on
  Information Theory, {ISIT} 2006, The Westin Seattle, Seattle, Washington,
  USA, July 9-14, 2006}, pages 2774--2778. {IEEE}, 2006.

\bibitem[LS25]{LS25}
Ray Li and Nikhil Shagrithaya.
\newblock Near-optimal list-recovery of linear code families.
\newblock {\em CoRR}, abs/2502.13877, 2025.

\bibitem[Meh84]{Mel84}
Kurt Mehlhorn.
\newblock {\em Data Structures and Algorithms 1: Sorting and Searching},
  volume~1 of {\em {EATCS} Monographs on Theoretical Computer Science}.
\newblock Springer, 1984.

\bibitem[MRR{\etalchar{+}}20]{MRRSW20}
Jonathan Mosheiff, Nicolas Resch, Noga Ron{-}Zewi, Shashwat Silas, and Mary
  Wootters.
\newblock {LDPC} codes achieve list decoding capacity.
\newblock In Sandy Irani, editor, {\em 61st {IEEE} Annual Symposium on
  Foundations of Computer Science, {FOCS} 2020, Durham, NC, USA, November
  16-19, 2020}, pages 458--469. {IEEE}, 2020.

\bibitem[MRRZ{\etalchar{+}}19]{MosheiffRRSW19}
Jonathan Mosheiff, Nicolas Resch, Noga Ron-Zewi, Shashwat Silas, and Mary
  Wootters.
\newblock {LDPC} codes achieve list decoding capacity, 2019.

\bibitem[NPR12]{NPR12}
Hung~Q. Ngo, Ely Porat, and Atri Rudra.
\newblock Efficiently decodable compressed sensing by list-recoverable codes
  and recursion.
\newblock In Christoph D{\"{u}}rr and Thomas Wilke, editors, {\em 29th
  International Symposium on Theoretical Aspects of Computer Science, {STACS}
  2012, February 29th - March 3rd, 2012, Paris, France}, volume~14 of {\em
  LIPIcs}, pages 230--241. Schloss Dagstuhl - Leibniz-Zentrum f{\"{u}}r
  Informatik, 2012.

\bibitem[NW94]{NW94}
Noam Nisan and Avi Wigderson.
\newblock Hardness vs randomness.
\newblock {\em Journal of Computer and System Sciences}, 49:149--167, 1994.
\newblock Preliminary version in {\em Proc. of FOCS'88}.

\bibitem[NW95]{NW95}
Ilan Newman and Avi Wigderson.
\newblock Lower bounds on formula size of boolean functions using hypergraph
  entropy.
\newblock {\em {SIAM} J. Discret. Math.}, 8(4):536--542, 1995.

\bibitem[PK22]{PK22}
Pavel Panteleev and Gleb Kalachev.
\newblock Asymptotically good quantum and locally testable classical {LDPC}
  codes.
\newblock In Stefano Leonardi and Anupam Gupta, editors, {\em {STOC} '22: 54th
  Annual {ACM} {SIGACT} Symposium on Theory of Computing, Rome, Italy, June 20
  - 24, 2022}, pages 375--388. {ACM}, 2022.

\bibitem[PSW25]{PSW25}
Francisco Pernice, Oscar Sprumont, and Mary Wootters.
\newblock List-decoding capacity implies capacity on the q-ary symmetric
  channel.
\newblock In {\em Proceedings of the 57th ACM Symposium on Theory of
  Computing}, 2025.

\bibitem[PV05]{PV05}
Farzad Parvaresh and Alexander Vardy.
\newblock Correcting errors beyond the {G}uruswami-{S}udan radius in polynomial
  time.
\newblock In {\em Proceedings of the 46th IEEE Symposium on Foundations of
  Computer Science}, pages 285--294, 2005.

\bibitem[Res20]{Res20}
Nicholas Resch.
\newblock List-decodable codes:(randomized) constructions and applications.
\newblock {\em PhD thesis, Carnegie Mellon University, 2020}, 2020.

\bibitem[RV25]{RV25}
Nicolas Resch and S.~Venkitesh.
\newblock List recoverable codes: The good, the bad, and the unknown (hopefully
  not ugly), 2025.

\bibitem[RW18]{RW18}
Atri Rudra and Mary Wootters.
\newblock Average-radius list-recoverability of random linear codes.
\newblock In {\em Proceedings of the 29th ACM-SIAM Symposium on Discrete
  Algorithms}, 2018.

\bibitem[Sha48]{S48}
Claude Shannon.
\newblock A mathematical theory of communications.
\newblock {\em Bell System Technical Journal}, 27:379--423, 623--656, 1948.

\bibitem[Sri25]{Sri25}
Shashank Srivastava.
\newblock Improved list size for folded reed-solomon codes.
\newblock In Yossi Azar and Debmalya Panigrahi, editors, {\em Proceedings of
  the 2025 Annual {ACM-SIAM} Symposium on Discrete Algorithms, {SODA} 2025, New
  Orleans, LA, USA, January 12-15, 2025}, pages 2040--2050. {SIAM}, 2025.

\bibitem[SS96]{SS96}
M.~Sipser and D.~Spielman.
\newblock Expander codes.
\newblock {\em IEEE Transactions on Information Theory}, 42(6):1710--1722,
  1996.
\newblock Preliminary version in {\em Proc. of FOCS'94}.

\bibitem[ST20]{ST20}
Chong Shangguan and Itzhak Tamo.
\newblock Combinatorial list-decoding of reed-solomon codes beyond the johnson
  radius.
\newblock In {\em Proceedings of the 52nd Annual ACM SIGACT Symposium on Theory
  of Computing}, pages 538--551, 2020.

\bibitem[ST25]{ST25}
Shashank Srivastava and Madhur Tulsiani.
\newblock List decoding expander-based codes up to capacity in near-linear
  time, 2025.

\bibitem[Sud97]{Sud97}
Madhu Sudan.
\newblock Decoding of reed solomon codes beyond the error-correction bound.
\newblock {\em J. Complex.}, 13(1):180--193, 1997.

\bibitem[Tam24]{Tam24}
Itzhak Tamo.
\newblock Tighter list-size bounds for list-decoding and recovery of folded
  reed-solomon and multiplicity codes.
\newblock {\em {IEEE} Trans. Inf. Theory}, 2024.

\bibitem[Tan81]{Tan81}
Robert~Michael Tanner.
\newblock A recursive approach to low complexity codes.
\newblock {\em {IEEE} Trans. Inf. Theory}, 27(5):533--547, 1981.

\bibitem[Tre99]{Tre99}
Luca Trevisan.
\newblock Construction of extractors using pseudo-random generators (extended
  abstract).
\newblock In Jeffrey~Scott Vitter, Lawrence~L. Larmore, and Frank~Thomson
  Leighton, editors, {\em Proceedings of the Thirty-First Annual {ACM}
  Symposium on Theory of Computing, May 1-4, 1999, Atlanta, Georgia, {USA}},
  pages 141--148. {ACM}, 1999.

\bibitem[TS17]{Ta-Shma17}
Amnon Ta-Shma.
\newblock Explicit, almost optimal, epsilon-balanced codes.
\newblock In {\em Proceedings of the 49th ACM Symposium on Theory of
  Computing}, STOC 2017, pages 238--251, New York, NY, USA, 2017. ACM.

\bibitem[Vad12]{Vadhan12}
Salil~P. Vadhan.
\newblock {\em Pseudorandomness}.
\newblock Now Publishers Inc., 2012.

\bibitem[Var57]{V57}
R.R. Varshamov.
\newblock Estimate of the number of signals in error correcting codes.
\newblock {\em Doklady Akademii Nauk SSSR}, 117:739--741, 1957.

\bibitem[YZ24]{YZ24}
Takashi Yamakawa and Mark Zhandry.
\newblock Verifiable quantum advantage without structure.
\newblock {\em J. ACM}, 2024.

\bibitem[ZP82]{ZP82}
Victor~V. Zyablov and Mark~S. Pinsker.
\newblock List cascade decoding.
\newblock {\em Problems of Information Transmission}, 1982.

\end{thebibliography}


\end{document}